\documentclass[manuscript]{acmart}
\acmJournal{TOCL}
\acmVolume{0}
\acmNumber{0}
\acmArticle{0}
\acmYear{0}
\acmMonth{0}
\acmArticleSeq{0}

\setcopyright{acmcopyright}
\copyrightyear{}
\acmYear{}
\acmDOI{}

\acmDOI{0000001.0000001}

\usepackage[utf8]{inputenc}
\usepackage[T1]{fontenc}
\usepackage{amsfonts}

\usepackage{amssymb}
\usepackage{amsthm}
\usepackage{amsmath}
\usepackage{stmaryrd}
\usepackage{wasysym}
\usepackage{graphicx}
\usepackage{rotating}
\usepackage{booktabs}
\usepackage{array}
\usepackage{multirow}
\usepackage{tikz}
\usepackage{color}
\usepackage{xcolor}
\usepackage{enumitem}
\usepackage{xspace}
\usepackage{thm-restate} 
\usepackage{microtype}
\usepackage{soul}
\usepackage{url}
\usepackage{hyperref}

\newcommand{\psif}{\psi_{f}}

\newcommand{\p}{\varphi}

\newcommand{\s}{\sigma}

\newcommand{\eset}{\emptyset}
\newcommand{\mdl}{\models}

\newcommand{\tr}{^{\dagger}}

\newcommand{\assign}{\ensuremath{\mathfrak{a}}\xspace}

\newcommand{\reduct}{\ensuremath{\Imf_{\mid_{\Sigma}}}\xspace}

\newcommand{\last}{\ensuremath{\textit{last}}\xspace}

\newcommand{\Pre}{\ensuremath{\textit{Pre}}\xspace}

\newcommand{\BadPre}{\ensuremath{\textit{BadPre}}\xspace}

\newcommand{\Ext}{\ensuremath{\textit{Ext}}\xspace}

\newcommand{\Clo}{\ensuremath{\textit{Clo}}\xspace}

\newcommand{\modelsinf}{\ensuremath{\models_{i}}\xspace}
\newcommand{\modelsfin}{\ensuremath{\models_{f}}\xspace}

\newcommand{\Traceinf}{\ensuremath{\textit{Trace}_{i}}\xspace}
\newcommand{\Tracefin}{\ensuremath{\textit{Trace}_{f}}\xspace}

\newcommand{\equivfin}[2]{\ensuremath{#1 \equiv_{f} #2}\xspace}
\newcommand{\equivinf}[2]{\ensuremath{#1  \equiv_{i} #2}\xspace}

\renewcommand{\U}{\ensuremath{\mathbin{\mathcal{U}}}\xspace}
\newcommand{\Until}{\ensuremath{\mathbin{\mathcal{U}}}\xspace}

\newcommand{\UntilP}{\ensuremath{\Until^{+}}\xspace\!}

\newcommand{\Release}{\ensuremath{\mathbin{\mathcal{R}}}\xspace}\newcommand{\ReleaseP}{\ensuremath{\Release^{+}}\xspace\!}

\newcommand{\Next}{{\ensuremath{\raisebox{0.25ex}{\text{\scriptsize{$\bigcirc$}}}}}}

\newcommand{\Wnext}{\CIRCLE}

\newcommand{\D}{\Diamond}

\newcommand{\DiamondP}{\ensuremath{\Diamond^{+}}\xspace}
\newcommand{\B}{\Box}

\newcommand{\monodic}{\ensuremath{\mathop{\ooalign{$\Box$ \cr \kern0.57ex \raisebox{0.2ex}{\scalebox{0.55}{$1$}}}\rule{0pt}{1.5ex} \kern-0.7ex}}\xspace}

\newcommand{\Nextone}{\ensuremath{\mathop{\ooalign{$\Next$ \cr \kern0.57ex \raisebox{0.3ex}{\scalebox{0.55}{$1$}}}\rule{0pt}{1.5ex} \kern-0.7ex}}\xspace}

\newcommand{\Wnextone}{\ensuremath{\mathop{\ooalign{$\Wnext$ \cr \kern0.57ex \raisebox{0.2ex}{\scalebox{0.55}{\textcolor{white}{$1$}}}}\rule{0pt}{1ex} \kern-0.7ex}}\xspace}

\newcommand{\EL}{\ensuremath{\mathcal{EL}}\xspace}

\newcommand{\DLite}{\textsl{DL-Lite}\xspace}

\newcommand{\ALC}{\ensuremath{\smash{\mathcal{ALC}}}\xspace}

\newcommand{\NC}{\ensuremath{{\sf N_C}}\xspace}
\newcommand{\NI}{\ensuremath{{\sf N_I}}\xspace}
\newcommand{\NR}{\ensuremath{{\sf N_R}}\xspace}
\newcommand{\NPr}{\ensuremath{\textsf{N}_{\textsf{P}}\xspace}}

\newcommand{\cl}{\ensuremath{\mathsf{cl}}\xspace}

\newcommand{\clc}{\ensuremath{\cl^{\mathsf{c}}}\xspace}
\newcommand{\clf}{\ensuremath{\cl^{\mathsf{f}}}\xspace}
\newcommand{\cls}{\ensuremath{\cl^\ast}\xspace}

\newcommand{\individuals}[1]{\NI(#1)}
\newcommand{\types}[1]{\mathsf{tp}(#1)}

\newcommand{\real}{\ensuremath{{\sf real}}\xspace}

\newcommand{\LTL}{\ensuremath{\textsl{L\!TL}}\xspace}
\newcommand{\LTLf}{\ensuremath{\LTL}\xspace}

\newcommand{\QTLfr}[3]{\ensuremath{{\textsl{T}_{#1}\mathcal{QL}^{#2}_{#3}}}\xspace}

\newcommand{\TALC}{\ensuremath{\smash{\textsl{T}_{\U}{\mathcal{ALC}}}}\xspace}

\newcommand{\TALCf}{\ensuremath{\smash{\textsl{T}_{\U}{\mathcal{ALC}}}}\xspace}

\newcommand{\TALCfk}{\ensuremath{\smash{\textsl{T}_{\U}{\mathcal{ALC}}}}\xspace}

\newcommand{\finit}[2]{\textnormal{\ensuremath{\smash{\textsf{F}_{#1 #2}}}\xspace}}
\newcommand{\infinit}[2]{\textnormal{\ensuremath{\smash{\textsf{I}_{#1 #2}}}\xspace}}

\newcommand{\NLogSpace}{\textsc{NLogSpace}}

\newcommand{\NP}{\textsc{NP}}

\newcommand{\ExpTime}{\textsc{ExpTime}}
\newcommand{\NExpTime}{\textsc{NExpTime}}
\newcommand{\ExpSpace}{\textsc{ExpSpace}}

\newcommand{\Cmf}{\ensuremath{\mathfrak{C}}\xspace}

\newcommand{\Emf}{\ensuremath{\mathfrak{E}}\xspace}
\newcommand{\Fmf}{\ensuremath{\mathfrak{F}}\xspace}

\newcommand{\Imf}{\ensuremath{\mathfrak{I}}\xspace}

\newcommand{\Mmf}{\ensuremath{\mathfrak{M}}\xspace}

\newcommand{\Rmf}{\ensuremath{\mathfrak{R}}\xspace}

\newcommand{\Tmf}{\ensuremath{\mathfrak{T}}\xspace}

\newcommand{\FEmf}{\ensuremath{\Fmf \cdot_{E} \Emf}\xspace}
\newcommand{\FEmfC}{\ensuremath{\Fmf \cdot_{E} \Emf}\xspace}
\newcommand{\FImf}{\ensuremath{\Fmf \cdot \Imf}\xspace}
\newcommand{\FEImf}{\ensuremath{\Fmf \cdot_{E} \Imf}\xspace}

\newcommand{\Cmc}{\ensuremath{\mathcal{C}}\xspace}
\newcommand{\Dmc}{\ensuremath{\mathcal{D}}\xspace}
\newcommand{\Emc}{\ensuremath{\mathcal{E}}\xspace}
\newcommand{\Fmc}{\ensuremath{\mathcal{F}}\xspace}

\newcommand{\Imc}{\ensuremath{\mathcal{I}}\xspace}

\newcommand{\Lmc}{\ensuremath{\mathcal{L}}\xspace}
\newcommand{\Mmc}{\ensuremath{\mathcal{M}}\xspace}

\newcommand{\Smc}{\ensuremath{\mathcal{S}}\xspace}
\newcommand{\Tmc}{\ensuremath{\mathcal{T}}\xspace}
\newcommand{\Umc}{\ensuremath{\mathcal{U}}\xspace}

\newcommand{\FEmc}{\ensuremath{\Fmc \cdot_{E} \Emc}}
\newcommand{\FImc}{\ensuremath{\Fmc \cdot \Imc}}
\newcommand{\FEImc}{\ensuremath{\Fmc \cdot_{E} \Imc}\xspace}

\newcommand{\Nbl}{\ensuremath{\mathbb{N}}\xspace}

\newcommand{\Nbb}{\ensuremath{\mathbb{N}}\xspace}

\newcommand{\I}{\ensuremath{\mathcal{I}}}
\newcommand{\T}{\ensuremath{\mathcal{T}}}
\begin{document}

\title[First-order Temporal Logic on Finite Traces]{First-order Temporal Logic on Finite Traces}
\subtitle{Semantic Properties, Decidable Fragments, and Applications}

\newcommand{\afilTheUnibz}{
\affiliation{%
  \institution{KRDB Research Centre, Faculty of Computer Science, Free University of Bozen-Bolzano}
  \city{Bolzano}
  \country{Italy}}
}

\newcommand{\afilUnib}{
\affiliation{%
  \institution{Department of Informatics, University of Bergen}
  \city{Bergen}
  \country{Norway}}
}

\author{Alessandro Artale}
\orcid{0000-0002-3852-9351}
\afilTheUnibz
\email{artale@inf.unibz.it}

\author{Andrea Mazzullo}
\orcid{0000-0001-8512-1933}
\afilTheUnibz
\email{mazzullo@inf.unibz.it}

\author{Ana Ozaki}
\orcid{0000-0002-3889-6207}
\afilUnib
\email{ana.ozaki@uib.no}
 
\pdfinfo{
   /Author (Alessandro Artale and Andrea Mazzullo and Ana Ozaki)
   /Title  (First-order Temporal Logic on Finite Traces)
   /Keywords (Temporal Logics;Automated Reasoning)
}

%
%

%
%


\renewcommand{\shortauthors}{A. Artale, A. Mazzullo, A. Ozaki}

\newtheorem{claim}{Claim} 

\begin{abstract}

Formalisms based on temporal logics interpreted over finite strict linear orders, known in the literature as \emph{finite traces}, have been used for temporal specification in automated planning, process modelling, (runtime) verification and synthesis of programs, as well as in knowledge representation and reasoning.
In this paper, we focus on \emph{first-order temporal logic on finite traces}.
We first investigate preservation of equivalences and satisfiability of formulas between finite and infinite traces, by providing  a set of semantic and syntactic conditions to guarantee when the distinction between reasoning in the two
cases
can be blurred.
Moreover,
we show that
the
satisfiability problem
on finite traces
for
several decidable fragments
of first-order temporal logic
is
$\ExpSpace$-complete, as in the infinite trace case, while it decreases to $\NExpTime$ when
finite traces bounded in the number of instants
are considered.
This leads also to new complexity results for temporal description logics over finite traces. Finally, we
investigate
applications
to planning and verification, in particular by establishing connections with the notions of 
insensitivity to infiniteness and safety from the literature.


\end{abstract}

\maketitle

\section{Introduction}
\label{sec:intro}


The study of formalisms based on propositional or first-order temporal
logics on linear flows of time has found a wide spectrum of
applications, ranging from verification of programs and model checking~\cite{Pnu,ManPnu,BaiKat},
to automated planning~\cite{BacKab98,BacKab,BaiMci}, process
modelling~\cite{AalPes,Maggi.et.al:BPM11}, and knowledge
representation.
In the latter context, several decidable fragments of first-order
temporal logic with the linear time operator \emph{until} $\Until$,
denoted $\QTLfr{\Until}{}{}$, have been
investigated~\cite{HodEtAl2,HodEtAl03,GabEtAl}.
\emph{Temporal description
  logics}~(see~\cite{WolZak,Baader:et:al:03,ArtFra2,AKKRWZ:TIME17,LutEtAl}
and references therein), obtained by suitably combining (linear)
temporal logic operators with \emph{description logics} (\emph{DLs})
constructs, are well-known examples of such fragments.
These logics usually lie within the \emph{two-variable monodic}
fragment of $\QTLfr{\Until}{}{}$, denoted as
$\QTLfr{\Until}{2}{\monodic}$, obtained by restricting the language to
formulas having at most two variables, and so that the temporal
operators are applied only to subformulas with at most one free
variable.
%
For instance, by
using the reflexive temporal operators $\Diamond^{+}$, meaning \emph{sometimes
  in the future}, and $\Box^{+}$, meaning \emph{always in the future},
the following formula
%
\begin{equation}
\label{eq:ex1}
\small
\forall x \big( \mathit{Reviewer}(x) \to \Diamond^{+} \Box^{+} \forall y \big( \mathit{Submission}(y) \land \mathit{reviews}(x, y) \to \Diamond^{+}
\mathit{Evaluated}(y)\big) \big)
\end{equation}
is a $\QTLfr{\Until}{2}{\monodic}$ formula stating that every reviewer
will reach a (present or future) moment after which all the
submissions they review will be eventually evaluated.
%
Other decidable languages considered
in the literature are the \emph{monadic} fragment
$\QTLfr{\Until}{mo}{}$ and the \emph{one-variable} fragment
$\QTLfr{\Until}{1}{}$, with formulas having, respectively, at most
unary predicates and at most one variable.
%
The complexity of the satisfiability problem ranges from \ExpSpace-complete, for $\QTLfr{\Until}{2}{\monodic}$, $\QTLfr{\Until}{mo}{}$ and $\QTLfr{\Until}{1}{}$~\cite{GabEtAl,HodEtAl03},
down to $\NExpTime$- or $\ExpTime$-complete, for temporal extensions of the DL $\ALC$ without temporalised roles and with restrictions on
the application of
temporal operators~\cite{LutEtAl,BaaEtAl2}.
Recent
work on temporal extensions of lightweight DLs in the $\DLite$ and
$\EL$ families~\cite{ArtEtAl3, guti:et:al:ijcai15}, used in conceptual
data modelling and underlying prominent profiles of the OWL standard,
shows that the complexity of reasoning can be even lowered down to
$\NP$ or $\NLogSpace$.

A widely studied semantics for temporal logics is defined on
structures based on the strict linear order of the natural
numbers~\cite{Pnu,Gol87,GabEtAl94}.  However, linear temporal
structures with only a finite number of time points, often called
\emph{finite traces}, have been investigated as
well~\cite{GabEtAl94,CerEtAl}, receiving a renewed interest in the
literature~\cite{DegVar1,FiondaGreco:aaai16,FioGre18}.  The finiteness
of the time dimension represents indeed a fairly natural restriction
for several applications.  In automated planning, or when modelling
(business) processes with a declarative formalism, we consider finite
action plans and terminating services, often within a given temporal
bound~\cite{BauHas,DegEtAl2,DegEtAl,CamEtAl}.  In
runtime verification only the current finite behaviour of the system
is taken into account, while infinite models are considered when checking
whether a given requirement is satisfied in some/all infinite
extensions of the finite trace~\cite{GiaHav,BauEtAl}.
These
needs from critical applications of
temporal logics can be reflected by a semantics based on finite
traces, with formulas having different satisfaction conditions
compared to the infinite case.
For instance, by using the formula $\last$ to refer to the last time
point of a finite trace, we have that Formula~(\ref{eq:ex1}) above is
equivalent
on finite traces to
\begin{equation}
\label{eq:ex2}
\small
	\forall x \big( \mathit{Reviewer}(x) \to \Diamond^{+} \big( \last \land \forall y \big( \mathit{Submission}(y) \land
 \mathit{reviews}(x,y)
	\to
	\mathit{Evaluated}(y)
	\big) \big) \big),
\end{equation}
stating that every reviewer will eventually reach a ``deadline'', represented by the $\last$ formula, when all the submission they review are
evaluated.
This follows from the fact that, on finite traces, formulas $\Diamond^{+}\Box^{+}\p$ and $\Diamond^{+} (\last \land \p)$ are equivalent~\cite{DegVar1}, and $\last \land \forall x ( \p \to \Diamond^{+} \psi)$ is in turn equivalent to $\last \land \forall x ( \p \to \psi)$.
%
%

This work focuses on first-order temporal logic on finite traces,
defined by extending the first-order language with linear time
operators interpreted on finite traces.
Part of the results contained in the current paper, establishing
bridges between finite and infinite traces semantics, have partially
been presented in~\cite{ArtEtAl18,ArtEtAl19,ArtEtAl20}.
We
comment on the main results presented in this paper while illustrating
its
structure.
%
Section~\ref{sec:relwork} is
devoted to a discussion of related work in this area.  Then, after
introducing in Section~\ref{sec:tdlf} the required preliminary notions
about first-order (linear) temporal logic, we provide the following
main contributions.
%
In Section~\ref{sec:dist}, we focus on bridging reasoning on finite
and infinite traces semantics.  In general, indeed, the sets of
formulas equivalent on finite and on infinite traces do not coincide,
as witnessed by the following examples: the formula
$\Diamond^{+} \Box \bot$, forcing the existence of a (present or
future) last instant of time, is equivalent to $\top$ only on finite
traces, whereas $\Box^{+} \Diamond \top$, stating that there is always
a future instant after the current one, is equivalent to $\top$ on
infinite, but not on finite, traces.  We establish here semantic and
syntactic conditions to guarantee that
two formulas equivalent on finite (respectively, infinite) traces are
also equivalent on infinite (respectively, finite) traces.
We also devise a semantic criterion and a syntactically defined class
of formulas to guarantee preservation of satisfiability from finite to
infinite traces.
%
%
In Section~\ref{sec:tdlf-to-tdl}, we
study the complexity of reasoning over finite traces for the
\emph{two-variable monodic} fragment $\QTLfr{\Until}{2}{\monodic}$,
the \emph{monadic} fragment $\QTLfr{\Until}{mo}{}$ and the
\emph{one-variable} fragment $\QTLfr{\Until}{1}{}$, showing that the
complexity
remains $\ExpSpace$-complete on finite traces, but lowers down to
$\NExpTime$ if we restrict to traces with a bound $k$ (given in
binary) on the number of instants.  Moreover, we show that these
fragments enjoy two kinds of \emph{bounded model properties}:
the \emph{bounded trace property}, which limits the finite traces
satisfying a formula to be at most double exponential in the size of
the input formula, and the \emph{bounded domain property}, which
limits the number of domain elements in a $k$-bounded trace (finite
trace with at most $k$ instants) to be at most double exponential both
in the size of the input formula and in the (binary representation) of
$k$.
%
%
In
the context of temporal DLs, we show
that complexity results similar to the infinite trace case, as well as
bounded model properties, also hold for the temporal DL $\TALCf$ when
interpreted on finite traces.  We further show a more challenging
result, i.e., that the complexity further reduces to $\ExpTime$ if
only global TBox $\TALCf$ axioms together with a temporal ABox are
interpreted on $k$-bounded traces.
Finally, in Section~\ref{sec:appl}, we investigate connections with
the planning and verification literature.
Concerning 
the former scenario, we study in our setting the notion of
\emph{insensitivity to infiniteness}~\cite{DegEtAl}, a property
applying to formulas that, once satisfiable on finite traces, remain
satisfiable on infinite traces verifying an \emph{end event} forever
and falsifying all other atomic formulas.  
%
Concerning the verification aspect, we establish connections between
the finite and infinite trace characterisations, as introduced in
Section~\ref{sec:dist}, and the notions of
\emph{safety}~\cite{Sis,BaiKat}, as well as other related notions from
the literature on runtime verification~\cite{BauEtAl}.
Section~\ref{sec:conc} concludes the paper.


\section{Related Work}
\label{sec:relwork}

Finite traces~\cite{GabEtAl94,CerEtAl,Ros18} have
regained momentum
in
formalisms for AI applications.
%
%
Together with (propositional) \emph{linear temporal logic} (\emph{LTL})~\cite{Pnu}, also the more expressive 
\emph{linear dynamic logic}~\cite{HenThi99}, \emph{alternating time logic}~\cite{AluEtAl02}, and \emph{mu-calculus}~\cite{Koz83,BarEtAl86} have been
investigated on semantics based on finite traces~\cite{GioEtAl01,DegEtAl20,BelEtAl18,BelEtAl19,LiuEtAl16}.
To deal with uncertainty in dynamic systems, a \emph{probabilistic} version of LTL over finite traces has been proposed as well~\cite{MagEtAl20}, while a recent paper
addresses problems in declarative process mining by introducing
\emph{metric temporal logic} on finite traces~\cite{DegEtAl21}.
Significant areas of applications for LTL
on finite traces are indeed 
in
the planning domain~\cite{CerMay,DegEtAl,CamEtAl,CerMay98,BaiMci,GerEtAl09,CalEtAl02},
in (declarative) business process modelling, as well as in runtime verification and monitoring~\cite{BauEtAl,Ros18,BarEtAl18,DegEtAl2,DegEtAl21a}.
In addition, LTL
on finite traces
has
found applications in the context of
synthesis~\cite{DegVar2,XiaoEtAl21,DegEtAl20a,DegEtAl21a,CamEtAl18,CamEtAl18a},
multi-agent systems~\cite{GutEtAl17,KonLom18,GutEtAl21},
temporal databases~\cite{SaaLip88},
and
answer-set programming~\cite{CabEtAl18,CabSch19,CabEtAl19}.
%
The problem of establishing connections between finite and infinite traces semantics
is also not new to the literature.
Several approaches have been proposed to show when satisfiability of formulas is preserved from the finite to the infinite case, so to reuse on finite traces algorithms developed for the infinite case~\cite{BauHas,DegEtAl}.
In this work, we determine conditions that preserve satisfiability in the other direction as well, from infinite to finite traces, thus leaning towards research directions
that aim at
the application of efficient finite traces reasoners to the infinite case~\cite{LiEtAl,FioGre18,ShiEtAl20}.

Given their connections with first-order temporal logic and their
relevance to the present article, we separately discuss related work
on temporal DLs.  For a general overview, we refer to the already
mentioned
surveys~\cite{WolZak,Baader:et:al:03,ArtFra2,AKKRWZ:TIME17,LutEtAl}.
%
In the linear time case, a wide body of research has focused on temporal DLs with semantics based 
on the natural numbers~\cite{BaaEtAl2,DBLP:conf/dlog/OzakiKR18,DBLP:conf/birthday/OzakiKR19} or the integers~\cite{ArtEtAl3}, 
possibly by extending the language with metric temporal operators as well~\cite{DBLP:conf/ecai/Gutierrez-Basulto16,DBLP:journals/tocl/BaaderBKOT20,DBLP:conf/frocos/BaaderBKOT17}.
In applications, temporalised DLs have been considered in the context of runtime verification~\cite{BaaEtAl3,BaaLip}
and business process modelling~\cite{AalEtAl,ArtEtAl19c}.
However,
such proposals
are based on the usual infinite trace semantics
or
are
limited in expressivity.
To the best of our knowledge, little has been done in order to combine finite traces and temporal DLs.
Recent work in this direction
  can be found in~\cite{ArtEtAl18,ArtEtAl19}.
The complexity landscape of temporal DLs on finite traces semantics has been further enriched by preliminary results on temporal $\DLite$ logics~\cite{ArtEtAl19b}.
These results, that match the corresponding ones on infinite traces semantics~\cite{ArtEtAl3}, are obtained by considering axioms interpreted either globally or locally, and by syntactically restricting the application of the temporal operators (allowing for $\Until$, or only for $\Box$ and $\Next$), or of the DL constructors (focussing on the so-called \emph{bool}, \emph{horn}, \emph{krom}, and \emph{core} fragments).

\section{First-order Temporal Logics}
\label{sec:tdlf}


The \emph{first-order}
\emph{temporal language} $\QTLfr{\Until}{}{}$~\cite{GabEtAl},
that we present in the following,
is obtained by
extending the usual first-order language with the temporal
operator \emph{until} $\Until$ interpreted over
linear
structures, called \emph{traces}.

\subsection{Syntax}
\label{sec:syntax}
The \emph{alphabet} of $\QTLfr{\Until}{}{}$ consists of countably infinite
and pairwise disjoint sets of \emph{predicates}
$\textsf{N}_{\textsf{P}}$
(with $\textsf{ar}(P) \in \mathbb{N}$ being the \emph{arity} of
$P \in \textsf{N}_{\textsf{P}}$), \emph{constants} (or
\emph{individual names}) $\NI$, and \emph{variables} $\mathsf{Var}$;
the \emph{logical operators} $\lnot$ (\emph{negation}) and $\land$ (\emph{conjunction}); the \emph{existential quantifier} $\exists$, and the \emph{temporal operator} $\Until$ (\emph{until}).  The \emph{formulas} of $\QTLfr{\Until}{}{}$ are of the
form:
\[
\p
::= 
P(\bar{\tau})
\mid
\neg \p \mid
(\p \land \p)
\mid \exists x \p \mid
(\p \U \p),
\]
where $P \in \textsf{N}_{\textsf{P}}$, 
$\bar{\tau} = (\tau_{1}, \ldots, \tau_{\textsf{ar}(P)})$ is a tuple of
\emph{terms}, i.e., constants or variables,
and
$x \in \textsf{Var}$. 
%
Formulas without the until operator are called \emph{non-temporal}.
We write $\p(x_{1}, \ldots, x_{m})$ to indicate that the free
variables of a formula $\p$ are exactly $x_{1}, \ldots, x_{m}$.
We write
$\p\{x/y\}$
for
the result of uniformly substituting the free occurrences of $y$ in 
$\p$ by $x$. 
%
For $p \in \Nbl$, the \emph{$p$-variable fragment} of $\QTLfr{\Until}{}{}$,
denoted by $\QTLfr{\Until}{p}{}$, consists of $\QTLfr{\Until}{}{}$ formulas with
at most $p$ variables.
The (propositional) language
$\LTL$ is
$\QTLfr{\Until}{0}{}$ with formulas constructed without the existential quantifier.
The \emph{monodic fragment} of $\QTLfr{\Until}{}{}$,
denoted by $\QTLfr{\Until}{}{\monodic}$, consists of formulas such that all
subformulas of the form $\p \Until \psi$ have \emph{at most one}
free variable.
The \emph{monadic fragment} $\QTLfr{\Until}{mo}{}$ is the fragment of $\QTLfr{\Until}{}{}$ with formulas
not containing predicates of arity
$\mathop{\geq} 2$.
Finally,
the \emph{constant-free one-variable
monadic
fragment}, 
$\QTLfr{\Until}{1,mo}{\not c}$,
is obtained from the one-variable
monadic
fragment by disallowing constants.


\subsection{Semantics}
\label{sec:semantics}

A \emph{first-order temporal interpretation} (or \emph{trace}) is a
pair
$\Mmf=(\Delta^{\Mmf}, (\Mmc_{n})_{n \in \Tmf})$, where $\Tmf$ is a
sub-order of $(\mathbb{N}, <)$ of the form $[0,\infty)$ or $[0,l]$,
with $l \in \mathbb{N}$,
and each $\Mmc_{n}$ is a classical first-order interpretation with
a non empty domain $\Delta^\Mmf$ (or simply $\Delta$): we have
$P^{\Mmc_{n}} \subseteq \Delta^{\textsf{ar}(P)}$, for each
$P \in \textsf{N}_{\textsf{P}}$, and
$a^{\Mmc_{i}}=a^{\Mmc_{j}}\in \Delta$ for all $a\in\NI$ and
$i,j \in \Nbb$, i.e., constants are \emph{rigid designators} (with
fixed interpretation, denoted simply by $a^{\Mmc}$).
The stipulation that all time points share the same domain $\Delta$ is
called the \emph{constant domain assumption} (meaning that objects are
not created or destroyed over time), and it is the most general choice
in the sense that increasing, decreasing, and varying domains can all
be reduced to it~\cite{GabEtAl}.
An \emph{assignment in $\Mmf$} 
(or simply an \emph{assignment}, when $\Mmf$ is clear from the
context) is a function $\assign$ from $\mathsf{Var}$ to $\Delta$, and
the \emph{value of a term $\tau$ in $\Mmf$ under $\assign$} is defined
as: $\assign(\tau) = \assign(x)$, if $\tau = x$, and
$\assign(\tau) = a^{\Mmc}$, if $\tau = a \in \NI$. Given a tuple of
$m$ terms $\bar{\tau} = (\tau_{1}, \ldots, \tau_{m})$, we set
$\assign(\bar{\tau}) = (\assign(\tau_{1}), \ldots,
\assign(\tau_{m}))$.  Given a formula~$\p$, the \emph{satisfaction of
  $\p$ in~$\Mmf$ at time point $n \in \Tmf$ under an assignment
  $\assign$}, written $\Mmf, n \models^{\assign} \p$, is inductively
defined as:
\[
\begin{array}{lcl}
		\Mmf, n \models^{\assign} 
		P(\bar{\tau}) 
		& \text{iff} & 
		\assign(\bar{\tau}) \in P^{\Mmc_{n}},\\
		\Mmf, n  \models^{\assign} \neg \psi & \text{iff} & \text{not } \Mmf, n  \models^{\assign} \psi, \\
		\Mmf, n  \models^{\assign} \psi \land \chi & \text{iff} & \Mmf, n  \models^{\assign} \psi \text{ and } \Mmf, n  \models^{\assign} \chi, \\
		\Mmf, n  \models^{\assign} \exists x \psi & \text{iff}
                             & \Mmf, n \models^{\assign'} \psi, \text{
                               for some assignment } \assign' \text{ that can differ from } \assign \text{ only on } x, \\
		\Mmf, n  \models^{\assign} \psi \Until \chi &
		                                                  		  \text{iff} & \text{there is }  m \in \Tmf, m > n  \colon \Mmf, m \models^{\assign}\chi \text{ and, } \text{for all } i \in (n,m), \Mmf, i \models^{\assign} \psi. \\
 \end{array}
\]
%
We say that $\varphi$ is \emph{satisfied in $\Mmf$}
\emph{under $\assign$}, writing $\Mmf \models^{\assign} \varphi$, if
$\Mmf, 0 \models^{\assign} \varphi$, and that
$\varphi$ is \emph{satisfied in $\Mmf$} (or that $\Mmf$ is a
\emph{model} of $\varphi$), denoted by $\Mmf \models \varphi$, if
$\Mmf \models^{\assign} \varphi$, for some $\assign$.
%
Moreover, $\p$ is said to be \emph{satisfiable} if it is satisfied in
some $\Mmf$.
A formula $\p$ \emph{logically implies} a formula $\psi$
if, for every interpretation $\Mmf$ and every assignment $\assign$,
$\Mmf \models^{\assign} \p$ implies $\Mmf \models^{\assign} \psi$,
and we
write $\p \mdl \psi$.
We say that $\p$ and $\psi$ are
\emph{equivalent}, writing $\p \equiv \psi$, if $\p \mdl \psi$ and
$\psi \mdl \p$.
Since the satisfaction of a formula $\p(x_1, \ldots, x_n )$ under an assignment $\assign$ depends only on the values of its free variables under $\assign$, we may write
$\Mmf, n \models \p[d_1, \ldots, d_n]$ in place of
$\Mmf, n \models^{\assign} \varphi( x_1, \ldots, x_n )$, where $\assign(x_1) = d_1, \ldots, \assign(x_n) = d_n$.
Also, given an assignment $\assign$ and an element $d$ in an
interpretation's domain $\Delta$
we denote by $\assign[x \mapsto d]$ the assignment obtained by
modifying $\assign$ so that $x$ maps to $d$.

In the following, we call \emph{finite trace} a trace
with $\Tmf =[0,l]$, often denoted by
$\Fmf = (\Delta^\Fmf, (\Fmc_n)_{n \in [0, l]})$, while \emph{infinite
  traces}, based on $\Tmf =[0,\infty)$, will be denoted by
$\Imf = (\Delta^\Imf, (\Imc_n)_{n \in [0, \infty)})$.
We say that a  $\QTLfr{\Until}{}{}$ formula $\p$ is \emph{satisfiable}
\emph{on infinite}, \emph{finite}, or \emph{$k$-bounded traces}, respectively, if it is satisfied in a trace in
the class of infinite, finite, or finite traces with at most $k \in \mathbb{N}, k > 0$ (given in binary) time points, respectively.
Moreover, given $\QTLfr{\Until}{}{}$ formulas $\p$ and $\psi$, we write $\p \modelsinf \psi$ (respectively, $\p \modelsfin \psi$) iff $\p$ logically implies $\psi$ on infinite (resp., finite) traces.
Similarly, we write
$\equivinf{\p}{\psi}$
if $\p$ and $\psi$ are
equivalent
on infinite traces, and
$\equivfin{\p}{\psi}$
if they are
equivalent on finite traces.
  In addition to the standard
  conventions on parenthesis and Boolean equivalences, we will use the
  following abbreviations for formulas:
  $\bot := \varphi \land \lnot \varphi$ (\emph{bottom});
  $\top := \lnot \bot$ (\emph{top});
  $\p \UntilP \psi := \psi \lor (\p \land \p \Until \psi)$;
$\p \Release \psi := \lnot ( \lnot \p \Until \lnot \psi)$
(\emph{releases});
$\p \ReleaseP \psi := \psi \land (\p \lor \p \Release \psi)$;
$\Diamond \p := \top \U \p$
(\emph{diamond});
$\DiamondP \p := \top \UntilP \p$;
$\B \p := \bot \Release \p$
(\emph{box});
$\B^{+} \p := \bot \ReleaseP \p$;
$\Next \p := \bot \Until \p$
(\emph{strong next});
and
$\Wnext \p := \top \Release \p$
(\emph{weak next}).
To refer to the last time point in a finite trace, we
use the formula 
$\textit{last}
:= \B \bot $, 
which is indeed satisfied at an instant of a trace iff that instant
does not have a successor, 
i.e.,
it holds that $\last \equiv \lnot \Next
\top$.  Observe that $\last \equiv_{i} \Next \bot \equiv_{i}
\bot$, whereas $\last \not \equiv_{f} \bot$.  Moreover,
we have that $\Wnext \p \equiv_{i} \Next \p \equiv_{i} \lnot \Next
\lnot \p$, while
none of the previous equivalences hold on finite traces.
Indeed, we have that $\Wnext \p \equiv \last \lor \Next
\p$, meaning that $\Wnext
\p$ is satisfied at a time point of a trace iff either it is the last instant
of the trace, or at the next time point $\p$ holds.

We now introduce the notation used in the
rest of the paper.
Given a trace $\Mmf= (\Delta^{\Mmf}, (\Mmc_{n})_{n \in \Tmf})$, the \emph{suffix of} $\Mmf$ \emph{starting at} $i \in
\Tmf$ is the trace $\Mmf^{i} = (\Delta^{\Mmf}, (\Mmc^{i}_{n})_{n \in
  \Tmf'})$, where: if $\Tmf = [0,l]$, we set $\Tmf' = [0, l - i]$ and $\Mmc^{i}_{n} = \Mmc_{i + n}$, for every $n \in \Tmf'$;
whereas, if $\Tmf = [0, \infty)$, we set $\Tmf' = [0,
\infty)$ and $\Mmc^{i}_{n} = \Mmc_{i + n}$, for every $n \in \Tmf'$.
The \emph{prefix of } $\Mmf$ \emph{ending at} $i \in
\Tmf$ is the trace $\Mmf_{i} = (\Delta^{\Mmf}, (\Mmc_{n})_{n \in
  \Tmf'})$, where $\Tmf' = [0, i]$.
%
Clearly, since
$\QTLfr{\Until}{}{}$ does not contain past temporal operators, given a
$\QTLfr{\Until}{}{}$ formula $\p$, a trace $\Mmf =(\Delta^{\Mmf},
(\Mmc_{n})_{n \in \Tmf})$, an assignment $\assign$ in
$\Mmf$, and an instant $n \in \Tmf$, we have that $\Mmf, n
\models^{\assign} \p$ iff $\Mmf^{n} \models^{\assign} \p$.
%

Let $\Fmf = (\Delta^\Fmf, (\Fmc_n)_{n \in [0, l]})$ and
$\Mmf = (\Delta^\Mmf, (\Mmc_n)_{n \in \Tmf})$ be, respectively,
a finite trace and a (finite or infinite) trace such that
$\Delta^{\Fmf} = \Delta^{\Mmf}$ (writing $\Delta$) and
$a^{\Fmc} = a^{\Mmc}$, for all $a \in \NI$.  We denote by
$\Fmf\cdot\Mmf = (\Delta^{\Fmf\cdot\Mmf}, (\Fmc\cdot\Mmc_{n})_{n \in \Tmf'})$ the
\emph{concatenation of $\Fmf$ with $\Mmf$}, defined as the trace with:
$\Delta^{\Fmf\cdot\Mmf} = \Delta$; $a^{\Fmc\cdot\Mmc} = a^{\Fmc}$,
for all $a \in \NI$; $\Tmf' = [0, \infty)$, if $\Tmf = [0,\infty)$, and $\Tmf' = [0, (l + l') +1]$, if $\Tmf = [0, l']$; and for $P \in \textsf{N}_{\textsf{P}}$,
$n\in \Tmf'$:
\begin{align*}
&P^{\Fmc\cdot\Mmc_{n}} =
			 \begin{cases}
			  P^{\Fmc_{n}}, & \text{if $n \in [0, l]$} \\
			  P^{\Mmc_{n - (l + 1)}}, & \text{otherwise.} \\
			 \end{cases}
\end{align*}
%
%
We define the set of \emph{extensions} of
a finite trace
$\Fmf$ as the set
of infinite traces
$\Ext(\Fmf) = \{ \Imf \mid \Imf = \Fmf \cdot \Imf',
\text{ for some infinite trace } \Imf' \}$.
Instead, given
a trace
$\Mmf$,
the set of \emph{prefixes} of
$\Mmf$
is the set
$\Pre(\Mmf) =
\{ \Fmf \mid \Mmf = \Fmf \cdot \Mmf', \text{ for some trace } \Mmf'
\}$.
Moreover, we call the \emph{frozen extension} of $\Fmf = (\Delta, (\Fmc_n)_{n \in
  [0, l]})$, denoted by
$\Fmf^{\omega}$, the
concatenation
of \Fmf with the infinite trace
$\Imf = (\Delta, (\Imc_{n})_{n \in [0, \infty)})$ such that
$\Imc_{n} = \Fmc_{l}$,
for every $n \in [0, \infty)$.
That is, $\Fmf^{\omega}$ is the infinite trace obtained from $\Fmf$ by repeating its last time point infinitely
often~\cite{BauHas}.
\section{Finite vs. Infinite Traces}
\label{sec:dist}

In this section, we compare finite and infinite traces semantics.
First, in Section~\ref{sec:redinf}, we lift to the first-order
temporal logic setting a well-known reduction of propositional linear
temporal logic formula satisfiability from finite traces to infinite
ones.  Then, in Section~\ref{sec:fininf}, we establish model-theoretic
conditions under which it is guaranteed that formulas equivalent on
finite (respectively, infinite) traces are also equivalent on infinite
(respectively, finite) traces.
In addition, we syntactically define classes of formulas that are
shown to satisfy such model-theoretic conditions, and for which the
corresponding results on preservation of formula equivalences are thus
inherited.
We finally restrict ourselves to the problem of preserving
satisfiability of a formula, from finite to infinite traces.  For this
case as well, we define a class of formulas for which it holds that
satisfiability on finite traces implies satisfiability on infinite
traces.

\subsection{Reduction to Satisfiability on Infinite Traces}
\label{sec:redinf}

In the following, we show how to reduce the formula satisfiability
problem on finite traces to the same problem on infinite traces. 
Similar to the encoding proposed in~\cite{DegVar1} for (propositional)
$\LTL$, to capture the finiteness of the temporal dimension, we
introduce a fresh unary predicate $E$, standing for the \emph{end of
  time}, with the following properties:
%
            $(i)$ there is at least one instant before the end of time;
            $(ii)$ the end of time comes for all objects;
            $(iii)$ the end of time comes at the same time for every object;
            $(iv)$ the end of time is permanent.
%
We axiomatise these properties as follows: 
%
\begin{align*}
    \psi_{f}^{1}  & = \forall x \lnot E(x) & \text{(Point~($i$)),}  \\
    \psi_{f}^{2}  & = \forall x \lnot E(x)  \Until \forall x E(x) & \text{(Points~$(ii)$, $(iii)$),}  \\
    \psi_{f}^{3}  & = \B \forall x (E(x) \to \Next E(x)) & \text{(Point~$(iv)$)}.
\end{align*}
%

We now characterise models satisfying the \emph{end of time formula} 
$\psi_{f} = \psi_{f}^{1}\land \psi_{f}^{2} \land \psi_{f}^{3}$.
Let
$\Fmf = (\Delta, (\Fmc_n)_{n \in [0, l]})$
and
$\Imf = (\Delta, (\Imc_n)_{n \in [0, \infty)})$
be, respectively, a finite and an infinite trace with the same domain $\Delta$ and such that
$a^{\Fmc} = a^{\Imc}$,
for all $a \in \NI$.
We denote by 
$\FEImf$
the
\emph{end extension of $\Fmf$ with $\Imf$},
defined as the
concatenation
of $\Fmf$ with $\Imf$
such that:
\begin{align*}
&E^{\FEImc_{n}} =
			 \begin{cases}
			  \emptyset, & \text{if $n \in [0, l]$}; \\
			   \Delta, & \text{if $n \in [l + 1, \infty)$}. \\
			 \end{cases}
\end{align*}
%
\noindent
Clearly, end extensions
characterise the satisfiability of
$\psi_{f}$. We formalise this in the next lemma.

\begin{restatable}{lemma}{LemmaModels}\label{lemma:finchar}
  For every infinite trace $\Imf$, $\Imf \mdl \psi_{f}$ iff
  $\Imf = \FEImf'$, for some finite trace $\Fmf$ and some infinite
  trace $\Imf'$.
\end{restatable}
\begin{proof}
  $\psi_{f}$ is satisfied in $\Imf$ iff there is $k>0$ such that, for
  all $d\in\Delta$, it holds that: $d\not\in E^{\Imc_{j}}$, for all
  $j \in[0, k)$; and $d\in E^{\Imc_{i}}$, for all $i \in[k, \infty)$.
  That is, $\Imf = \FEImf'$, for some finite trace $\Fmf$ and some
  infinite trace $\Imf'$.
\end{proof}

We now introduce a translation $\cdot^{\dagger}$
for $\QTLfr{\Until}{}{}$ formulas,
used together with the end of time
formula, $\psi_f$, to capture
satisfiability on finite traces.
More formally, a $\QTLfr{\Until}{}{}$ formula $\p$ is satisfiable
on finite traces
if and only if its translation $\p^{\dagger}$ is
satisfied in an infinite trace
that also satisfies the formula $\psi_f$.
The translation $\cdot^{\dagger}$  is
defined as: 
\begin{align*}
	(P(\bar{\tau}))^{\dagger} & = P(\bar{\tau}), \\
 	(\lnot \psi)^{\dagger} & = \lnot \psi^{\dagger}, \\
 	(\psi  \land \chi)^{\dagger} & = \psi^{\dagger} \land \chi^{\dagger}, \\
 	(\exists x \psi)^{\dagger} & = \exists x \psi^{\dagger}, \\
 	(\psi  \Until \chi)^{\dagger} & = \psi^{\dagger} \Until (\chi^{\dagger} \land \psi_{f}^{1}).
\end{align*}
Before
showing the correctness of the translation, the
following lemma shows the relevance of end extensions when
interpreting translated formulas.
%
\begin{restatable}{lemma}{LemmaTranslation}\label{lemma:finextchar}
  Let $\FEImf$ be an end extension of a finite trace $\Fmf$.  For
  every $\QTLfr{\Until}{}{}$ formula $\p$ and every assignment $\assign$,
  $\Fmf \mdl^{\assign} \p \ \text{iff} \ \ \FEImf \mdl^{\assign} \p\tr$.
\end{restatable}
\begin{proof}
Let
$\Fmf = (\Delta^\Fmf, (\Fmc_n)_{n \in [0, l]})$
be a finite trace
and let
$\FEImf = (\Delta^{\FImf}, (\FImc_n)_{n \in [0, \infty)})$
be an end extension of $\Fmf$. 
%
We prove by structural induction the following more general statement.
For all $n \in [0, l]$ and all assignments $\assign$:
\[
\Fmf,n \mdl^{\assign} \p \ \text{iff} \ \FEImf,n \mdl^{\assign} \p\tr.
\]

For the base case $\p = P(\bar{\tau})$, the statement follows
from the definitions of $\FEImf$ and $\cdot\tr$, while
the proof of the inductive cases $\p = \lnot \psi$,
$\p = (\psi \land \chi)$, and $\p = \exists x \psi$ is
straightforward.

We show the inductive case $\p = (\psi \Until \chi)$.  We have that
$\Fmf, n \mdl^{\assign} \psi \Until \chi$ iff there is $m \in (n, l]$
such that $\Fmf, m \models^{\assign} \chi$ and for all $i \in (n,m)$,
$\Fmf,i \models^{\assign} \psi$.  By the inductive hypothesis,
this happens iff there is $m \in (n, l]$ such that
$\FEImf, m \models^{\assign} \chi\tr$ and for all $i \in (n,m)$,
$\FEImf,i \models^{\assign} \psi\tr$.  Since $E^{\FEImc_{j}} = \eset$
for all $j \in [0, l]$, this means that
$\FEImf, n \mdl^{\assign} \psi\tr \Until (\chi\tr \land \forall x
\lnot E(x))$.  That is,
$\FEImf, n \mdl^{\assign} (\psi \Until \chi)\tr$.
%
\end{proof}

Using the previous lemmas, we can show the correctness of the
reduction of the $\QTLfr{\Until}{}{}$ satisfiability problem on finite
traces to the same problem for $\QTLfr{\Until}{}{}$ on infinite
traces.

\begin{restatable}{theorem}{TheorFiniteChar}\label{theor:finsatchar}
  A $\QTLfr{\Until}{}{}$ formula $\varphi$ is satisfiable on finite
  traces iff $\varphi\tr \land \psi_{f}$ is satisfiable on infinite
  traces.
\end{restatable}
\begin{proof} 
If
$\p$ is satisfied in some finite trace $\Fmf$, then (by
Lemmas~\ref{lemma:finchar} and \ref{lemma:finextchar}) any end
extension $\FEImf$ satisfies $\p\tr \land \psif$. Conversely, suppose
that $\p\tr \land \psif$ is satisfied in some infinite trace $\Imf$
under an assignment $\assign$. By Lemma~\ref{lemma:finchar},
$\Imf = \FEImf'$, for some finite trace $\Fmf$ and some infinite trace $\Imf'$. 
Since $\FEImf' \mdl^{\assign} \p\tr$, by Lemma~\ref{lemma:finextchar}, we have that  $\Fmf \mdl^{\assign} \p$.
\end{proof}

%


\subsection{Blurring the Distinction Between Finite and Infinite Traces}
\label{sec:fininf}

While certain formulas, such as $\B \top$, are satisfiable both on
finite and infinite traces, others, e.g., $\D \last$ and
$\B^{+} \Next \top$, are only satisfiable on finite traces and on
infinite traces, respectively.
It is thus of interest to understand in which cases satisfiability on
finite and infinite traces coincide, so that solving the problem in
one case answers to the other as well.
A similar question can be posed for the problem of equivalences
between formulas.
For example, $\D \B (\p \vee \psi)$ and $\D \B \p \vee \D \B\psi$ are
equivalent on finite traces but not on infinite traces~\cite{BauHas}.
Moreover, $\B^{+} \D^{+} \p$ and $\D^{+} \B^{+} \p$ are not equivalent
on infinite traces, whereas on finite traces they are both equivalent
to
$\D^{+} (\last \land \p)$~\cite{DegVar1}.  Conversely, $\bot$ and
$\last$
are only equivalent on infinite traces.

In this section we address these questions and investigate the
distinction between reasoning on finite and on infinite
traces.
%
We first propose semantic properties under which it is guaranteed that
formula
equivalences are preserved from finite to infinite traces, or vice versa,
thus allowing to blur the distinction between these semantics.  
Then, we syntactically define classes of formulas satisfying some of these semantic properties, so to provide a sufficient criterion for the preservation of equivalences
from finite to infinite traces, or vice versa.
Finally, we focus on preserving satisfiability from the finite to the infinite case, devising a wider class of formulas for which this preservation holds.


\subsubsection{Finite vs. Infinite Traces: Semantic Characterisation}

%
For a $\QTLfr{\Until}{}{}$ formula $\varphi$ and a quantifier
$Q\in\{\exists, \forall \}$,
we say that \emph{$\varphi$ satisfies $\finit{}{Q}$} (or that
\emph{$\varphi$ is $\finit{}{Q}$}) if, for all finite traces $\Fmf$
and all assignments $\assign$, it satisfies the \emph{finite trace
  property}:
\[
\Fmf \models^{\assign} \p \Leftrightarrow Q \Imf \in \textit{Ext}(\Fmf). \Imf \models^{\assign} \p,
\]
and, similarly, that \emph{$\varphi$ satisfies $\infinit{}{Q}$} (or
that \emph{$\varphi$ is $\infinit{}{Q}$}) if, for all infinite traces
$\Imf$ and all assignments $\assign$, it satisfies the \emph{infinite
  trace property}:
\[
\Imf \models^{\assign} \p \Leftrightarrow Q \Fmf \in \textit{Pre}(\Imf).\Fmf \models^{\assign} \p.
\]
%
%
To
show intuitive examples, let us
consider the case where $\varphi$ is a Boolean combination of atomic
formulas.
Examples
of formulas satisfying $\finit{}{\forall}$ and $\infinit{}{\exists}$
are formulas of the form $\D^{+} \varphi$.
Formulas of the form $\D \varphi$
are also $\infinit{}{\exists}$, but in general not
$\finit{}{\forall}$, as witnessed, for instance, by $\D \top$.  On the
other hand, the properties $\finit{}{\exists}$ and
$\infinit{}{\forall}$ capture for example formulas of the form
$\B^{+} \varphi$.
Formulas of the form $\B \varphi$
are also $\infinit{}{\forall}$, but not necessarily
$\finit{}{\exists}$, because of, e.g., $\B \bot$.  These observations,
that can easily be checked, are also immediate consequences of
Lemmas~\ref{thm:diamond} and~\ref{thm:box} below.

We also restrict to
the ``one directional'' version of the above properties.
We denote by $\finit{\circ}{Q}$ and $\infinit{\circ}{Q}$, where
$\circ \in \{ \Rightarrow, \Leftarrow \}$, the corresponding
`$\Rightarrow$' and `$\Leftarrow$' directions of the $\finit{}{Q}$ and
$\infinit{}{Q}$ properties, respectively.
%
Finally, given a property $\mathsf{P}$, we denote by $\QTLfr{\Until}{}{}(\mathsf{P})$ the set of $\QTLfr{\Until}{}{}$ formulas satisfying $\mathsf{P}$.


The semantic properties $\finit{}{Q}$ and $\infinit{}{Q}$ capture different classes 
of  $\QTLfr{\Until}{}{}$ formulas, as
illustrated by the following example.
\begin{example}
\label{tab:inscomp}
The following formulas satisfy exactly one of 
the corresponding finite or infinite trace properties.
\begin{table}[h]
\centering
\begin{tabular}
{| c | c |}
\hline
$\finit{}{\exists}$ & $\D^{+} \last \lor \D P(x)$
\\
\hline
$\finit{}{\forall}$ & $\forall x \D^{+} P (x)$
\\
\hline
$\infinit{}{\exists}$ & $\B \Next \top \lor \last$
\\
\hline
$\infinit{}{\forall}$ & $\B^{+}P(x) \lor \D^{+}(P(x) \land \last)$
\\
\hline
\end{tabular}
\end{table}
%
\end{example}

Indeed, by using the formulas from Example~\ref{tab:inscomp}, we can prove
the following.

\begin{proposition}

The sets
$\QTLfr{\Until}{}{}(\finit{}{\exists})$,
$\QTLfr{\Until}{}{}(\finit{}{\forall})$,
$\QTLfr{\Until}{}{}(\infinit{}{\exists})$,
and
$\QTLfr{\Until}{}{}(\infinit{}{\forall})$
are mutually incomparable with respect to inclusion. 
\end{proposition}
\begin{proof}
For every
$X, Y \in \{ \QTLfr{\Until}{}{}(\finit{}{\exists}), \QTLfr{\Until}{}{}(\finit{}{\forall}), \QTLfr{\Until}{}{}(\infinit{}{\exists}), \QTLfr{\Until}{}{}(\infinit{}{\forall}) \}$,
we use the formulas from Example~\ref{tab:inscomp} to show that $X \not \subseteq Y$.
\begin{itemize}
	\item $\QTLfr{\Until}{}{}(\finit{}{\exists}) \not \subseteq Y$, with $Y \in \{ \QTLfr{\Until}{}{}(\finit{}{\forall}), \QTLfr{\Until}{}{}(\infinit{}{\exists}), \QTLfr{\Until}{}{}(\infinit{}{\forall}) \}$.
	It can be seen that the formula $\Diamond^{+} \last \lor \D P(x)$ is $\finit{}{\exists}$. However,
		\begin{itemize}
			\item it is not $\finit{}{\forall}$:
				(under any assignment)
				the formula
				is satisfied in a finite trace $\Fmf = ( \Delta, (\Fmc_{0}))$, with $0$ as its only time point and such that $P^{\Fmc_{0}} = \emptyset$, but an extension $\Imf \in \Ext(\Fmf)$ such that $\Imf = ( \Delta, (\Imc_{n})_{n \in [0, \infty)})$, with $\Fmc_{0} = \Imc_{0}$ and $P^{\Imc_{n}} = \emptyset$, for every $n \in [0, \infty)$, does not satisfy it;
			\item it is not $\infinit{}{\exists}$:
				(under any assignment) an infinite trace $\Imf = ( \Delta, (\Imc_{n})_{n \in [0, \infty)})$ such that $P^{\Imc_{n}} = \emptyset$, for every $n \in [0, \infty)$ does not satisfy the formula, whereas any (and thus some) prefix $\Fmf \in \Pre(\Imf)$ satisfies it;
			\item it is not $\infinit{}{\forall}$: shown as in the previous case.
		\end{itemize}
	\item $\QTLfr{\Until}{}{}(\finit{}{\forall}) \not \subseteq Y$, with $Y \in \{ \QTLfr{\Until}{}{}(\finit{}{\exists}), \QTLfr{\Until}{}{}(\infinit{}{\exists}), \QTLfr{\Until}{}{}(\infinit{}{\forall}) \}$.
		It can be seen that the formula $\forall x \D^{+} P (x)$ is \finit{}{\forall}. However,
		\begin{itemize}
			\item it is not $\finit{}{\exists}$: (under any assignment) the finite trace $\Fmf = ( \Delta, (\Fmc_{0}))$, with $0$ as its only time point and such that $P^{\Fmc_{0}} = \emptyset$, does not satisfy the formula, whereas an extension $\Imf \in \Ext(\Fmf)$ such that $\Imf = ( \Delta, (\Imc_{n})_{n \in [0, \infty)})$, with $P^{\Imc_{1}} = \Delta$, satisfies it;
			\item it is not $\infinit{}{\exists}$: (under any assignment) an infinite trace $\Imf = ( \{ d_{i} \}_{i \in \mathbb{N}}, (\Imc_{n})_{n \in [0, \infty)})$ such that $P^{\Imc_{i}} = \{ d_{i} \}$, for every $i \in \mathbb{N}$, satisfies the formula, but there is no finite prefix $\Fmf \in \Pre(\Imf)$ that satisfies it;
			\item it is not $\infinit{}{\forall}$: (under any assignment) an infinite trace $\Imf = ( \Delta, (\Imc_{n})_{n \in [0, \infty)})$ such that $P^{\Imc_{0}} = \emptyset$ and $P^{\Imc_{1}} = \Delta$ satisfies the formula, but the prefix $\Fmf \in \Pre(\Imf)$ such that $\Fmf = (\Delta, (\Fmc_{0}))$, where $\Fmc_{0} = \Imc_{0}$, does not satisfy it.
		\end{itemize}
	\item $\QTLfr{\Until}{}{}(\infinit{}{\exists}) \not \subseteq Y$, with $Y \in \{ \QTLfr{\Until}{}{}(\finit{}{\exists}), \QTLfr{\Until}{}{}(\finit{}{\forall}), \QTLfr{\Until}{}{}(\infinit{}{\forall}) \}$.
			It can be seen that the formula $\B \Next \top \lor \last$ is $\infinit{}{\exists}$. However,
		\begin{itemize}
			\item it is not $\finit{}{\exists}$: (under any assignment) a finite trace $\Fmf = ( \Delta, (\Fmc_{0}, \Fmc_{1}))$, with $0, 1$ as its only time points, does not satisfy the formula (at time point $0$), whereas any (and thus some) extension $\Imf \in \Ext(\Fmf)$ satisfies it;
			\item it is not $\finit{}{\forall}$: shown as in the previous case;
			\item it is not $\infinit{}{\forall}$: (under any assignment) an infinite trace $\Imf = ( \Delta, (\Imc_{n})_{n \in [0, \infty)})$ satisfies the formula, but a prefix $\Fmf \in \Pre(\Imf)$ such that $\Fmf = (  \Delta, (\Fmc_{0}, \Fmc_{1}))$, with $\Fmc_{i} = \Imc_{i}$, for $i \in \{ 0, 1 \}$, does not.
		\end{itemize}
	\item $\QTLfr{\Until}{}{}(\infinit{}{\forall}) \not \subseteq Y$, with $Y \in \{ \QTLfr{\Until}{}{}(\finit{}{\exists}), \QTLfr{\Until}{}{}(\finit{}{\forall}), \QTLfr{\Until}{}{}(\infinit{}{\exists}) \}$.
			It can be seen that the formula $\B^{+}P(x) \lor \D^{+}(P(x) \land \last)$ is $\infinit{}{\forall}$. However,
		\begin{itemize}
			\item it is not $\finit{}{\exists}$: (under any assignment) the formula is satisfied in a finite trace $\Fmf = ( \Delta, (\Fmc_{0}, \Fmc_{1}))$, with $0, 1$ as its only time points and such that $P^{\Fmc_{0}} = \emptyset$ and $P^{\Fmc_{1}} = \Delta$, whereas no extension $\Imf \in \Ext(\Fmf)$ satisfies it;
			\item it is not $\finit{}{\forall}$: shown as in the previous case;
			\item it is not $\infinit{}{\exists}$: (under any assignment) an infinite trace $\Imf = ( \Delta, (\Imc_{n})_{n \in [0, \infty)})$ such that $P^{\Imc_{0}} = \Delta$ and $P^{\Imc_{i}} = \emptyset$, for $i > 0$, does not satisfy the formula, whereas the prefix $\Fmf \in \Ext(\Imf)$ such that $\Fmf = ( \Delta, (\Fmc_{0}))$, with $\Fmc_{0} = \Imc_{0}$, satisfies it.
			\qedhere
		\end{itemize}
\end{itemize}
\end{proof}

On the relationships between the one directional properties, we have
the following.
\begin{proposition}
  The following statements hold.
\begin{enumerate}
\item Given $Q, Q' \in \{ \exists, \forall \}$, with $Q \neq Q'$,
we have:
  $\p \in \QTLfr{\Until}{}{}(\finit{\Rightarrow}{Q})$ iff
  $\lnot \p \in \QTLfr{\Until}{}{}(\finit{\Leftarrow}{Q'})$, and
  $\p \in \QTLfr{\Until}{}{}(\infinit{\Rightarrow}{Q})$ iff
  $\lnot \p \in \QTLfr{\Until}{}{}(\infinit{\Leftarrow}{Q'})$.
\item
Given $\mathsf{P} \in \{ \mathsf{F}, \mathsf{I} \}$,
we have
  $\QTLfr{\Until}{}{}(\mathsf{P}_{\Rightarrow\forall}) \subseteq
  \QTLfr{\Until}{}{}(\mathsf{P}_{\Rightarrow\exists})$ and
  $\QTLfr{\Until}{}{}(\mathsf{P}_{\Leftarrow\exists}) \subseteq
  \QTLfr{\Until}{}{}(\mathsf{P}_{\Leftarrow\forall})$.
 %
\item $\QTLfr{\Until}{}{}(\finit{\Rightarrow}{\forall})$ and
  $\QTLfr{\Until}{}{}(\infinit{\Rightarrow}{\exists})$, as well as
  $\QTLfr{\Until}{}{}(\infinit{\Rightarrow}{\forall})$ and
  $\QTLfr{\Until}{}{}(\finit{\Rightarrow}{\exists})$, are incomparable
  with respect to inclusion.
\end{enumerate}
\end{proposition}
\begin{proof}
  $(1)$ Let $Q, Q' \in \{ \exists, \forall \}$, with $Q \neq Q'$, and
  suppose that $\p \in \QTLfr{\Until}{}{}(\finit{\Rightarrow}{Q})$.
  This means that, for all $\Fmf$ and all assignments $\assign$, we
  have:
  $\Fmf \models^{\assign} \p \Rightarrow Q \Imf \models^{\assign} \p$.
  By contraposition, the previous step means, for all $\Fmf$ and all
  assignments $\assign$:
  $Q' \Imf \not \models^{\assign} \p \Rightarrow \Fmf \not
  \models^{\assign} \p$.  That is, for all $\Fmf$ and all assignments
  $\assign$:
  $\Fmf \models^{\assign} \lnot \p \Leftarrow Q' \Imf
  \models^{\assign} \lnot \p$.  Thus,
  $\lnot \p \in \QTLfr{\Until}{}{}(\finit{\Leftarrow}{Q'})$.
Similarly, it can be seen that $\p \in \QTLfr{\Until}{}{}(\infinit{\Rightarrow}{Q})$ iff $\lnot \p \in \QTLfr{\Until}{}{}(\infinit{\Leftarrow}{Q'})$, for $Q, Q' \in \{ \exists, \forall \}$, $Q \neq Q'$.

$(2)$ Straightforward from the definitions.

$(3)$ We have, e.g., that $\B^{+} \Next \top$ always holds on infinite traces, but it is unsatisfiable on finite traces.
Hence, $\B^{+} \Next \top \in \QTLfr{\Until}{}{}(\finit{\Rightarrow}{\forall})$, whereas $\B^{+} \Next \top \not \in \QTLfr{\Until}{}{}(\infinit{\Rightarrow}{\exists})$.
By Point~$(2)$, this implies also that $\B^{+} \Next \top \in \QTLfr{\Until}{}{}(\finit{\Rightarrow}{\exists})$ and $\B^{+} \Next \top \not \in \QTLfr{\Until}{}{}(\infinit{\Rightarrow}{\forall})$.
On the other hand, $\D \last$ always holds on finite traces, while it is unsatisfiable on infinite ones. Hence, $\D \last \in \QTLfr{\Until}{}{}(\infinit{\Rightarrow}{\forall})$, but $\D \last \not \in \QTLfr{\Until}{}{}(\finit{\Rightarrow}{\exists})$.
By Point~$(2)$, we obtain also that $\D \last \in \QTLfr{\Until}{}{}(\infinit{\Rightarrow}{\exists})$, and $\D \last \not \in \QTLfr{\Until}{}{}(\finit{\Rightarrow}{\forall})$.
\end{proof}

We now consider the problem of formula equivalence, by showing under
which semantic properties equivalence between formulas can be
blurred.
The following theorem provides sufficient conditions to
preserve formula equivalence from the infinite to the finite
case (cf. the notion of \emph{$\LTL$ compliance} in~\cite{BauEtAl}).

%
\begin{restatable}{theorem}{Theoremusonlyif}\label{thm:equivusif}
Given $X \in \{ \QTLfr{\Until}{}{}(\finit{}{\exists}), \QTLfr{\Until}{}{}(\finit{}{\forall}),
 \QTLfr{\Until}{}{}(\finit{\Rightarrow}{\exists}) \cap \QTLfr{\Until}{}{}(\infinit{\Rightarrow}{\forall}) \}$
 and $\p, \psi \in X$, it holds that
$\equivinf{\p}{\psi}$ implies $\equivfin{\p}{\psi}$.
\end{restatable}
%


\begin{proof}

  First, assume
$\p, \psi \in \QTLfr{\Until}{}{}(\finit{}{\exists})$.
 Let $\Fmf$ be a finite trace and $\assign$ an assignment such that $\Fmf \mdl^{\assign} \p$.  By
  $\finit{}{\exists}$, there is an infinite trace
  $\Imf \in \Ext(\Fmf)$ such that $\Imf \mdl^{\assign} \p$. Since
  $\equivinf{\p}{\psi}$, we have that $\Imf \mdl^{\assign} \psi$.  By
  $\finit{}{\exists}$, $\Fmf \models^{\assign} \psi$.  
   The converse
   direction can be obtained similarly, by swapping $\p$ and $\psi$.
   The proof for
   $\p, \psi \in \QTLfr{\Until}{}{}(\finit{}{\forall})$
  is analogous.

   We now show the statement for
     $\p, \psi \in \QTLfr{\Until}{}{}(\finit{\Rightarrow}{\exists}) \cap \QTLfr{\Until}{}{}(\infinit{\Rightarrow}{\forall})$.  
   Given a finite trace $\Fmf$ and an assignment
   $\assign$, suppose that $\Fmf \mdl^{\assign} \p$.  Since $\p$
   is $\finit{\Rightarrow}{\exists}$, we have that for some
   $\Imf \in \Ext(\Fmf)$, $\Imf \mdl^{\assign} \p$.  By assumption,
   $\equivinf{\p}{\psi}$, so $\Imf \mdl^{\assign} \psi$.  As $\psi$ is
   $\infinit{\Rightarrow}{\forall}$, $\Imf \mdl^{\assign} \psi$
   implies that, for all $\Fmf' \in \Pre(\Imf)$,
   $\Fmf' \mdl^{\assign} \psi$.  Thus, in particular,
   $\Fmf \mdl^{\assign} \psi$.  The other direction can be obtained
     by swapping $\p$ and $\psi$.
\end{proof}

%
Theorem~\ref{thm:equivusif} does not hold for formulas that satisfy   
only 
$\infinit{}{\exists}$ or $\infinit{}{\forall}$.
Consider the formulas $\B \Next \top \lor \last$, from Example~\ref{tab:inscomp}, and
$\B \Next \top \lor \Next \last$, which are both $\infinit{}{\exists}$.
These formulas are equivalent only on infinite traces.
Also, $\B^{+}P(x) \lor \D^{+}(P(x) \land \last)$, from Example~\ref{tab:inscomp},
and
$\B^{+}P(x) \lor \D^{+}(P(x) \land \Next \last)$ are
$\infinit{}{\forall}$, and equivalent on infinite but not on finite
traces.  The last
example also shows that the
condition $\infinit{\Rightarrow}{\forall}$ 
alone is not sufficient for Theorem~\ref{thm:equivusif}.  Moreover,
$\finit{\Rightarrow}{\exists}$ alone is also not sufficient.
To see this, consider, e.g., $\B^+\D\top\vee (P(x)\wedge\D \last)$ 
and $\B^+\D\top \vee \D \last$ , which are $\finit{\Rightarrow}{\exists}$  
but
equivalent only on infinite traces. 
%

We now present sufficient conditions   to preserve
equivalences from the finite to the infinite case.

\begin{restatable}{theorem}{Theoremusonlyif}\label{thm:equivconverse}
Given $X \in \{ \QTLfr{\Until}{}{}(\infinit{}{\exists}), \QTLfr{\Until}{}{}(\infinit{}{\forall}),
 \QTLfr{\Until}{}{}(\finit{\Rightarrow}{\forall}) \cap \QTLfr{\Until}{}{}(\infinit{\Rightarrow}{\exists}) \}$
 and $\p, \psi \in X$, it holds that
$\equivfin{\p}{\psi}$ implies $\equivinf{\p}{\psi}$.
\end{restatable}
%


%

\begin{proof}


  %
  First, suppose that $\p, \psi \in \QTLfr{\Until}{}{}(\infinit{}{\exists})$.
  Let \Imf be an infinite trace
  and $\assign$ an assignment 
  such that
  $\Imf \models^{\assign} \p$.  By $\infinit{}{\exists}$,
  there is $\Fmf \in \Pre(\Imf)$ such that $\Fmf \mdl^{\assign} \p$.
  As $\equivfin{\p}{\psi}$, this means that
  $\Fmf \mdl^{\assign} \psi$.  Since $\psi$ is $\infinit{}{\exists}$,
  we have $\Imf \mdl^{\assign} \psi$.  The converse direction is
  obtained similarly, by swapping $\p$ and $\psi$. 
   The proof
  for
  $\p, \psi \in \QTLfr{\Until}{}{}(\infinit{}{\forall})$
  is analogous.

  We now show
  the statement for $\p, \psi \in  \QTLfr{\Until}{}{}(\finit{\Rightarrow}{\forall}) \cap \QTLfr{\Until}{}{}(\infinit{\Rightarrow}{\exists})$.
  Let $\Imf$ be an infinite trace and
  $\assign$ be an assignment such that $\Imf \mdl^{\assign} \p$.  As
  $\p$ is $\infinit{\Rightarrow}{\exists}$, there is
  $\Fmf \in \Pre(\Imf)$ such that $\Fmf \mdl^{\assign} \p$.  Given
  that $\equivfin{\p}{\psi}$, $\Fmf \mdl^{\assign} \psi$.  By
  $\finit{\Rightarrow}{\forall}$, for all $\Imf' \in \Ext(\Fmf)$:
  $\Imf' \mdl^{\assign} \psi$.  Therefore, we have also
  $\Imf \mdl^{\assign} \psi$.  The converse direction is obtained in a
  similar way by swapping $\p$ and $\psi$.
\end{proof}

%
The properties
$\finit{}{\exists}$ or $\finit{}{\forall}$
alone
are not
sufficient to ensure that formula equivalence on finite traces
implies formula equivalence on infinite traces.  To illustrate this,
consider for example the formulas
$\D^{+} \last \lor \D P(x)$,
 shown in Example~\ref{tab:inscomp},
and
$\D^{+} \last \lor \D P(x) \lor \vartheta_{i}$,
where Formula~$\vartheta_{i}$
is from
Section~\ref{sec:boundprop}.
These formulas are $\finit{}{\exists}$, however, they are only
equivalent on finite traces.
Moreover, if we take
$\forall x \D^{+} P (x)$, 
from
Example~\ref{tab:inscomp},
and
$\forall x \D^{+} P (x)\lor \vartheta_{i}$,
we have that they are both
$\finit{}{\forall}$, though equivalent only on finite traces.  The
last example also shows that the condition
$\finit{\Rightarrow}{\forall}$ 
alone is not sufficient for Theorem~\ref{thm:equivconverse}.  We now
argue that $\infinit{\Rightarrow}{\exists}$ alone is also not
sufficient.
  To see this,
consider, e.g., $(P(x)\wedge\B^+\D\top)\vee \D \textit{last}$ and
$\B^+\D\top\vee \D \textit{last}$, which are
$\infinit{\Rightarrow}{\exists}$ but are equivalent only on finite
traces.


From Theorems~\ref{thm:equivusif} and~\ref{thm:equivconverse} 
we have that if
$\p, \psi \in \QTLfr{\Until}{}{}(\finit{}{\exists})$ or
$\p, \psi \in \QTLfr{\Until}{}{}(\finit{}{\forall})$,
and
$\p, \psi \in \QTLfr{\Until}{}{}(\infinit{}{\exists})$ or
$\p, \psi \in \QTLfr{\Until}{}{}(\infinit{}{\forall})$,
then
$\equivfin{\p}{\psi}$ if and only if $\equivinf{\p}{\psi}$.
In particular, the above examples show that if, from a given pair of
conditions $\finit{}{Q}$ and $\infinit{}{Q'}$, we remove any of the
two properties, then formula equivalences on finite and infinite
traces may not coincide.

\subsubsection{Preserving Formula Equivalences: Syntactic Characterisation}

We now analyse syntactic features of the properties introduced so far,
providing classes of formulas that satisfy them. This will in turn
allow us to show results on preservation of equivalences, for such
formulas, between finite and infinite traces.

First, we make the following observation concerning non-temporal
$\QTLfr{\Until}{}{}$ formulas.
\begin{restatable}{proposition}{TheoremBoolean}\label{thm:boolean}
  For every non-temporal $\QTLfr{\Until}{}{}$ formula $\p$, every
  finite trace $\Fmf = (\Delta^{\Fmf}, (\Fmc_{n})_{n \in [0, l]})$,
  every infinite trace
  $\Imf = (\Delta^{\Imf}, (\Imc_{n})_{n \in [0, \infty)})$, every
  $n \geq 0$, and every assignment $\assign$ in $\Fmf$ or $\Imf$,
  respectively, the following hold, where $Q \in \{ \exists, \forall \}$:
\begin{itemize}
	\item $\Fmf, n \models^{\assign} \p \Leftrightarrow Q \Imf \in \textit{Ext}(\Fmf). \Imf, n \models^{\assign} \p;$
	\item $\Imf, n \models^{\assign} \p \Leftrightarrow Q \Fmf \in \textit{Pre}(\Imf).\Fmf, n \models^{\assign} \p.$
\end{itemize}
In particular, $\p \in \QTLfr{\Until}{}{}(\mathsf{P})$, for every
$\mathsf{P} \in \{ \finit{}{\exists}, \finit{}{\forall},
\infinit{}{\exists}, \infinit{}{\forall} \}$.
\end{restatable}
%
\begin{proof}
Clearly, since $\p$ has no temporal operators, for any
finite or infinite trace 
\Mmf,  $\Mmf$ satisfies $\p$ at $n$ under $\assign$ iff  
any extension or prefix of \Mmf satisfies $\p$ at $n$ under $\assign$, respectively. 
\end{proof}

We now introduce the relevant fragments of $\QTLfr{\Until}{}{}$ that will be analysed in the rest of this section.
%
%
First, \emph{$\UntilP$-formulas} $\varphi,\psi$ are built according to the grammar (with
$P\in\textsf{N}_{\textsf{P}}$):
\[
P(\bar{\tau})\mid \neg
P(\bar{\tau}) \mid
\varphi\wedge \psi \mid
\varphi\vee \psi \mid
\exists x \varphi \mid
\p \UntilP \psi.
\]
Moreover, we call \emph{$\Until$-formulas} the set of formulas generated by
allowing $\p \Until \psi$ in the grammar rule for
$\UntilP$-formulas\footnote{Recall that $\UntilP$ is syntactic sugar in the fragment of {$\Until$}-formulas,
since $\p \UntilP \psi := \psi \lor (\p \land \p \Until \psi)$.},
and we call
\emph{$\UntilP\forall$-formulas} the result of allowing $\forall
x\varphi$ in the grammar rule for $\UntilP$-formulas.


Next,
\emph{$\ReleaseP$-formulas}
$\varphi,\psi$ are built according to the grammar (with
$P\in\textsf{N}_{\textsf{P}}$):
\[
P(\bar{\tau})\mid \neg
P(\bar{\tau}) \mid
\varphi\wedge \psi \mid
\varphi\vee \psi \mid
\forall x \varphi \mid
\p \ReleaseP \psi.
\]
We call
\emph{$\Release$-formulas} the set of formulas generated by
allowing $\p \Release \psi$ in the grammar rule for $\ReleaseP$-formulas\footnote{Recall that $\ReleaseP$ is syntactic sugar in the fragment of {$\Release$}-formulas, since $\p \ReleaseP \psi := \psi \land (\p \lor \p \Release \psi)$.}, and
we call \emph{$\ReleaseP\exists$-formulas} the result of allowing
$\exists x\varphi$ in the grammar rule for $\ReleaseP$-formulas.



\begin{table}[h]
\caption{
Syntactic fragments with corresponding semantic properties and preservation of equivalences.
}
\centering
\renewcommand{\arraystretch}{1.1}
\begin{tabular}{| >{\centering}m{1cm} | >{\centering}m{2.05cm}  | >{\centering}m{2.7cm} | }
\cline{1-3}
& Properties & Equivalences
\tabularnewline
\hline
	\multirow{ 2}{*}{$\UntilP\forall$}
	&
	$\finit{}{\forall}$	
	&
	$i$ $\Rightarrow$ $f$
\tabularnewline
  & (Lemma~\ref{thm:diamond}) & (Theorem~\ref{thm:equivusif})
\tabularnewline                               
\hline
	\multirow{ 2}{*}{$\ReleaseP\exists$}
	&
	$\finit{}{\exists}$
	&
	$i$ $\Rightarrow$ $f$
\tabularnewline
     & (Lemma~\ref{thm:box}) & (Theorem~\ref{thm:equivusif})
\tabularnewline
\hline
	\multirow{ 2}{*}{$\Until$}
	&
	$\infinit{}{\exists}$
	&
	$f$ $\Rightarrow$ $i$
\tabularnewline
  & (Lemma~\ref{thm:diamond}) & (Theorem~\ref{thm:equivconverse})
\tabularnewline                               
\hline
	\multirow{ 2}{*}{$\Release$}
	&
	$\infinit{}{\forall}$
	&
	$f$ $\Rightarrow$ $i$
\tabularnewline
     & (Lemma~\ref{thm:box}) & (Theorem~\ref{thm:equivconverse})
\tabularnewline
\hline
	\multirow{ 2}{*}{$\UntilP$}
	&
	$\finit{}{\forall}$, $\infinit{}{\exists}$
	&
	$f$ $\Leftrightarrow$ $i$
\tabularnewline
  & (Lemma~\ref{thm:diamond}) & (Theorems~\ref{thm:equivusif}-\ref{thm:equivconverse})
\tabularnewline                               
\hline
	\multirow{ 2}{*}{$\ReleaseP$}
	&
	$\finit{}{\exists}$, $\infinit{}{\forall}$
	&
	$f$ $\Leftrightarrow$ $i$
\tabularnewline
     & (Lemma~\ref{thm:box}) & (Theorems~\ref{thm:equivusif}-\ref{thm:equivconverse})
\tabularnewline
\hline
\end{tabular}
\label{tab:mainres}
\end{table}

Having introduced such fragments, the rest of this section will be
devoted to the proof of the following theorem, which is a consequence
of Theorems~\ref{thm:equivusif}-\ref{thm:equivconverse} above and
Lemmas~\ref{thm:diamond}-\ref{thm:box} below, as outlined in
Table~\ref{tab:mainres}.

\begin{theorem}
\label{thm:mainequiv}
%
The following hold:
\begin{enumerate}
	\item for all $\UntilP$- or $\ReleaseP$-formulas $\p$ and $\psi$, $\equivfin{\p}{\psi}$ if and only if $\equivinf{\p}{\psi}$;
	\item for all $\UntilP\forall$- or $\ReleaseP\exists$-formulas $\p$ and $\psi$, $\equivinf{\p}{\psi}$ implies $\equivfin{\p}{\psi}$;
	\item for all $\Until$- or $\Release$-formulas $\p$ and $\psi$, $\equivfin{\p}{\psi}$ implies $\equivinf{\p}{\psi}$.
\end{enumerate}
\end{theorem}


%

We first show that every
$\UntilP\forall$-formula
is $\finit{}{\forall}$, and
every
$\Until$-formula
is
$\infinit{}{\exists}$.  As an immediate consequence, we obtain that
every
$\UntilP$-formula
is both
$\finit{}{\forall}$ and $\infinit{}{\exists}$.
\begin{restatable}{lemma}{TheoremDiamond}\label{thm:diamond}
  $\UntilP\forall$-formulas are $\finit{}{\forall}$
      and
      $\Until$-formulas are $\infinit{}{\exists}$.
      Thus,
      $\UntilP$-formulas are both $\finit{}{\forall}$ and
      $\infinit{}{\exists}$.
\end{restatable}
\begin{proof}

We first show that all
$\UntilP\forall$-formulas
are $\finit{}{\forall}$.
  In Claim~\ref{cl:diamond1}, we show that all
  $\UntilP\forall$-formulas
  are 
  $\finit{\Rightarrow}{\forall}$
  (in fact, for the $\finit{\Rightarrow}{\forall}$ case,
  we can also allow
    $\p \Until \psi$
  in the grammar).
  Then, in Claim~\ref{cl:diamond2}, 
  we show that all 
    $\UntilP\forall$-formulas
  are 
  $\finit{\Leftarrow}{\forall}$.
\begin{claim}\label{cl:diamond1}
%
  $\UntilP\forall$-formulas
  are
  $\finit{\Rightarrow}{\forall}$.
\end{claim}
\begin{proof}[Proof of Claim~\ref{cl:diamond1}]
Given a finite trace $\Fmf = (\Delta, (\Fmc_{n})_{n \in [0,l]})$ and an assignment $\assign$ in $\Fmf$,
we show that
$\Fmf \models^{\assign} \p$ implies
that, for every $\Imf \in \Ext(\Fmf)$,
$\Imf \models^{\assign} \p$.
The proof is by structural induction on $\p$.
By Proposition~\ref{thm:boolean}, the statement holds for the base cases of $\p = P(\bar{\tau})$ and $\p = \lnot P(\bar{\tau})$.
We now proceed with the inductive steps.
\begin{itemize}
\item
$\p = \psi \UntilP \chi$.  Suppose that
      $\Fmf\models^\assign \psi \UntilP \chi$.  This means that
      there exists $n\in [0,l]$ such that
      $\Fmf,n\models^\assign \chi$, i.e.,
      $\Fmf^n\models^\assign \chi$, and, for every $i \in [0, n)$,
      $\Fmf, i \models^\assign \psi$, i.e.,
      $\Fmf^{i} \models^\assign \psi$.
      By the inductive hypothesis,
      we have that there exists $n\in [0,l]$ such that
      $\Imf \models^\assign \chi$, for all
      $\Imf\in\textit{Ext}(\Fmf^n)$, and, for every $i \in [0, n)$,
      $\Fmf^{i} \models^\assign \psi$.
      Since, for every
      $\Imf \in \Ext(\Fmf)$ and $m \in (0, l]$, we have that $\Imf = \Fmf_{m - 1} \cdot \Imf'$, for some $\Imf' \in \Ext(\Fmf^{m})$, the previous
      step implies that, for all $\Imf\in \Ext(\Fmf)$, there exists
      $n\in [0,l]$ such that $\Imf \models^\assign \chi$, and, for
      every $i \in [0, n)$, $\Imf \models^\assign \psi$.
      That is,
      $\Imf \models^\assign \psi \UntilP \chi$.
\item $\varphi =\forall x\psi$. Suppose that
  $\Fmf\models^\assign \forall x\psi$.
This means that, for all $d\in\Delta$,
 $\Fmf\models^{\assign[x \mapsto d]} \psi$.
 By the inductive hypothesis, we have that, for all $d\in\Delta$ and
 all $\Imf\in\textit{Ext}(\Fmf)$,
 $\Imf\models^{\assign[x \mapsto d]} \psi$.
 Thus,
 for all $\Imf\in\textit{Ext}(\Fmf)$,
 $\Imf\models^\assign \forall x\psi$.
%
\item $\varphi =\exists x\psi$.  Suppose that
  $\Fmf\models^\assign \exists x\psi$.
  This means that there is $d\in\Delta$ such that
  $\Fmf\models^{\assign[x \mapsto d]} \psi$.  By the inductive
  hypothesis, we have that, for all $\Imf\in\textit{Ext}(\Fmf)$,
  $\Imf\models^{\assign[x \mapsto d]} \psi$.  Thus,
  for all $\Imf\in\textit{Ext}(\Fmf)$,
  $\Imf\models^\assign \exists x\psi$.
\item The other cases  can 
be proved in a straightforward way   using the 
inductive hypothesis. 
\qedhere
\end{itemize}
\end{proof}
%

\begin{claim}\label{cl:diamond2}
%
$\UntilP\forall$-formulas are $\finit{\Leftarrow}{\forall}$.
\end{claim}
\begin{proof}[Proof of Claim~\ref{cl:diamond2}]
We show the following (stronger) claim: for every
finite trace
$\Fmf = \linebreak (\Delta, (\Fmc_{n})_{n \in [0, l]})$, and every $\assign$ in
$\Fmf$,
$\Fmf^\omega \models^\assign \varphi$ implies
$\Fmf \models^\assign \varphi$.  The proof is by structural induction
on $\p$.
The proof
for the base cases
$\p = P(\bar{\tau})$ and $\p = \lnot P(\bar{\tau})$
is straightforward.
We now proceed with the inductive steps.
\begin{itemize}
  \item $\p = \psi \UntilP \chi$. If
    $\Fmf^{\omega} \models^\assign \psi \UntilP \chi$, then there is
    $n \geq 0$ such that $\Fmf^{\omega}, {n} \models^\assign \chi$
    and, for every $i \in [0, n)$,
    $\Fmf^{\omega}, {i} \models^\assign \psi$.  This means that there
    is $n \geq 0$ such that $(\Fmf^{\omega})^{n} \models^\assign \chi$
    and, for every $i \in [0, n)$,
    $(\Fmf^{\omega})^{i} \models^\assign \psi$.
    If
    $n > l$, then $(\Fmf^{\omega})^{n} = (\Fmf^{l})^{\omega}$.
    Hence, without loss of generality, we can assume that $n \leq l$,
    for which it holds by definition that
    $(\Fmf^{\omega})^{n} = (\Fmf^{n})^{\omega}$.  Thus, by the
    inductive hypothesis, we obtain $\Fmf^{n} \models^\assign \chi$,
    and, for every $i \in [0, n)$, $\Fmf^{i} \models^\assign \psi$,
    meaning that $\Fmf \models^{\assign} \psi \UntilP \chi$.
        Hence,
    $(\Fmf^{l})^{\omega} \models^\assign \chi$ and, for every
    $i \in [0, l)$, $(\Fmf^{i})^{\omega} \models^\assign \psi$.  By
    the inductive hypothesis, we obtain that
    $\Fmf^{l} \models^\assign \chi$ and, for every $i \in [0, l)$,
    $\Fmf^{i} \models^\assign \psi$, again implying that
    $\Fmf \models^{\assign} \psi \UntilP \chi$.
\item $\p = \forall x \psi$. If $\Fmf^\omega \models^\assign \forall x\psi$, then for all
  $d\in\Delta$, $\Fmf^\omega \models^{\assign[x \mapsto d]} \psi$.
  By the inductive hypothesis, for all $d\in\Delta$,
  $\Fmf \models^{\assign[x \mapsto d]} \psi$.  So,
  $\Fmf \models^\assign \forall x\psi$.
\item $\p = \exists x\psi$. If $\Fmf^\omega \models^\assign \exists x\psi$, then there
  is $d\in\Delta$ such that
  $\Fmf^\omega \models^{\assign[x \mapsto d]} \psi$.  By the inductive hypothesis, for
  some $d \in \Delta$, $\Fmf \models^{\assign[x \mapsto d]} \psi$.
  Thus, $\Fmf \models^\assign \exists x\psi$.
\item The remaining cases follow by a straightforward application 
of the inductive hypothesis.
  \qedhere
\end{itemize}
\end{proof}

We
now show the second part of
Lemma~\ref{thm:diamond}, i.e., that
    $\Until$-formulas
    are $\infinit{}{\exists}$.  In
Claim~\ref{cl:diamond3}, we show that
    $\Until$-formulas
    are $\infinit{\Rightarrow}{\exists}$.  Then,
in Claim~\ref{cl:diamond4}, we show that
$\Until$-formulas
are
$\infinit{\Leftarrow}{\exists}$. Before proving
Claim~\ref{cl:diamond3}, we show the following
statement.
%
\begin{claim}\label{lem:aux3}
Let $\varphi$ be an
  $\Until$-formula.
  For every finite
  trace $\Fmf = (\Delta, (\Fmc_{n})_{n \in [0, l]})$, every prefix
  $\Fmf' = (\Delta, (\Fmc'_{n})_{n \in [0, l']}) \in \Pre(\Fmf)$, and every
  assignment $\assign$, $\Fmf'\models^\assign \varphi$ implies
  $\Fmf\models^\assign \varphi$.
\end{claim}
\begin{proof}[Proof of Claim~\ref{lem:aux3}]
  The proof is by structural induction on $\p$.  Clearly, the
  statement holds for the base cases of $\p = P(\bar{\tau})$ and
  $\p = \lnot P(\bar{\tau})$.
We now proceed with the inductive cases.
\begin{itemize}
\item
$\p = \psi \Until \chi$.  If
      $\Fmf'\models^\assign \psi \Until \chi$, then there is
      $n \in (0, l']$ such that $\Fmf',n\models^\assign \varphi$ and,
      for every $i \in (0, n)$, $\Fmf', i \models^{\assign} \psi$.
      That is, for some $n \in (0, l']$,
      $\Fmf'^{n} \models^{\assign} \chi$ and, for every $i \in (0,n)$,
      $\Fmf'^{i} \models^{\assign} \psi$.  As $\Fmf'^{m}$ is a prefix
      of $\Fmf^{m}$, for every $m \in [0, l']$, and since $l' \leq l$,
      we have by the induction hypothesis that there exists
      $n \in (0, l]$ such that $\Fmf^{n} \models^\assign \chi$ and,
      for every $i \in (0,n)$, $\Fmf^{i} \models^\assign \psi$.
      Equivalently, for some $n \in (0, l]$,
      $\Fmf, n \models^\assign \chi$ and, for every $i \in (0,n)$,
      $\Fmf, i \models^\assign \psi$.  Thus,
      $\Fmf \models^\assign \psi \Until \chi$.
\item The other cases can be proved by straightforward applications of
  the inductive hypothesis.\qedhere
\end{itemize}
\end{proof}


We can now proceed with the following claim.
\begin{claim}\label{cl:diamond3}
%
  $\Until$-formulas are
      $\infinit{\Rightarrow}{\exists}$.
\end{claim}
\begin{proof}[Proof of Claim~\ref{cl:diamond3}]
Given
an infinite trace $\Imf = (\Delta, (\Imc_{n})_{n \in [0, \infty)})$
and an assignment $\assign$ in $\Imf$, we show that
$\Imf \models^{a} \p$ implies that there exists
$\Fmf\in \textit{Pre}(\Imf)$ such that $\Fmf\models^\assign \varphi$.
The proof is by structural induction on $\p$. 
By Proposition~\ref{thm:boolean}, the statement holds for the base
cases of $\p = P(\bar{\tau})$ and $\p = \lnot P(\bar{\tau})$.
We now proceed with the inductive steps.
%
\begin{itemize}
\item
$\p = \psi \Until \chi$.  Suppose that
      $\Imf\models^\assign \psi \Until \chi$, meaning that there
      exists $n > 0$ such that $\Imf, n \models^\assign \chi$ and, for
      every $i \in (0, n)$, $\Imf, i \models^\assign \chi$.  In other
      words, there exists $n > 0$ such that
      $\Imf^{n} \models^\assign \chi$ and, for every $i \in (0,n)$,
      $\Imf^{i} \models^\assign \psi$.  By the inductive hypothesis,
      the previous step implies that there exists $n > 0$ such that
      $\Imf^{n}_{n_{j}} \models^\assign \chi$, for some
      $n_{j} \geq n$,
      and, for every $i \in (0,n)$,
      $\Imf^{i}_{i_{k}} \models^\assign \psi$, for some
      $i_{k} \geq i$.
      For such an $n > 0$, let
      $\overline{n_{j}} = \min\{ n_{j} \mid \Imf^{n}_{n_{j}}
      \models^\assign \chi \}$ and, for every $i \in (0, n)$, let
      $\overline{i_{k}} = \min\{ i_{k} \mid \Imf^{i}_{i_{k}}
      \models^\assign \psi \}$.  In addition, let $m$ be the maximum
      among $\overline{n_{j}}$ and $\overline{i_{k}}$, for
      $i \in (0, n)$.
      We have that
      $\Imf^{n}_{\overline{n_{j}}} \in \Pre(\Imf^{n}_{m})$,
      and $\Imf^{i}_{\overline{i_{k}}} \in \Pre(\Imf^{i}_{m})$, for
      every $i \in (0, n)$. Since
      $\Imf^{n}_{\overline{n_{j}}} \models^{\assign} \chi$ and
      $\Imf^{i}_{\overline{i_{k}}} \models^{\assign} \psi$, by
      Claim~\ref{lem:aux3} we obtain that, for some $n > 0$,
      $\Imf^{n}_{m} \models^{\assign} \chi$ and
      $\Imf^{i}_{m} \models^{\assign} \psi$, for every $i \in (0,n)$.
      In conclusion, there exists $\Fmf = \Imf_{m} \in \Pre(\Imf)$
      such that $\Fmf \models^{\assign} \psi \Until \chi$.
  %
\item $\p = \exists x\psi$.
Suppose that
$\Imf\models^\assign \exists x\psi$.  This means that there is
$d\in\Delta$ such that $\Imf\models^{\assign[x \mapsto d]} \psi$.  By
the inductive hypothesis,
there is $\Fmf\in \textit{Pre}(\Imf)$ such that
$\Fmf\models^{\assign[x \mapsto d]} \psi$, for some $d \in \Delta$. 
%
So, $\Fmf\models^\assign \exists x\psi$.
\item
$\p = \psi \land \chi$.
Suppose that $\Imf\models^\assign \psi \land \chi$.  This means that
$\Imf\models^\assign \psi$ and $\Imf\models^\assign \chi$.  By the
inductive hypothesis, there are $\Fmf,\Fmf'\in \textit{Pre}(\Imf)$
such that $\Fmf\models^\assign \psi$ and $\Fmf'\models^\assign \chi$.
By definition of $\Fmf$ and $\Fmf'$, either $\Fmf'$ is a prefix of
$\Fmf$ or vice versa.  Assume without loss of generality that $\Fmf'$
is a prefix of $\Fmf$.  By Claim~\ref{lem:aux3}, if
$\Fmf'\models^\assign \chi$, then $\Fmf\models^\assign \chi$.  Then,
$\Fmf\models^\assign \psi$ and $\Fmf\models^\assign \chi$, i.e.,
$\Fmf\models^\assign \psi \wedge \chi$.
\item The remaining cases follow by a straightforward application 
of the inductive hypothesis.\qedhere
\end{itemize}
\end{proof}

We now conclude the proof of Lemma~\ref{thm:diamond} by showing the
following claim.
\begin{claim}\label{cl:diamond4}
%
  $\Until$-formulas are
      $\infinit{\Leftarrow}{\exists}$.
\end{claim}
\begin{proof}[Proof of Claim~\ref{cl:diamond4}]
Given
an infinite trace $\Imf = (\Delta, (\Imc_{n})_{n \in [0, \infty)})$
and an assignment $\assign$, we show that $\Fmf \models^{\assign} \p$,
for some $\Fmf \in \Pre(\Imf)$, implies $\Imf \models^{\assign} \p$.
The proof is by structural induction on $\p$.
By Proposition~\ref{thm:boolean}, the statement holds for the base
cases of $\p = P(\bar{\tau})$ and $\p = \lnot P(\bar{\tau})$.  We now
proceed with the inductive steps.
\begin{itemize}
\item
$\p = \psi \Until \chi$.  Suppose that there is
      $\Fmf = (\Delta, (\Fmc_{n})_{n \in [0, l]}) \in
      \textit{Pre}(\Imf)$ such that
      $\Fmf \models^\assign \psi \Until \chi$.  This means that there
      exists $n \in (0, l]$ such that $\Fmf,n \models^\assign \chi$
      and, for every $i \in (0, n)$, $\Fmf, i \models^\assign \psi$.
      In other words, there exists $n \in (0, l]$ such that
      $\Fmf^n \models^\assign \chi$ and, for every $i \in (0, n)$,
      $\Fmf^{i} \models^\assign \psi$.  By the inductive hypothesis,
      the previous step implies that, for some $n \in (0, l]$,
      $\Imf^{n} \models^\assign \chi$ and, for every $i \in (0, n)$,
      $\Imf^{i} \models^\assign \psi$.  Thus, there exists $n > 0$
      such that $\Imf,n \models^\assign \chi$ and, for every
      $i \in (0, n)$, $\Imf, i \models^{\assign} \psi$, meaning that
      $\Imf \models^\assign \psi \Until \chi$.
\item $\p = \exists x \psi$.
  Suppose that there is $\Fmf\in \textit{Pre}(\Imf)$ such that
  $\Fmf\models^\assign \exists x\psi$.  This means that there is
  $d\in\Delta$ such that $\Fmf\models^{\assign[x \mapsto d]} \psi$.
  By the inductive hypothesis,
  we obtain $\Imf \models^{\assign[x \mapsto d]} \psi$,
  for some $d \in \Delta$.  Hence,
  $\Imf\models^\assign \exists x\psi$.
\item The remaining cases follow by a straightforward application of
  the inductive hypothesis.  \qedhere
\end{itemize}
\end{proof}
\end{proof}
%
The results of Lemma~\ref{thm:diamond} are tight in the sense that we
cannot extend the grammar rule for $\Until$-formulas (and not even for
$\UntilP$-formulas) with $\forall x\varphi$ while still satisfying
$\infinit{}{\exists}$,
and we cannot extend the grammar rule for $\UntilP\forall$-formulas
with $\varphi \Until \psi$ and satisfy $\finit{}{\forall}$.  Simple
counterexamples are $\forall x \D^{+} P(x)$
and $\D\top$, which are not $\infinit{}{\exists}$ and
$\finit{}{\forall}$, respectively.  To see that
$\forall x \D^{+} P(x)$ is not $\infinit{\Rightarrow}{\exists}$,
and thus not $\infinit{}{\exists}$, consider the model given by an
infinite trace
$\Imf$
with a (countably) infinite domain
$\Delta = \{a_1, a_2,\ldots, a_n, \ldots\}$, where
the $n$-th
domain element
is in the extension of $P$ exactly at time point $n\in\mathbb{N}$,
i.e., $\Imf,n \models P(a_n)$ and $\Imf,i \not\models P(a_n)$, for any
$i\neq n$.
It can be seen that there is no finite prefix of this infinite trace
where $\forall x \D^{+} P(x)$ holds.
%
On the other hand, $\D \top$ holds in any infinite trace, but not on a
finite trace with only one time point. Thus,
$\D \top$
is not $\finit{\Leftarrow}{\forall}$, and hence not
$\finit{}{\forall}$.


We now move to the case of
$\ReleaseP\exists$-,
$\Release$-,
and
$\ReleaseP$-formulas,
by proving a result similar to
Lemma~\ref{thm:diamond}.

\begin{restatable}{lemma}{TheoremBox}\label{thm:box}
%
   $\ReleaseP\exists$-formulas are
  $\finit{}{\exists}$ and
  $\Release$-formulas are $\infinit{}{\forall}$.
  Thus,
  $\ReleaseP$-formulas are
  both
  $\finit{}{\exists}$ and
  $\infinit{}{\forall}$.
\end{restatable}
\begin{proof}
  We first show that all
  $\ReleaseP\exists$-formulas
  are $\finit{}{\exists}$.
  In Claim~\ref{cl:box1} we show that all
  $\ReleaseP\exists$-formulas
  are
  $\finit{\Rightarrow}{\exists}$.  Then, in Claim~\ref{cl:box2} we
  show that all
  $\ReleaseP\exists$-formulas
  are
  $\finit{\Leftarrow}{\exists}$ (in fact, for the
  $\finit{\Leftarrow}{\exists}$ case,
  we can also
  allow
  $\p \Release \psi$
  in the grammar).

\begin{claim}\label{cl:box1}
%
  $\ReleaseP\exists$-formulas are
      $\finit{\Rightarrow}{\exists}$.
\end{claim}
\begin{proof}[Proof of Claim~\ref{cl:box1}]
We show the stronger claim that, for every finite trace
$\Fmf = (\Delta, (\Fmc_{n})_{n \in [0, l]})$ and every assignment
$\assign$, $\Fmf\models^\assign \varphi$ implies
$\Fmf^{\omega}\models^\assign \varphi$.  The proof is by structural
induction on $\p$.
By Proposition~\ref{thm:boolean}, the statement holds for the base
cases of $\p = P(\bar{\tau})$ and $\p = \lnot P(\bar{\tau})$.
We now proceed with the inductive cases.
\begin{itemize}
\item
$\p = \psi \ReleaseP \chi$.  Suppose that
      $\Fmf \models^{\assign} \psi \ReleaseP \chi$.  This means that,
      for every $n \in [0, l]$, we have
      $\Fmf, n \models^{\assign} \chi$, or there exists $i \in [0, n)$
      such that $\Fmf, i \models^{\assign} \psi$.  That is, for every
      $n \in [0, l]$, it holds that $\Fmf^{n} \models^{\assign} \chi$,
      or there exists $i \in [0, n)$ such that
      $\Fmf^{i} \models^{\assign} \psi$.  By the inductive hypothesis,
      this implies that, for every $n \in [0,l]$, we have
      $(\Fmf^{n})^{\omega} \models^{\assign} \chi$, or there exists
      $i \in [0, n)$ such that
      $(\Fmf^{i})^{\omega} \models^{\assign} \psi$.  For every
      $m \in [0,l]$, it holds that
      $(\Fmf^{m})^{\omega} = (\Fmf^{\omega})^{m}$, thus we obtain, for
      every $n \in [0,l]$, that
      $(\Fmf^{\omega})^{n} \models^{\assign} \chi$, or there exists
      $i \in [0, n)$ such that
      $(\Fmf^{\omega})^{i} \models^{\assign} \psi$.  Moreover, since
      $(\Fmf^{\omega})^{l} = (\Fmf^{\omega})^{m}$, for every $m > l$,
      we have that $(\Fmf^{\omega})^{m} \models^{\assign} \chi$, or
      there exists $i \in [0, m)$ such that
      $(\Fmf^{\omega})^{i} \models^{\assign} \psi$.  In conclusion,
      for every $n \geq 0$, $\Fmf^{\omega}, n \models^{\assign} \chi$,
      or there exists $i \in [0, n)$ such that
      $\Fmf^{\omega}, i \models^{\assign} \psi$.  Hence,
      $\Fmf^{\omega} \models^{\assign} \psi \ReleaseP \chi$.
\item 
$\varphi =\forall x\psi$.
Suppose that $\Fmf\models^\assign \forall x\psi$. 
This means that, for all $d\in\Delta$,
$\Fmf\models^{\assign[x \mapsto d]}  \psi$. 
By the inductive hypothesis, for all $d\in\Delta$,
$\Fmf^\omega \models^{\assign[x \mapsto d]}  \psi$.
Thus, 
$\Fmf^{\omega}\models^\assign \forall x\psi$.
\item
$\varphi =\exists x\psi$.
Suppose that $\Fmf\models^\assign \exists x\psi$. 
This means that there is $d\in\Delta$ such that 
$\Fmf\models^{\assign[x \mapsto d]}  \psi$.
By the inductive hypothesis,
we have
$\Fmf^\omega\models^{\assign[x \mapsto d]}  \psi$,
for some $d\in\Delta$.
That is,
$\Fmf^{\omega}\models^\assign \exists x\psi$.
\item The other cases  can 
be proved in a straightforward way   using the 
inductive hypothesis. 
\qedhere
\end{itemize} 
\end{proof}



\begin{claim}\label{cl:box2}
  $\ReleaseP\exists$-formulas are
      $\finit{\Leftarrow}{\exists}$. 
\end{claim}
\begin{proof}[Proof of Claim~\ref{cl:box2}]
Given a finite trace $\Fmf = (\Delta, (\Fmc_{n})_{n \in [0, l]})$ and
an assignment $\assign$, we show that $\Imf \models^{\assign} \p$, for
some $\Imf \in \Ext(\Fmf)$, implies $\Fmf \models^{\assign} \p$.  The
proof is by structural induction on $\p$.
By Proposition~\ref{thm:boolean}, the statement holds for the base
cases of $\p = P(\bar{\tau})$ and $\p = \lnot P(\bar{\tau})$.
We now proceed with the inductive steps.
\begin{itemize}
\item
$\p = \psi \ReleaseP \chi$.  Suppose that
      $\Imf\models^\assign \psi \ReleaseP \chi$, for some
      $\Imf \in \Ext(\Fmf)$.  This implies that, for every $n \geq 0$,
      $\Imf,n\models^\assign \chi$, or there exists $i \in [0,n)$ such
      that $\Imf, i \models^{\assign} \psi$.  That is, for every
      $n \geq 0$, $\Imf^n \models^\assign \chi$, or there exists
      $i \in [0,n)$ such that $\Imf^{i} \models^{\assign} \psi$.  By
      the inductive hypothesis, the previous step implies, for every
      $n \in [0, l]$, that $\Fmf^{n} \models^{\assign} \chi$, or there
      exists $i \in [0, n)$ such that
      $\Fmf^{i} \models^{\assign} \psi$, Equivalently, for every
      $n \in [0, l]$, $\Fmf, n \models^{\assign} \chi$, or there
      exists $i \in [0, n)$ such that
      $\Fmf, i \models^{\assign} \psi$.  Thus,
      $\Fmf \models^\assign \psi \ReleaseP \chi$.
  %
\item $\p = \forall x\psi$.  Suppose that there is
  $\Imf\in\textit{Ext}(\Fmf)$ such that
  $\Imf\models^\assign \forall x\psi$.  This means that, for all
  $d\in\Delta$, $\Imf\models^{\assign[x \mapsto d]} \psi$.
  By the
  inductive hypothesis, we have
  $\Fmf\models^{\assign[x \mapsto d]} \psi$, for all $d\in\Delta$.
 Hence,  $\Fmf\models^\assign \forall x\psi$.  
\item $\p = \exists x\psi$.  Suppose that there is
  $\Imf\in\textit{Ext}(\Fmf)$ such that
  $\Imf\models^\assign \exists x\psi$.  This means that there is
  $d\in\Delta$ such that $\Imf\models^{\assign[x \mapsto d]} \psi$.
  By the
  inductive hypothesis, we have
  $\Fmf\models^{\assign[x \mapsto d]} \psi$, for some $d\in\Delta$,
  i.e., $\Fmf\models^\assign \exists x\psi$.
\item The remaining cases are a straightforward application 
of the inductive hypothesis.
\qedhere
\end{itemize}
\end{proof}


We
now show the second part of
Lemma~\ref{thm:box}, i.e., that $\Release$-formulas are
$\infinit{}{\forall}$.  In Claim~\ref{cl:box3}, we show that all
$\Release$-formulas
are
$\infinit{\Rightarrow}{\forall}$.  Then, in Claim~\ref{cl:box4}, we
show that all
$\Release$-formulas
are
$\infinit{\Leftarrow}{\forall}$. Before proving Claim~\ref{cl:box3} we
show the following claim.

%
\begin{claim}\label{lem:aux1}
  Let $\varphi$ be an
  $\Release$-formula.
  For every finite
  trace $\Fmf = (\Delta, (\Fmc_{n})_{n \in [0, l]})$, every prefix
  $\Fmf' = (\Delta, (\Fmc_{n})_{n \in [0, l']}) \in \Pre(\Fmf)$, and
  every assignment $\assign$, $\Fmf'\not\models^\assign \varphi$
  implies $\Fmf \not\models^\assign \varphi$.
\end{claim}
\begin{proof}[Proof of Claim~\ref{lem:aux1}]
  The proof is by structural induction on $\p$.
  Clearly, the
  statement holds for the base cases of $\p = P(\bar{\tau})$ and
  $\p = \lnot P(\bar{\tau})$.
We now proceed with the inductive cases.
%
\begin{itemize}
\item
$\p = \psi \Release \chi$.  Suppose that
      $\Fmf'\not\models^\assign \psi \Release \chi$.  This means that
      there exists $n \in (0, l']$ such that
      $\Fmf',n \not\models^\assign \chi$, and, for every
      $i \in (0, n)$, $\Fmf',n \not\models^\assign \psi$, i.e., for
      some $n > 0$, $\Fmf'^{n} \not\models^\assign \chi$, and, for
      every $i \in (0, n)$, $\Fmf'^{i} \not\models^\assign \psi$.  As
      $\Fmf'^{m}$ is a prefix of $\Fmf^{m}$, for every $m \in [0,l']$,
      and since $l' \leq l$, by the inductive hypothesis we obtain
      that, for some $n \in (0, l]$,
      $\Fmf^{n} \not\models^\assign \chi$, and, for every
      $i \in (0, n)$, $\Fmf^{i} \not\models^\assign \psi$,
      Equivalently, there exists $n \in (0, l]$ such that
      $\Fmf, n \not\models^\assign \chi$, and, for every
      $i \in (0, n)$, $\Fmf, i \not\models^\assign \chi$.  Thus,
      $\Fmf \not \models^\assign \psi \Release \chi$.
  %
  %
\item The other cases can be proved by straightforward applications of
  the inductive hypothesis.
  \qedhere
\end{itemize}
\end{proof}
%
%
We can now proceed with the following claim.
\begin{claim}\label{cl:box3}
$\Release$-formulas
are $\infinit{\Rightarrow}{\forall}$. 
\end{claim}
\begin{proof}[Proof of Claim~\ref{cl:box3}]
Given an infinite trace
$\Imf  = (\Delta, (\Imc_{n})_{n \in [0, \infty)})$
and an assignment $\assign$,
we show that $\Imf \models^{\assign} \p$ implies $\Fmf \models^{\assign} \p$, for every $\Fmf \in \Pre(\Imf)$.
The proof is by structural induction on $\p$.
By Proposition~\ref{thm:boolean}, we have that the claim holds for the base cases of $\p = P(\bar{\tau})$ and $\p = \neg P(\bar{\tau})$.
We now proceed with the inductive cases.
\begin{itemize}
\item
$\p = \psi \Release \chi$.
By contraposition,
suppose that
$\Fmf \not \models^\assign  \psi \Release \chi$,
for some $\Fmf  = (\Delta, \Fmc_{n \in [0, l]}) \in \Pre(\Imf)$,
meaning that
there exists $n \in (0, l]$ such that
$\Fmf^{n} \not \models^{\assign} \chi$
and, for every $i \in (0, n)$,
$\Fmf^{i} \not \models^{\assign} \psi$.
By the (contrapositive of the) inductive hypothesis, the previous step
implies that 
there exists $n \in (0, l]$ such that
$\Imf^{n} \not \models^{\assign} \chi$
and, for every $i \in (0, n)$,
$\Imf^{i} \not \models^{\assign} \psi$.
Thus, 
there exists $n > 0$ such that
$\Imf, {n} \not \models^{\assign} \chi$
and, for every $i \in (0, n)$,
$\Imf, {i} \not \models^{\assign} \psi$.
This means that 
$\Imf \not \models^{\assign} \psi \Release \chi$.
\item
$\p = \forall x\psi$.
 Suppose that  $\Imf\models^\assign \forall x\psi$.
This means that,
 for every $d\in\Delta$, $\Imf\models^{\assign[x \mapsto d]} \psi$.
 By applying the inductive hypothesis,
 we have that, for every $d\in\Delta$ and every $\Fmf \in \Pre(\Imf)$, $\Fmf\models^{\assign[x \mapsto d]} \psi$.
 Thus, 
 $\Fmf\models^\assign \forall x\psi$, for every $\Fmf \in \Pre(\Imf)$.
\item The remaining cases
follow
by a straightforward application 
of the inductive hypothesis.
\qedhere
\end{itemize} 
\end{proof}

We now conclude the proof of Lemma~\ref{thm:box} by showing the
following claim.
\begin{claim}\label{cl:box4}
  $\Release$-formulas
  are
  $\infinit{\Leftarrow}{\forall}$.
\end{claim}
\begin{proof}[Proof of Claim~\ref{cl:box4}]
Given an infinite trace
$\Imf = (\Delta, (\Imc_{n})_{n \in [0, \infty)})$ and an assignment
$\assign$, we show that $\Fmf \models^{\assign} \p$, for every
$\Fmf \in \Pre(\Imf)$, implies $\Imf \models^{\assign} \p$.  The proof
is by structural induction on $\p$.  By Proposition~\ref{thm:boolean},
we have that the claim holds for the base cases of
$\p = P(\bar{\tau})$ and $\p = \neg P(\bar{\tau})$. We now proceed
with the inductive steps.
\begin{itemize}
\item
$\p = \psi \Release \chi$.  By contraposition,
      suppose that $\Imf \not \models^{\assign} \psi \Release \chi$.
      This means that there exists $n > 0$ such that
      $\Imf, {n} \not \models^{\assign} \chi$ and, for every
      $i \in (0, n)$, $\Imf, {i} \not \models^{\assign} \psi$, i.e.,
      there exists $n > 0$ such that
      $\Imf^{n} \not \models^{\assign} \chi$ and, for every
      $i \in (0, n)$, $\Imf^{i} \not \models^{\assign} \psi$.  By the
      (contrapositive of the) inductive hypothesis, the previous step
      implies that there exists $n > 0$ such that
      $\Fmf^{n}_{j} \not \models^{\assign} \chi$, for some $j \geq n$,
      and, for every $i \in (0, n)$,
      $\Fmf^{i}_{i_{k}} \not \models^{\assign} \psi$, for some
      $i_{k} \geq i$.  Since $\mathbb{N}$ is well-founded, we can
      assume without loss of generality that such $n_{j}$ and $i_{k}$,
      for every $i \in (0, n)$, are the minimum numbers for which the
      previous step holds.  By taking $m$ as the maximum among such
      $n_{j}$ and $i_{k}$, for every $i \in (0, n)$, since
        $\Fmf^{n}_{n_{j}}\in \Pre(\Fmf^{n}_{m})$ and every
      $\Fmf^{i}_{i_{k}} \in \Pre(\Fmf^{i}_{m})$,
      by Claim~\ref{lem:aux1} we obtain 
      that $\Fmf^{n}_{m} \not \models^{\assign} \chi$ and, for every
      $i \in (0, n)$, $\Fmf^{i}_{m} \not \models^{\assign} \psi$.
      Since $\Fmf = \Imf_{m} \in \Pre(\Imf)$, we have that there
      exists
      $\Fmf = (\Delta, (\Fmc_{n})_{n \in [0, m]}) \in \Pre(\Imf)$ such
      that, for some $n \in [0, m]$,
      $\Fmf, n \not \models^{\assign} \chi$ and, for every
      $i \in (0, n)$, $\Fmf, {i} \not \models^{\assign} \psi$, meaning
      that $\Fmf \not \models^{\assign} \psi \Release \chi$.
\item $\p = \forall x \psi$.
%
  Suppose that $\Fmf\models^\assign \forall x\psi$, for all
  $\Fmf\in \textit{Pre}(\Imf)$.  This means that, for all
  $\Fmf\in \textit{Pre}(\Imf)$ and all $d\in\Delta$,
  $\Fmf\models^{\assign[x \mapsto d]} \psi$.  By applying the
  inductive hypothesis, we obtain
  $\Imf\models^{\assign[x \mapsto d]} \psi$, for all $d\in\Delta$.
 Hence, $\Imf\models^\assign \forall x\psi$. 
\item $\p = \psi\vee\chi$.  Assume that, for all
  $\Fmf\in \textit{Pre}(\Imf)$, $\Fmf\models^\assign \psi\vee\chi$,
  and suppose towards a contradiction that
  $\Imf\not\models^\assign \psi\vee\chi$, i.e.,
  $\Imf \not \models^\assign \psi$ and
  $\Imf \not \models^\assign \chi$.  By applying the (contrapositive
  of the) inductive hypothesis, we obtain that there are
  $\Fmf',\Fmf''\in \textit{Pre}(\Imf)$ such that
  $\Fmf'\not\models^\assign \psi$ and
  $\Fmf''\not\models^\assign \chi$.  By definition, either $\Fmf''$ is
  a prefix of $\Fmf'$, or vice versa.  Assume without loss of
  generality that $\Fmf''$ is a prefix of $\Fmf'$.  By
  Claim~\ref{lem:aux1}, we have that $\Fmf'\not\models^\assign \chi$.
  Hence, $\Fmf'\not\models^\assign \psi\vee\chi$, contradicting the
  assumption that, for all $\Fmf\in \textit{Pre}(\Imf)$,
  $\Fmf\models^\assign \psi\vee\chi$.  Therefore,
  $\Imf\models^\assign \psi\vee\chi$.
%
\item The remaining cases follow by a straightforward application of
  the inductive hypothesis.\qedhere
\end{itemize}
\end{proof}
\end{proof}

The
results of
Lemma~\ref{thm:box} are also tight in the sense that we cannot extend
the grammar rule for
$\Release$-formulas
(and not even for $\ReleaseP$-formulas)
with $\exists x\varphi$,
while still satisfying
$\infinit{}{\forall}$,
and we
cannot extend the grammar rule for
$\ReleaseP\exists$-formulas
with
$\p \Release\psi$,
while satisfying
$\finit{}{\exists}$.
Simple counterexamples are
$\exists x \B^{+} \neg P(x)$ and $\last := \B\bot$, which are not
$\infinit{}{\forall}$ and $\finit{}{\exists}$, respectively.  To see
that $\exists x \B^{+} \neg P(x)$ is not $\infinit{\Leftarrow}{\forall}$,
and thus not $\infinit{}{\forall}$,
consider again the model described above with an infinite (and
countable) domain, where each element is in the extension of $P$ at a
specific time point $n\in\mathbb{N}$.  The formula
$\exists x \B^{+} \neg P(x)$ holds in every finite prefix but it does
not hold in this infinite trace. Thus, it is not
$\infinit{}{\forall}$.
%
On the other hand, clearly, $\last$ holds in a finite trace \Fmf with
only one time point but it does not on any extension of $\Fmf$.
Therefore, it is not $\finit{\Rightarrow}{\exists}$, and thus not $\finit{}{\exists}$.

Finally, we comment on the results of Theorem~\ref{thm:mainequiv}.
We observe that $\Diamond \top$ and $\top$ are examples of
$\Until$-formulas that are equivalent on infinite, but not on finite,
traces. Similarly, $\Box \bot$ and $\bot$ are $\Release$-formulas
equivalent on infinite traces only. Thus, the converse of Point~(3) of
Theorem~\ref{thm:mainequiv} does not hold for such sets of formulas.
However, we leave as an open problem to determine whether the converse
of Point~(2) in Theorem~\ref{thm:mainequiv} holds for
$\UntilP\forall$- and $\ReleaseP\exists$-formulas. We conjecture that,
for these fragments, which are in negation normal form and allow for
only one kind of reflexive temporal operator (i.e., either $\UntilP$
or $\ReleaseP$), the set of equivalent formulas on finite and infinite
traces coincide.
Finally, as
stated in Point~(1) of Theorem~\ref{thm:mainequiv},
we remark that there is no
distinction between reasoning on finite and infinite traces whenever a
formula is either an $\UntilP$- or a $\ReleaseP$-formula. As already
pointed out, however, $\D^+\B^+ P(x)$ and $\B^+\D^+ P(x)$ are only
equivalent on finite traces, and so, when considering formula
equivalences, the distinction between finite and infinite traces
cannot be blurred for the class of formulas that allow both
$\UntilP$ and $\ReleaseP$.

\subsubsection{Preserving Formula Satisfiability: From Finite to Infinite Traces}
%
%
In this section, we consider the problem of preserving satisfiability
of a $\QTLfr{\Until}{}{}$ formula $\varphi$ from finite to infinite
traces, i.e.,
under which conditions,
knowing that $\p$ is finitely satisfiable, we can conclude that $\p$ is
also satisfiable on infinite traces.
Identifying classes of formulas for which this question can be
positively answered is of interest also to develop
more
efficient automated reasoners.
Indeed, under certain conditions which guarantee that
satisfiability of a formula on finite traces implies its satisfiability on infinite ones,
solvers can simply stop trying to build the lasso of an infinite trace,
once a finite trace satisfying the formula is found.

In order to connect this problem with the results obtained in the
previous sections, we make the following observations.
%
%
First, in Theorem~\ref{theor:finsatchar}, we have seen that
$\QTLfr{\Until}{}{}$ formulas interpreted on finite traces can be
translated into equisatisfiable
formulas on infinite traces.  However, such translation is not always
needed, since for some classes of formulas
satisfiability is already preserved.
%
For instance, given
$\p \in \QTLfr{\Until}{}{}(\finit{\Rightarrow}{\exists})$,
we clearly have that, if $\p$ is
satisfiable on finite traces, then it is
satisfiable on infinite traces.
%
Moreover, the problem of preserving satisfiability from finite to
infinite traces can be seen as a special case of the problem of
preserving formula equivalences from infinite to finite ones, where we
are only interested in determining if a $\QTLfr{\Until}{}{}$ formula
$\varphi$ that is equivalent to $\bot$ on infinite traces (i.e.,
unsatisfiable on infinite traces) is also unsatisfiable on finite
traces. This is not the case in general.
  For instance, $\last$, which is equivalent to $\bot$ on
infinite traces but satisfiable on finite traces, is
a formula
for which
satisfiability is not preserved from finite to infinite traces.
Instead, from Theorem~\ref{thm:mainequiv}, we obtain in particular that,
for every
$\UntilP\forall$-
or
$\ReleaseP\exists$-formulas
satisfiability is preserved from finite to infinite traces.


However, the results of the previous section do not allow us to determine classes of formulas that involve both operators
$\UntilP$
and
$\ReleaseP$,
and for which satisfiability from finite to infinite traces is preserved.
Formulas like
$\D^+\B^+ P(x)$ and 
$\B^+ \D^+ P(x)$,
for instance,
are such that their
satisfiability is preserved from finite to infinite traces, but they do not fall in any of the fragments identified above.
%
Our aim in the rest of this section is to show  that indeed
satisfiability from finite to infinite traces 
is preserved for a larger class of formulas, introduced in the following.

\emph{$\UntilP\ReleaseP$-formulas} $\varphi,\psi$ are built according to the grammar (with
$P\in\textsf{N}_{\textsf{P}}$):
\[
P(\bar{\tau})\mid \neg
P(\bar{\tau}) \mid
\varphi \wedge \psi \mid
\varphi \vee \psi \mid
\exists x \varphi\mid \forall x \varphi \mid
\varphi \UntilP \psi
\mid \varphi \ReleaseP \psi.
\]
It can be seen that the set of
$\UntilP\ReleaseP$-formulas is just a syntactic variant (in negation
normal form) of the fragment
$\QTLfr{\UntilP\ReleaseP}{}{}$ of
$\QTLfr{\Until}{}{}$, i.e., the fragment allowing only for
$\UntilP$ and
$\ReleaseP$ as temporal operators.  A typical example of an
$\UntilP\ReleaseP$-formula, used to express properties in the context
of specification and verification of reactive systems, is $\Box^{+}
\forall x (P(x) \to \Diamond^{+} Q(x) )$~\cite{GabEtAl}.
%

We show in the following that the language generated by the grammar
rule for
$\UntilP\ReleaseP$-formulas contains only formulas whose
satisfiability on finite traces implies satisfiability on infinite
traces.  This result, formalised by the following theorem, is an
immediate consequence of Lemma~\ref{thm:boxinf} below.
\begin{theorem}
\label{thm:dbformsat}
All $\UntilP\ReleaseP$-formulas satisfiable on finite traces are
satisfiable on infinite traces.
\end{theorem}
The converse of Theorem~\ref{thm:dbformsat}, however, does not hold,
as illustrated by the next example.  Consider the
$\UntilP\ReleaseP$-formula
\begin{align}
\notag
	& \Box^+  \forall x \big( (P(x) \land \lnot Q(x) ) \lor  (Q(x) \land \lnot P(x) \big)
	\ \land \\
	& \Box^+ \forall x  \big( ( P(x) \to \Diamond^+ Q(x) ) \land
					( Q(x) \to \Diamond^+ P(x) ) \big).
\label{eq:infnotfinform}
\end{align}
We have
that~$(\ref{eq:infnotfinform})$ is satisfiable on infinite traces
only, since it requires $P(x)$ and
$Q(x)$ to alternate infinitely often.
Therefore, for $\UntilP\ReleaseP$-formulas,
satisfiability on infinite traces does not imply satisfiability on finite traces.

In order to prove Theorem~\ref{thm:dbformsat}, we introduce the
following preliminary notion.
A $\QTLfr{\Until}{}{}$ formula $\varphi$
\emph{is $\finit{}{\omega}$} iff, for all finite traces
$\Fmf$ and all assignments $\assign$, it satisfies the \emph{frozen
  trace property}:
\[
\Fmf \models^{\assign} \p \Leftrightarrow \Fmf^\omega \models^{\assign} \p.
\]
We denote by $\QTLfr{\Until}{}{}(\finit{}{\omega})$ the set of
$\QTLfr{\Until}{}{}$ formulas that are $\finit{}{\omega}$.  Clearly,
if $\p \in \QTLfr{\Until}{}{}(\finit{}{\omega})$ is satisfiable on
finite traces, then $\p$ is satisfiable on infinite traces.  Thus,
Theorem~\ref{thm:dbformsat} above is an immediate consequence of the
following lemma.
%
%
\begin{restatable}{lemma}{TheoremBoxinf}\label{thm:boxinf}
  $\UntilP\ReleaseP$-formulas
  are $\finit{}{\omega}$. 
\end{restatable}
\begin{proof}
We write $\finit{\Rightarrow}{\omega}$ and $\finit{\Leftarrow}{\omega}$
for the ``one directional'' version of $\finit{}{\omega}$.
In Claim~\ref{cl:omega1} we show that all 
$\UntilP\ReleaseP$-formulas
  are $\finit{\Rightarrow}{\omega}$.
 Then, in Claim~\ref{cl:omega2}, we show that all 
 $\UntilP\ReleaseP$-formulas
 are $\finit{\Leftarrow}{\omega}$. 

\begin{claim}\label{cl:omega1}
$\UntilP\ReleaseP$-formulas
are $\finit{\Rightarrow}{\omega}$.
\end{claim}
%


\begin{proof}[Proof of Claim~\ref{cl:omega1}]
We show, by structural induction on $\p$, that $\Fmf\models^\assign \varphi$
implies $\Fmf^{\omega}\models^\assign \varphi$, for any finite trace
$\Fmf= (\Delta, (\Fmc_{n})_{n \in [0, l]})$ and any assignment $\assign$.
The base cases of $\p = P(\bar{\tau})$ and $\p = \lnot P(\bar{\tau})$, as well as the inductive cases of $\p = \psi \land \chi$, $\p = \psi \lor \chi$, $\p = \exists x \psi$, $\p = \forall x  \psi$, and $\p = \psi \ReleaseP \chi$,
can be shown as in the proof of Claim~\ref{cl:box1}.
We now show the remaining inductive case.
\begin{itemize}
\item $\varphi = \psi\UntilP\chi$.  Suppose that
  $\Fmf\models^\assign\psi\UntilP\chi$, then there exists
  $n \in [0, l]$ such that $\Fmf, n \models^\assign \chi$ and, for
  every $i \in [0, n)$, $\Fmf, i \models^\assign \psi$.  In other
  words, there exists $n \in [0, l]$ such that
  $\Fmf^{n} \models^\assign \chi$ and, for every $i \in [0, n)$,
  $\Fmf^{i} \models^\assign \psi$.  By the inductive hypothesis, the
  previous step implies that there exists $n \in [0, l]$ such that
  $(\Fmf^{n})^\omega \models^\assign \chi$ and, for every
  $i \in [0, n)$, $(\Fmf^{i})^\omega \models^\assign \psi$.  Since,
  for every $m \in [0, l]$, we have that
  $(\Fmf^{m})^{\omega} = (\Fmf^{\omega})^{m}$, the previous step
  implies that there exists $n \geq 0$ such that
  $(\Fmf^{\omega})^{n} \models^\assign \chi$ and, for every
  $i \in [0, n)$, $(\Fmf^{\omega})^{i} \models^\assign \psi$.  In
  other words, there exists $n \geq 0$ such that
  $\Fmf^{\omega}, {n} \models^\assign \chi$ and, for every
  $i \in [0, n)$, $\Fmf^{\omega}, {i} \models^\assign \psi$, i.e.,
  $\Fmf^{\omega} \models^\assign \psi \UntilP \chi$.
\qedhere
\end{itemize}
\end{proof}

\begin{claim}\label{cl:omega2}
$\UntilP\ReleaseP$-formulas
 are $\finit{\Leftarrow}{\omega}$. 
\end{claim}
\begin{proof}[Proof of Claim~\ref{cl:omega2}]
We
  show, by structural induction on $\p$, that
$\Fmf^{\omega}\models^\assign \varphi$ implies
$\Fmf\models^\assign \varphi$, for any finite trace $\Fmf= (\Delta, (\Fmc_{n})_{n \in [0, l]})$ and any assignment $\assign$.
%
The base cases of $\p = P(\bar{\tau})$ and $\p = \lnot P(\bar{\tau})$, as well as the inductive cases of $\p = \psi \land \chi$, $\p = \psi \lor \chi$, $\p = \exists x \psi$, $\p = \forall x  \psi$, and $\p = \psi \UntilP \chi$,
can be shown as in the proof of Claim~\ref{cl:diamond2}.
We now show the remaining inductive case.
\begin{itemize}
\item
$\varphi = \psi \ReleaseP \chi$.
Suppose that
$\Fmf^\omega\models^\assign \psi \ReleaseP \chi$, then, for every $n \geq 0$,
we have $\Fmf^\omega, {n} \models^\assign \chi$ or there exists $i \in [0, n)$ such that $\Fmf^\omega, {i} \models^\assign \psi$.
Thus, in particular, 
for every $n \in [0, l]$,
$(\Fmf^{\omega})^{n} \models^\assign \chi$,
or 
there exists $i \in [0, n)$ such that $(\Fmf^\omega)^{i} \models^\assign \psi$.
Since, for every $m\in [0,l]$, 
we have that $(\Fmf^{\omega})^{m}=(\Fmf^{m})^{\omega}$,
the previous step is equivalent to:
for every $n \in [0, l]$,
$(\Fmf^{n})^{\omega} \models^\assign \chi$,
or 
there exists $i \in [0, n)$ such that $(\Fmf^{i})^{\omega} \models^\assign \psi$.
By the the inductive hypothesis,
we obtain that,
for every $n \in [0,l]$,
$\Fmf^{n} \models^\assign \chi$,
or
there exists $i \in [0, n)$ such that
$\Fmf^{i} \models^\assign \psi$.
In other words,
for every $n\in [0,l]$,
$\Fmf, n \models^\assign \psi$, 
or there exists $i \in [0, n)$ such that
$\Fmf, {i} \models^\assign \psi$,
i.e.,
$\Fmf \models^\assign \psi \ReleaseP \chi$.
\qedhere
\end{itemize}
\end{proof}
\end{proof}



 \section{Complexity of Decidable Fragments on Finite and $k$-Bounded Traces}
\label{sec:tdlf-to-tdl}

In this section, we study the complexity of the satisfiability problem
for formulas taken from well-known decidable fragments of first-order
temporal logic, ranging from the constant-free one-variable monadic, to the monadic monodic, or the
two-variables monodic, fragments (as introduced in
Section~\ref{sec:syntax}). First, we consider satisfiability on
arbitrary finite traces, showing that the complexity does not change
compared to the infinite case, i.e., it remains
$\ExpSpace$-complete. Then, we analyse the case of satisfiability on
$k$-bounded traces, proving that the complexity lowers down to
$\NExpTime$-complete. Finally, we show that these fragments
interpreted on finite traces enjoy both the bounded trace and the
bounded domain properties, that is, they are satisfiable on finite
traces iff they are satisfied on finite traces with a bounded number
of time points, and of elements in the domain, respectively, with a
bound that depends on the size of the formula.
We
conclude the section with an excursus on temporal
DLs,
by investigating the complexity of the
satisfiability problem in the temporal extension of the
DL
$\ALC$.


\subsection{Complexity Results on Finite Traces}


%

We analyse the complexity of decidable fragments of first-order
temporal logic on finite traces.  To start with, we show that
$\ExpSpace$-hardness holds already for the constant-free one-variable
monadic fragment $\QTLfr{\Until}{1,mo}{\not c}$.  This fragment can be
considered as a notational variant of the propositional language of
the two-dimensional product $\mathbf{LTL}^{f} \times \bf{S5}$, defined
similarly to the product
$\mathbf{LTL} \times \mathbf{S5}$~\cite{GabEtAl}, where
$\mathbf{LTL}^{f}$ denotes $\mathbf{LTL}$ interpreted on finite traces.
In particular, the $\bf{S5}$-modality
is replaced by the universal quantifier $\forall x$, and propositional
letters $p$ are substituted by unary predicates $P(x)$, with free
variable $x$.  The lower bound can be proved by applying similar ideas
as those used to show hardness of $\mathbf{LTL} \times \mathbf{S5}$
satisfiability.
%
\begin{proposition}
\label{thm:ltlfxs5}
$\QTLfr{\Until}{1,mo}{\not c}$ formula satisfiability on finite traces is $\ExpSpace$-hard.
\end{proposition}
\begin{proof}
The proof is an adaptation of~\cite[Theorem 5.43]{GabEtAl} to the case of $\QTLfr{\Until}{1,mo}{\not c}$ on finite traces.
A \emph{tile type} is a $4$-tuple $t = (\mathit{up}(t), \mathit{down}(t), \mathit{left}(t), \mathit{right}(t))$ of \emph{colours} (from a set that we assume to include the colour \emph{white}).
Let $\mathbb{T}$ be a finite set of tile types, with $t_{0}, t_{1} \in \mathbb{T}$, and let $n \in \mathbb{N}$, given in binary.
The \emph{$m \times 2^{n}$ corridor tiling problem} 
is the problem of deciding whether there exist $m \in \mathbb{N}$ and
a function, called \emph{tiling}, $\tau \colon m \times 2^{n}
\rightarrow \mathbb{T}$ (where $m \times 2^{n}$ denotes the set of all
pairs $(i, j) \in \mathbb{N} \times \mathbb{N}$ with $0 \leq i < m$
and $0 \leq j < 2^{n}$)
such that:
\begin{itemize}
	\item $\tau(0, 0) = t_{0}$, $\tau(m - 1, 0) = t_{1}$;
	\item $\mathit{up}(\tau(i, j)) = \mathit{down}(\tau(i, j + 1))$, for $0 \leq i < m, 0 \leq j < 2^{n} - 1$, and $\mathit{right}(\tau(i, j)) = \mathit{left}(\tau(i + 1, j))$, for $0 \leq i < m - 1, 0 \leq j < 2^{n}$;
	\item $\mathit{down}(\tau(i, 0)) = \mathit{up}(\tau(i, 2^{n} - 1)) = \mathit{white}$, for $0 \leq i < m$.
\end{itemize}
The $m \times 2^{n}$ corridor tiling problem is known to be
$\ExpSpace$-complete~\cite{Boa97}.
  In the following, we will reduce this problem to
$\QTLfr{\Until}{1,mo}{\not c}$ formula satisfiability on finite
traces.

Given a finite set of tile types $\mathbb{T}$, with $t_{0}, t_{1} \in \mathbb{T}$, and an $n \in \mathbb{N}$, our aim is to construct a $\QTLfr{\Until}{1,mo}{\not c}$ formula $\p_{n, \mathbb{T}}$ such that:
$(i)$ the length of $\p_{n, \mathbb{T}}$ is polynomial in $n$ and $| \mathbb{T} |$;
$(ii)$ $\p_{n, \mathbb{T}}$ is satisfiable on finite traces iff there exist $m \in \mathbb{N}$ and a function $\tau \colon m \times 2^{n} \rightarrow \mathbb{T}$ tiling the $m \times 2^{n}$ corridor (as described by the conditions above).

We start by taking $n$ distinct unary predicates $S_{0}, \ldots, S_{n - 1}$, and let $S_{i}^{0} = \lnot S_{i}$ and $S_{i}^{1} = S_{i}$, for $0 \leq i \leq n - 1$.
Then, we define a binary counter,
called \emph{$\sigma$-counter},
up to $2^{n}$
by setting
\[
 \sigma_{j}(x) = S_{n - 1}^{d_{n - 1}}(x) \land \ldots \land S_{0}^{d_{0}}(x),
\]
where $d_{i}$ is the $i$-th bit in the binary representation of $0 \leq j \leq 2^{n} - 1$.
Moreover, we require
\begin{equation}
	\Box^{+} \bigwedge_{0 \leq i \leq n - 1} (\forall x S_{i}(x) \lor \forall x \lnot S_{i}(x)),
	\label{eq:countval}
\end{equation}
so that, at each time point, the $\sigma$-counter value will be the
same for every element of the domain.  The following formula will be
used to set the value of the $\sigma$-counter to $0$ at the first
instant of a trace, and to increase its value by one at each
subsequent instant (if any). Once the $\sigma$-counter reaches the
value of $2^{n} - 1$, it goes back to $0$ at the following time point
(if any).
\begin{align}
	\notag
	& \sigma_{0}(x) \ \land \\
	\notag
	&\Box^{+}
		\bigwedge_{0 \leq k < n}
		\Big(
			\big(
				\bigwedge_{0 \leq i < k}
				S_{i}(x) \land \lnot S_{k}(x)
			\big)
			\to
			\big(
				\bigwedge_{k < j < n}
				(
								S_{j}(x) \leftrightarrow \Next S_{j} (x) )
			\big)
			\land
			\Next
			\big(
				\bigwedge_{0 \leq i < k}\hspace{-.5em}   \lnot S_{i}(x) \land S_{k}(x) 
			\big)
		\Big)
			\ \land \\
	& \Box^{+}
		\Big(
			\bigwedge_{0 \leq i < n} S_{i}(x)
			\to
			\Wnext
			\big(
				\bigwedge_{0 \leq i < n} \lnot S_{i}(x) 
			\big) 
		\Big).
	\label{eq:counter}
\end{align}
%
Now consider $n$ fresh unary predicates $P_{0}, \ldots, P_{n - 1}$.
By setting
\begin{equation}
\label{eq:pval}
	\Box^{+} \forall x \bigwedge_{0 \leq i < n} ( \lnot \last \to (P_{i}(x) \leftrightarrow \Next P_{i}(x) ) ),
\end{equation}
we force their extension to be fixed along the temporal dimension.
Then, we set
\[
	\pi_{j}(x) = P_{n - 1}^{d_{n - 1}}(x) \land \ldots \land P_{0}^{d_{0}}(x),
\]
where again $d_{n - 1} \ldots d_{0}$ is the binary representation of $0 \leq j \leq 2^{n} - 1$, while $P_{i}^{0} = \lnot P_{i}$ and $P_{i}^{1} = P_{i}$, for $0 \leq i \leq n - 1$.

Moreover, we define the formulas
\begin{align*}
	\mathit{equ}(x) & = \bigwedge_{0 \leq i < n} (P_{i}(x) \leftrightarrow S_{i}(x) ), \\
	\mathit{mark}(x) &  =  \bigvee_{t \in \mathbb{T}} T(x), \\
	\mathit{tile}(x) & =  \mathit{equ}(x) \land \mathit{mark}(x) \land \Box \lnot \mathit{mark}(x),
\end{align*}
where $T$ is a fresh unary predicate for each $t \in \mathbb{T}$.

Then, consider the formulas
\begin{align}
	& \D^{+} (\last \land \sigma_{2^{n} - 1}(x) ),
	\label{eq:untilform} \\
	& \mathit{tile}(x) \land \Box \exists x \, \mathit{tile}(x).
	\label{eq:tileboxmarktotile}
\end{align}
Observe that, for Formulas~$(\ref{eq:counter})$ and~$(\ref{eq:untilform})$ to be satisfied, a trace has
to be finite and so that in its last instant the value of the
$\sigma$-counter is $2^{n} - 1$.
As it will become clear below, this step differs from the proof of~\cite[Theorem 5.43]{GabEtAl}, since we exploit the last instant of a finite trace to indicate that the construction of the corridor is completed. 

Now define
\begin{equation}
\label{eq:corridor}
	\mathit{corridor}(x)
\end{equation}
as the conjunction
of~(\ref{eq:countval})-(\ref{eq:tileboxmarktotile}).
Let $\Fmf = (\Delta^{\Fmf}, (\Fmc_{n})_{n \in [0, l]})$ be a
finite trace satisfying $\mathit{corridor}(x)$.
We have that
$\Fmf, 0 \models \mathit{corridor}[d_{0}]$,
for some
$d_{0} \in \Delta^{\Fmf}$.  It can be
seen 
that this implies the existence of $m \cdot 2^{n}$ distinct elements
$d_{0}, \ldots, d_{m \cdot 2^{n} - 1}$ of $\Delta^{\Fmf}$, for some
$m \in \mathbb{N}$, such that $\Fmf, i \models \mathit{tile}[d_{i}]$,
for $0 \leq i \leq m \cdot 2^{n} - 1$.

As a next step, we set
\begin{align*}
	\mathit{up}(x) & = \Next \mathit{tile}(x),  \\
	\mathit{right}(x) & = \mathit{equ}(x) \land (\lnot \mathit{equ}(x) \Until \mathit{tile}(x)).
\end{align*}
Given a finite trace $\Fmf$ satisfying $\mathit{corridor}(x)$,
it can be seen 
that the following
hold:
\begin{itemize}
	\item for every $0 \leq i < m \cdot 2^{n} - 1$, $\Fmf, i \models \mathit{up}[d_{i + 1}]$ and $\Fmf, i \not \models \mathit{up}[d_{j}]$, for every $j \neq i + 1$;
	\item for every $0 \leq i < (m - 1) \cdot 2^{n} $, $\Fmf, i \models \mathit{right}[d_{i + 2^{n}}]$ and $\Fmf, i \not \models \mathit{right}[d_{j}]$, for every $j \neq i + 2^{n}$.
\end{itemize}

The following formula will ensure that each point of the corridor is covered by at most one tile:
\begin{equation}
\label{eq:tileeach}
	\Box^{+} \forall x \bigwedge_{\substack{t, t' \in \mathbb{T}, \\ t \neq t'}} \lnot ( T(x) \land T'(x) ).
\end{equation}
In addition, we impose that tile $t_{0}$ is put onto the point $(0,0)$ of the corridor and that tile $t_{1}$ covers $(m - 1, 0)$ by using the following formulas:
\begin{align}
\label{eq:t0}
	& T_{0}(x) \\
	& \Box^{+} \forall x \big( \sigma_{0}(x) \land \mathit{mark}(x) \land \Box \lnot \sigma_{0}(x) \to T_{1}(x) \big).
\end{align}
The condition about matching colours on adjacent sides of adjacent tiles is encoded by the formulas:
\begin{align}
\label{eq:tileupdown}
	& \Box^{+} \forall x \Big( \lnot \sigma_{2^{n} - 1}(x) \to \bigwedge_{\substack{t, t' \in \mathbb{T}, \\ \mathit{up}(t) \neq \mathit{down}(t')}} \big( T(x) \to \forall x (\mathit{up}(x) \to \Box \lnot T'(x) ) \big) \Big), \\	
	\label{eq:tileleftright}
	& \Box^{+} \forall x \Big(  \bigwedge_{\substack{t, t' \in \mathbb{T}, \\ \mathit{right}(t) \neq \mathit{left}(t')}} \big( T(x) \to \forall x (\mathit{right}(x) \to \Box \lnot T'(x) ) \big) \Big).
\end{align}
Finally, we represent as follows
that the bottom and the top side of the corridor
have to be $\mathit{white}$:
\begin{align}
 & \Box^{+} \forall x (\sigma_{0}(x) \land \mathit{mark}(x) \to \bigvee_{\substack{t \in \mathbb{T}, \\ \mathit{down}(t) = \mathit{white}}} T(x) ), \\
  & \Box^{+} \forall x (\sigma_{2^{n} - 1}(x) \land \mathit{mark}(x) \to \bigvee_{\substack{t \in \mathbb{T}, \\ \mathit{up}(t) = \mathit{white}}} T(x) ).
  \label{eq:upwhite}
\end{align}
We then define the $\QTLfr{\Until}{1,mo}{\not c}$ formula $\p_{n, \mathbb{T}} (x)$ as the conjunction of~$(\ref{eq:corridor})$-$(\ref{eq:upwhite})$.
Clearly, the length of $\p_{n, \mathbb{T}} (x)$ is polynomial in $n$ and $| \mathbb{T} |$.

Suppose that $\Fmf, 0 \models \p_{n, \mathbb{T}} [d_{0}]$, for some
$\Fmf = (\Delta^{\Fmf}, (\Fmc_{n})_{n \in [0, l]})$ and some
$d_{0} \in \Delta^{\Fmf}$.  It can be
seen 
(see also Figure~\ref{fig:nopbdp}) that there exists
$m \in \mathbb{N}$ such that the function
$\tau \colon m \times 2^{n} \to \mathbb{T}$ defined so that
$\tau(i, j) = t$ iff
$\Fmf, i \cdot 2^{n} + j \models T[d_{i \cdot 2^{n} + j}]$, for
$0 \leq i < m$ and $0 \leq j < 2^{n}$, is a tiling of the
$m \times 2^{n}$ corridor.

Conversely, if there exist $m \in \mathbb{N}$ and a function
$\tau \colon m \times 2^{n} \to \mathbb{T}$ tiling the
$m \times 2^{n}$ corridor, then we can construct
a finite trace
$\Fmf = (\Delta^{\Fmf}, (\Fmc_{n})_{n \in [0, m \cdot 2^{n} - 1]})$,
such that $\Delta^{\Fmf} = \{ d_{0}, \ldots, d_{m \cdot 2^{n} - 1} \}$
and $\Fmf, 0 \models \p_{n, \mathbb{T}}[d_{0}]$ (see also
Figure~\ref{fig:nopbdp}).
\end{proof}

%
%
%
%

The reduction given in Theorem~\ref{theor:finsatchar} allows us to
transfer $\ExpSpace$ upper bounds for the following fragments of
first-order temporal logic on infinite traces to the finite traces
case (see~\cite{HodEtAl03} and~\cite[Theorem 11.31]{GabEtAl}): the
monadic monodic fragment $\QTLfr{\Until}{mo}{\monodic}$, and the
two-variable monodic fragment $\QTLfr{\Until}{2}{\monodic}$.
\begin{proposition}
\label{thm:upper}
$\QTLfr{\Until}{mo}{\monodic}$ and $\QTLfr{\Until}{2}{\monodic}$ formula satisfiability on finite traces is in $\ExpSpace$.
\end{proposition}

Thanks to the hardness and membership results of, respectively, Propositions~\ref{thm:ltlfxs5}
and~\ref{thm:upper},
since
$\QTLfr{\Until}{1,mo}{\not c}$ is contained both in
$\QTLfr{\Until}{mo}{\monodic}$ and $\QTLfr{\Until}{1}{}$, and since
$\QTLfr{\Until}{1}{}$ is contained in $\QTLfr{\Until}{2}{\monodic}$,
we obtain the following result.

\begin{theorem}
$\QTLfr{\Until}{1,mo}{\not c}$,
$\QTLfr{\Until}{1}{}$, $\QTLfr{\Until}{mo}{\monodic}$ and $\QTLfr{\Until}{2}{\monodic}$ formula satisfiability on finite traces are $\ExpSpace$-complete problems.
\end{theorem}

\subsection{Complexity Results on $k$-Bounded Traces}
\label{sec:compkbound}

We now study satisfiability of the decidable fragments considered
above on traces
with at most $k$ time points, where $k$ is given in binary as part of
the input.
We show that in this case the complexity of the satisfiability problem
in the fragments considered in the previous section decreases from
$\ExpSpace$~to $\NExpTime$. We
start by showing the lower
bound for $\QTLfr{\Until}{1,mo}{\not c}$.
%
%
\begin{proposition}
\label{thm:nexptimelowerbound}
$\QTLfr{\Until}{1,mo}{\not c}$ formula satisfiability on $k$-bounded traces is $\NExpTime$-hard.
\end{proposition}
\begin{proof}
  The proof is an adaptation of Proposition~\ref{thm:ltlfxs5} to the case
  of $\QTLfr{\Until}{1,mo}{\not c}$ on $k$-bounded traces.
  As
  above, a \emph{tile type} is a
  $4$-tuple
  $t = (\mathit{up}(t), \mathit{down}(t), \mathit{left}(t),
  \mathit{right}(t))$ of \emph{colours}.  Let $\mathbb{T}$ be a finite
  set of tile types, with $t_{0} \in \mathbb{T}$.
For an $n \in \mathbb{N}$,
the \emph{$2^{n} \times 2^{n}$ grid tiling problem} 
is the problem of deciding whether there exists a tiling $\tau \colon 2^{n} \times 2^{n} \rightarrow \mathbb{T}$ such that:
\begin{itemize}
\item $\tau(0, 0) = t_{0}$;
\item $\mathit{up}(\tau(i, j)) = \mathit{down}(\tau(i, j + 1))$, for
  $0 \leq i < 2^{n}, 0 \leq j < 2^{n} - 1$,
\item $\mathit{right}(\tau(i, j)) = \mathit{left}(\tau(i + 1, j))$,
  for $0 \leq i < 2^{n} - 1, 0 \leq j < 2^{n}$.
\end{itemize}
The $2^{n} \times 2^{n}$ grid tiling problem
is known to be $\NExpTime$-complete~\cite{Boa97}.
  In the following, we will reduce this problem to
$\QTLfr{\Until}{1,mo}{\not c}$ formula satisfiability on $k$-bounded
traces, with
$k = 2^{2n}$.

Let $\mathbb{T}$ be a finite set of tile types, with
$t_{0} \in \mathbb{T}$, and let $n \in \mathbb{N}$.
We modify the proof of Proposition~\ref{thm:ltlfxs5} to construct a
$\QTLfr{\Until}{1,mo}{\not c}$ formula $\p_{n, \mathbb{T}}$ such that:
$(i)$ the length of $\p_{n, \mathbb{T}}$ is polynomial in $n$ and
$| \mathbb{T} |$; $(ii)$ $\p_{n, \mathbb{T}}$ is satisfiable on
$k$-bounded traces, with
$k = 2^{2n}$,
iff there exists a function
$\tau \colon 2^{n} \times 2^{n} \rightarrow \mathbb{T}$ tiling the
$2^{n} \times 2^{n}$ grid (as described by the conditions above).

Recall the definition of the $\sigma$-counter, given in the proof of
Proposition~\ref{thm:ltlfxs5}.  We introduce other $n$ distinct unary
predicates $R_{0}, \ldots, R_{n - 1}$, and let
$R_{i}^{0} = \lnot R_{i}$ and $R_{i}^{1} = R_{i}$, for
$0 \leq i \leq n - 1$.  Then, we define another binary counter, called
$\rho$-counter up to $2^{n}$ by setting
\[
 \rho_{j}(x) = R_{n - 1}^{d_{n - 1}}(x) \land \ldots \land R_{0}^{d_{0}}(x),
\]
where $d_{i}$ is the $i$-th bit in the binary representation of $0 \leq j \leq 2^{n} - 1$.
In addition, we require
\begin{equation}
	\Box^{+} \bigwedge_{0 \leq i \leq n - 1} (\forall x R_{i}(x) \lor \forall x \lnot R_{i}(x)),
	\label{eq:rhocountval}
\end{equation}
so that, at each time point, the counter value will be the same for
every element of the domain.  The following formula will set the value
of this counter to $0$ at the first instant of a trace, and increase
its value by one at each future instant where $\sigma_{0}$ holds (if
any).

\begin{align}
	\notag
	&
	\rho_{0}(x)
	\ \land
	\\
	\notag
	&
	\Box^{+}
	\bigwedge_{0 \leq k < n}
	\Big(
		\Next \lnot \sigma_{0}(x) \to \big( R_{k}(x) \leftrightarrow
		\Next
		R_{k}(x) \big)
	\Big)
	\ \land
	\\
	\notag
	&
	\Box^{+}
	\bigwedge_{0 \leq k < n}
	\Big(
		\Next \sigma_{0}(x) \to
		\Big( 
		\big(
				\bigwedge_{0 \leq i < k}  R_{i}(x) \land \lnot R_{k}(x)
			\big)
			\to
			\\
	&	
	\phantom{
		\Box^{+}
		\bigwedge_{0 \leq k < n}
		\Big(
		\Wnext \sigma_{0}(x) \to
		\Big( 
	}
			\big(
				\bigwedge_{k < j < n} 
				(
				R_{j}(x) \leftrightarrow
				\Next
				R_{j} (x) )
			\big)
			\land
			\Next
			\big(
				\bigwedge_{0 \leq i < k}  \lnot R_{i}(x) \land R_{k}(x) 
			\big)
			\Big)
	\Big).
		\label{eq:rhocounter}
\end{align}
Formula~$(\ref{eq:rhocounter})$ implies that, if $\sigma_{0}(x)$ is satisfied at a given instant, then $\rho_{j}(x)$ is satisfied as well, for some $0 \leq j < 2^{n}$.

We then modify Formula~$(\ref{eq:untilform})$ by imposing instead
\begin{align}
	& \D^{+} (\last \land \sigma_{2^{n} - 1}(x) \land \rho_{2^{n} - 1}(x) ), \label{eq:untilrhoform}
\end{align}
so to force a finite trace to be such that in its last instant both
$\sigma_{2^{n} - 1}(x)$ and $\rho_{2^{n} - 1}(x)$ hold.

Now define $\p_{n,T}$ as the conjunction of
$(\ref{eq:countval})$-$(\ref{eq:pval})$,
$(\ref{eq:tileboxmarktotile})$,
$(\ref{eq:tileeach})$, $(\ref{eq:t0})$,
$(\ref{eq:tileupdown})$, $(\ref{eq:tileleftright})$,
$(\ref{eq:rhocountval})$-$(\ref{eq:untilrhoform})$, whose length is
polynomial in $n$ and $| \mathbb{T} |$. 
%

We now show the correctness of the encoding (see also
Figure~\ref{fig:nopbdp}). Suppose that
$\Fmf, 0 \models \p_{n, \mathbb{T}} [d_{0}]$, for some
$k$-bounded trace $\Fmf = (\Delta^{\Fmf}, (\Fmc_{i})_{i \in [0, l]})$,
with
$k= 2^{2n}$ and $l < k$,
and some $d_{0} \in \Delta^{\Fmf}$.
It can be seen
that the function $\tau \colon 2^{n} \times 2^{n} \to \mathbb{T}$
defined so that $\tau(i, j) = t$ iff
$\Fmf, i \cdot 2^{n} + j \models T[d_{i \cdot 2^{n} + j}]$, for
$0 \leq i < 2^{n}$ and $0 \leq j < 2^{n}$, is a tiling of the
$2^{n} \times 2^{n}$ grid.

Conversely, if there exists a function
$\tau \colon 2^{n} \times 2^{n} \to \mathbb{T}$ tiling the
$2^{n} \times 2^{n}$ grid, then we can construct a
$k$-bounded trace
$\Fmf = (\Delta^{\Fmf}, (\Fmc_{i})_{i \in [0, k-1]})$, with
$k= 2^{2n}$, such that
$\Delta^{\Fmf} = \{ d_{0}, \ldots, d_{2^{2n} - 1} \}$ and
$\Fmf, 0 \models \p_{n, \mathbb{T}}[d_{0}]$.

\end{proof}

\begin{figure}[t]
\centering
\begin{tikzpicture}
\tikzset{
dot/.style = {draw, fill=black, circle, inner sep=0pt, outer sep=0pt, minimum size=2pt},
edot/.style = {draw, fill=white, circle, inner sep=0pt, outer sep=0pt, minimum size=2pt}
}
\node (0) at (0,-1) {};
\draw (0) node[label=center:$0$] {};

\node (1) at (1,-1) {};
\draw (1) node[label=center:$1$] {};

\node (2) at (2,-1) {};
\draw (2) node[label=center:$2$] {};

\node (3) at (3,-1) {};
\draw (3) node[label=center:$3$] {};

\node (d0) at (-1,0) {};
\draw (d0) node[label=center:$d_0$] {};

\node (d1) at (-1,1) {};
\draw (d1) node[label=center:$d_1$] {};

\node (d2) at (-1,2) {};
\draw (d2) node[label=center:$d_2$] {};

\node (d3) at (-1,3) {};
\draw (d3) node[label=center:$d_3$] {};

\node (00) at (0,0) {};
\draw (00) node[dot,  label={[label distance=-0.1cm]north:$\mathit{equ}$}, label={[label distance=-0.1cm]south:$\mathit{tile}$}] {};

\node (01) at (1,0) {};
\draw (01) node[edot] {};

\node (02) at (2,0) {};
\draw (02) node[edot, label={[label distance=-0.1cm]north:$\mathit{equ}$}] {};

\node (03) at (3,0) {};
\draw (03) node[edot] {};

\node (10) at (0,1) {};
\draw (10) node[dot, label={[label distance=-0.1cm]south:$\mathit{up}$}] {};

\node (11) at (1,1) {};
\draw (11) node[dot, label={[label distance=-0.1cm]north:$\mathit{equ}$}, label={[label distance=-0.1cm]south:$\mathit{tile}$}] {};

\node (12) at (2,1) {};
\draw (12) node[edot] {};

\node (13) at (3,1) {};
\draw (13) node[edot,  label={[label distance=-0.1cm]north:$\mathit{equ}$}] {};

\node (20) at (0,2) {};
\draw (20) node[dot,  label={[label distance=-0.1cm]north:$\mathit{equ}$}, label={[label distance=-0.1cm]south:$\mathit{right}$}] {};

\node (21) at (1,2) {};
\draw (21) node[dot,  label={[label distance=-0.1cm]south:$\mathit{up}$}] {};

\node (22) at (2,2) {};
\draw (22) node[dot, label={[label distance=-0.1cm]north:$\mathit{equ}$}, label={[label distance=-0.1cm]south:$\mathit{tile}$}] {};

\node (23) at (3,2) {};
\draw (23) node[edot] {};

\node (30) at (0,3) {};
\draw (30) node[dot] {};

\node (31) at (1,3) {};
\draw (31) node[dot,  label={[label distance=-0.1cm]north:$\mathit{equ}$}, label={[label distance=-0.1cm]south:$\mathit{right}$}] {};

\node (32) at (2,3) {};
\draw (32) node[dot, label={[label distance=-0.1cm]south:$\mathit{up}$}] {};

\node (33) at (3,3) {};
\draw (33) node[dot,  label={[label distance=-0.1cm]north:$\mathit{equ}$}, label={[label distance=-0.1cm]south:$\mathit{tile}$}] {};


\draw[dashed]  (0-0.5,0-0.5) -- (7,0-0.5);
\draw[dashed]  (1+0.5,1+0.5) -- (7,1+0.5);
\draw[dashed]  (2+0.5,3+0.5) -- (7,3+0.5);

\draw[dashed]  (6,0+0.5) -- (7,0+0.5);

\draw[dashed]  (6,0-0.5) -- (6,1+0.5);
\draw[dashed]  (7,0-0.5) -- (7,1+0.5);


\draw[-]  (0-0.5,0-0.5) -- (0+0.5,0-0.5);
\draw[-]  (0-0.5,0+0.5) -- (1+0.5,0+0.5);
\draw[-]  (1-0.5,1+0.5) -- (2+0.5,1+0.5);
\draw[-]  (2-0.5,2+0.5) -- (3+0.5,2+0.5);
\draw[-]  (3-0.5,3+0.5) -- (3+0.5,3+0.5);

\draw[-]  (5,0-0.5) -- (6,0-0.5);
\draw[-]  (5,0+0.5) -- (6,0+0.5);

\draw[-]  (5,1+0.5) -- (7,1+0.5);
\draw[-]  (6,2+0.5) -- (7,2+0.5);
\draw[-]  (6,3+0.5) -- (7,3+0.5);

\draw[-]  (0-0.5,0-0.5) -- (0-0.5,0+0.5);
\draw[-]  (0+0.5,0-0.5) -- (0+0.5,1+0.5);
\draw[-]  (1+0.5,1-0.5) -- (1+0.5,2+0.5);
\draw[-]  (2+0.5,2-0.5) -- (2+0.5,3+0.5);
\draw[-]  (3+0.5,3-0.5) -- (3+0.5,3+0.5);

\draw[-]  (5,0-0.5) -- (5,1+0.5);
\draw[-]  (6,0-0.5) -- (6,1+0.5);

\draw[-]  (6,2-0.5) -- (6,3+0.5);
\draw[-]  (7,2-0.5) -- (7,3+0.5);



\node (s0) at (0,-2) {};
\draw (s0) node[label=center:$\sigma_0$] {};

\node (s1) at (1,-2) {};
\draw (s1) node[label=center:$\sigma_1$] {};

\node (s2) at (2,-2) {};
\draw (s2) node[label=center:$\sigma_0$] {};

\node (s3) at (3,-2) {};
\draw (s3) node[label=center:$\sigma_1$] {};


\node (r0) at (0,-2.5) {};
\draw (r0) node[label=center:$\rho_0$] {};

\node (r1) at (1,-2.5) {};
\draw (r1) node[label=center:$\rho_0$] {};

\node (r2) at (2,-2.5) {};
\draw (r2) node[label=center:$\rho_1$] {};

\node (r3) at (3,-2.5) {};
\draw (r3) node[label=center:$\rho_1$] {};


\node (p0) at (-2,0) {};
\draw (p0) node[label=center:$\pi_0$] {};

\node (p1) at (-2,1) {};
\draw (p1) node[label=center:$\pi_1$] {};

\node (p2) at (-2,2) {};
\draw (p2) node[label=center:$\pi_0$] {};

\node (p3) at (-2,3) {};
\draw (p3) node[label=center:$\pi_1$] {};


\draw[->, >=stealth, shorten <= 5pt, shorten >= 5pt]  (p0) edge (d0);
\draw[->, >=stealth, shorten <= 5pt, shorten >= 5pt]  (p1) edge (d1);
\draw[->, >=stealth, shorten <= 5pt, shorten >= 5pt]  (p2) edge (d2);
\draw[->, >=stealth, shorten <= 5pt, shorten >= 5pt]  (p3) edge (d3);


\draw[->, >=stealth, shorten <= 5pt, shorten >= 5pt]  (s0) edge (0);
\draw[->, >=stealth, shorten <= 5pt, shorten >= 5pt]  (s1) edge (1);
\draw[->, >=stealth, shorten <= 5pt, shorten >= 5pt]  (s2) edge (2);
\draw[->, >=stealth, shorten <= 5pt, shorten >= 5pt]  (s3) edge (3);

%
\end{tikzpicture}
\caption{Tiling of the $2^{n} \times 2^{n}$ grid, with $n = 1$.}
\label{fig:nopbdp}
\end{figure}
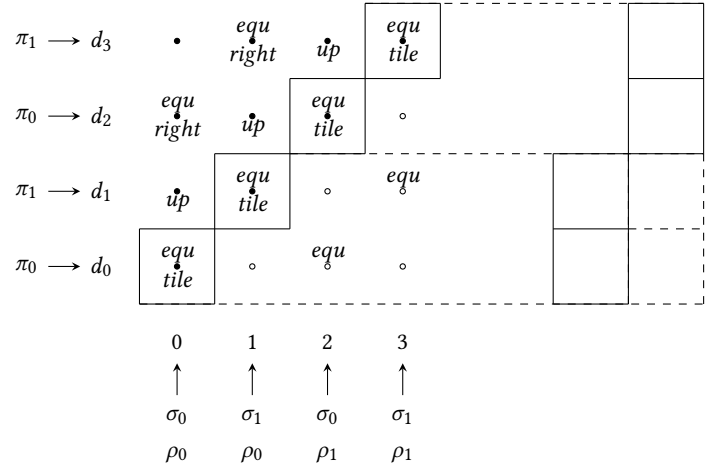

%

For the  upper bound, we resort to a classical 
abstraction of models called \emph{quasimodels}~\cite{GabEtAl}.  
One can show that there is a model with at most $k$ time
points iff there is a quasimodel with a sequence of  states (sets of 
 subformulas with certain constraints) of
length at most $k$.  Then, our upper bound is obtained by guessing an
exponential size sequence of states 
 which serves as a
certificate for the existence of a quasimodel (and therefore a model)
for the input formula.
\begin{restatable}{proposition}{TheorTALCkComplete}
  \label{th:fixe}
  $\QTLfr{\Until}{mo}{\monodic}$ and $\QTLfr{\Until}{2}{\monodic}$
  formula satisfiability on $k$-bounded traces is in $\NExpTime$.
\end{restatable}
\begin{proof}
We use standard definitions for
\emph{quasimodels}~\cite{HodEtAl2,GabEtAl}, presented here for
convenience of the reader.
In the following, with an abuse of notation, $\psi(x)$ denotes a
formula $\psi$ with \emph{at most} $x$ as free variable. Let $\varphi$
be a $\QTLfr{\Until}{}{\monodic}$ sentence\footnote{The restriction to
  sentences is without loss of generality, since a monodic formula
  $\p( x_1, \ldots, x_n )$ is equisatisfiable with the monodic
  sentence $\exists x_1 \ldots \exists x_n \p( x_1, \ldots, x_n )$.},
let $\individuals{\varphi}$ be the set of individuals occurring
in~$\varphi$, and let ${\sf sub}(\varphi)$ be the set of subformulas of
$\varphi$.
For every formula $\psi(y)$ of the form $\psi_1\Umc\psi_2 $ with one
free variable $y$, we fix a \emph{surrogate} $R_\psi(y)$; and for
every sentence $\psi$ of the form $\psi_1\Umc\psi_2 $, we fix a
surrogate $p_\psi$, where $R_\psi$ and $p_\psi$ are symbols not
occurring in $\varphi$.
Given a $\QTLfr{\Until}{}{\monodic}$ formula $\varphi$, we denote by
$\overline{\varphi}$ the result of replacing in $\varphi$ all
subformulas of the form $\psi_1\Umc\psi_2$ which are not in the scope
of any other occurrence of \Umc by their surrogates. Thus,
$\overline{\varphi}$ does not contain occurrences of temporal
operators.
%
%
Let ${\sf sub_0}(\varphi)$ be the set of all sentences in
${\sf sub}(\varphi)$.  Let $x$ be a variable not occurring in
$\varphi$, and ${\sf sub}_x(\varphi)$ be the closure under (single)
negation of all formulas 
$\psi\{x/y\}$ with $\psi(y)\in {\sf sub}(\varphi).$
A \emph{type} for $\varphi$ is a subset $t$ 
of $\{\overline{\psi}\mid \psi\in {\sf sub}_x(\varphi)\}\cup 
\individuals{\varphi}$ 
such that:
\begin{itemize}
\item $\overline{\psi_1}\wedge \overline{\psi_2}\in t$ iff
  $\overline{\psi_1}\in t$ and $\overline{\psi_2}\in t$, for every
  $\psi_1\wedge \psi_2 \in {\sf sub}_x(\varphi)$;
\item $\neg\overline{\psi}\in t$ iff $\overline{\psi}\notin t$, for
  every $\neg \psi \in {\sf sub}_x(\varphi)$; and
\item $t$ contains at most one element of $\individuals{\varphi}$.
\end{itemize}
We omit `for $\varphi$' when there is no risk of confusion.  
We say that the types $t,t'$ \emph{agree on ${\sf sub_0}(\varphi)$} if
$t\cap {\sf sub_0}(\overline{\varphi})=t'\cap {\sf
  sub_0}(\overline{\varphi})$.
Denote with $\types{\varphi}$ the set of all types for $\varphi$.  If
$a\in t\cap\individuals{\varphi}$, then $t$ ``describes'' a named
element.  We write $t^a$ to indicate this and call it a \emph{named
  type}.  A \emph{state candidate} is a subset $\Cmf$ of
$\types{\varphi}$ with only types that agree on
${\sf sub_0}(\varphi)$, containing exactly one $t^a$ for each
$a\in \individuals{\varphi}$, and such that
$\{t\setminus{\{a\}}\mid t^a\in \Cmf\}\subseteq\Cmf$.
%
Given a classical first-order interpretation $I$, let
$t^I(d)=\{\overline{\psi}\mid \psi\in {\sf sub}_x(\varphi),I\models
\overline{\psi}[d]\}$,  with $d$ in the
domain $\Delta$, and
let
$t^{I}(a)=\{\overline{\psi}\mid \psi\in {\sf sub}_x(\varphi),I\models
\overline{\psi}[a^I]\}\cup \{a\}$, with $a\in\individuals{\varphi}$.
Clearly, each $t^I(d)$ is a type for $\varphi$ while each $t^{I}(a)$
is a named type.
%
An interpretation $I$ \emph{realizes} a state canditate $\Cmf$ if
$\Cmf=\{t^I(d)\mid d \in\Delta \}\cup \{t^{I}(a)\mid
a\in\individuals{\varphi}\}$.  Vice versa, we say that $\Cmf$ is
(\emph{finitely})
\emph{realisable} if there is a
(\emph{finite})
interpretation realizing it.

Given a state candidate $\Cmf$, we define the sentence
$$\real_\Cmf=\bigwedge_{t\in \Cmf} \exists x \bigwedge_{\psi(x)\in t} 
\psi(x)\wedge\forall x \bigvee_{t\in \Cmf}\bigwedge_{\psi(x)\in t}\psi(x)\wedge\bigwedge_{t^a\in \Cmf}\bigwedge_{\psi(x)\in t^a}\psi\{a/x\}$$
%
We have that
a state candidate $\Cmf$ is 
(finitely)
realisable iff the sentence ${\sf real}_\Cmf$
is true in some
(finite)
first-order interpretation~\cite[Lemma 11.6]{GabEtAl}.
Moreover,
for
$\QTLfr{\Until}{mo}{\monodic}$ and $\QTLfr{\Until}{2}{\monodic}$,
it is known that realisability of state candidates coincides with finite realisability,
since the monadic and the $2$-variable fragments of first-order logic enjoy the
\emph{exponential} (and thus \emph{finite})
\emph{model property}~\cite[Proposition 6.2.7, Corollary 8.1.5]{BoeEtAl97}.
Finally, the following holds (for details, see the proof in~\cite[Theorem 11.31]{GabEtAl}).
%
\begin{lemma}\label{lem:aux}
  Given a state candidate $\Cmf$ for a $\QTLfr{\Until}{mo}{\monodic}$
  or a $\QTLfr{\Until}{2}{\monodic}$ sentence $\varphi$, there is an $\NExpTime$ algorithm that checks whether
    ${\sf real}_\Cmf$ is satisfiable, and thus whether \Cmf is
    realisable.
\end{lemma}
Now, consider a language
$\Lmc \in \{\QTLfr{\Until}{mo}{\monodic}, \QTLfr{\Until}{2}{\monodic}
\}$.
A \emph{quasimodel} for an $\Lmc$
sentence $\varphi$ is a pair $(S,\Rmf)$ where $S$ is a finite sequence
$(S(0),\ldots,S(n))$
of 
realizable state candidates $S(i)$, and \Rmf is a set of functions $r$
from $\{0,\ldots,n\}$ to
    $\bigcup_{0 \leq i \leq n} S(i)$,
    called \emph{runs}, mapping
each $i\in\{0,\ldots,n\}$ to a type in $S(i)$, and satisfying the
following conditions:
\begin{enumerate}
\item for every $\psi_1\Umc\psi_2\in {\sf sub}_x(\varphi)$ and every
  $i\in [0,n]$, we have $\overline{\psi_1\Umc\psi_2} \in r(i)$ iff
  there is $j \in (i,n]$ such that $\overline{\psi_2}\in r(j)$ and
  $\overline{\psi_1}\in r(l)$ for all $l\in (i,j)$;
\item for every $a\in \individuals{\varphi}$, every $r\in\Rmf$ and
  every $i,j\in [0,n]$, we have $a\in r(i)$ iff $a\in r(j)$;
\item $\overline{\varphi}\in t$ for
  some $t\in S(0)$; 
  and
  %
\item for every $i \in [0,n]$ and every $t\in S(i)$ there is a run 
$r\in \Rmf$ such that $r(i)=t$. 
\end{enumerate}


Every quasimodel for~$\varphi$ describes an interpretation
satisfying~$\varphi$ and, conversely, every such interpretation can be
abstracted into a quasimodel for~$\varphi$. We formalise this
well-known fact in the next lemma, that follows from an adaptation
of~\cite[Lemma 11.22]{GabEtAl} to the case of $k$-bounded traces.
\begin{restatable}{lemma}{Lemmaquasimodel}\label{lem:qm}
  Let $\p$ be an
  $\Lmc \in \{ \QTLfr{\Until}{mo}{\monodic},
  \QTLfr{\Until}{2}{\monodic} \}$
  sentence.
  There is a
  $k$-bounded trace satisfying $\varphi$
  iff there is a quasimodel for $\varphi$ with a
  sequence of quasistates of length at most $k$.
\end{restatable}
Let us now turn to our main result for $\Lmc$ formulas interpreted on
$k$-bounded traces. Since
the length of a trace
is bounded from the input, the complexity differs from the
satisfiability checking procedures on infinite or finite traces, which
are instead in $\ExpSpace$.  In particular, in the following we devise a
non-deterministic exponential time algorithm to check satisfiability
on $k$-bounded traces of an $\Lmc$ sentence $\varphi$.
%
It follows from the definition of types that the number of distinct
types for $\varphi$ is exponential in $|\varphi|$.
Thus, we first guess:
a sequence $(S(0),\ldots,S(n))$ of sets of types for $\varphi$ of
length
$n + 1 \leq k$;
and, for each type at position $i$ in this sequence, a sequence of
types of length
$n + 1$.
Denote by $\Rmf$ the set of such sequences of types.  Then, we check:
$(a)$ whether each $S(i)$ 
is a realisable state candidate, which, by Lemma~\ref{lem:aux}, can be
done in $\NExpTime$ in the size of $\varphi$; $(b)$ whether each
sequence in \Rmf satisfies conditions~(1)
and~(2); 
and $(c)$ whether $\varphi$ is in a type in $S(0)$, i.e., whether
condition~(3) holds. Condition~(4) is satisfied by definition of \Rmf.
All these conditions can be checked in non-deterministic exponential
time with respect to $|\varphi|$ and the binary size of $k$,
$|k|$. The algorithm returns `satisfiable' iff all conditions are
satisfied, thus implying
that $(S,\Rmf)$ is a quasimodel for $\varphi$.
By Lemma~\ref{lem:qm}, given an
$\Lmc \in \{ \QTLfr{\Until}{mo}{\monodic}, \QTLfr{\Until}{2}{\monodic}
\}$ formula $\varphi$, there is a finite trace satisfying $\p$ with at
most $k$ time points iff there is a quasimodel for $\p$ with at most
$k$ quasistates.  Thus, we showed Proposition~\ref{th:fixe}
illustrating a \NExpTime\ algorithm for checking the satisfiability of
formulas in
$\{\QTLfr{\Until}{mo}{\monodic}, \QTLfr{\Until}{2}{\monodic} \}$.
\qedhere

\end{proof}

We
can now state the main result of this section.
Thanks
to the lower and upper bounds shown in
Propositions~\ref{thm:nexptimelowerbound} and~\ref{th:fixe}, since
$\QTLfr{\Until}{1,mo}{\not c}$ is contained both in
$\QTLfr{\Until}{mo}{\monodic}$ and $\QTLfr{\Until}{1}{}$, and since
$\QTLfr{\Until}{1}{}$ is contained in $\QTLfr{\Until}{2}{\monodic}$,
we obtain the following complexity result.%
\begin{restatable}{theorem}{TheorTQLComplete}
\label{th:tqlkboundnexptimecomplete}
$\QTLfr{\Until}{1,mo}{\not c}$, and $\QTLfr{\Until}{1}{}$,
$\QTLfr{\Until}{mo}{\monodic}$ and $\QTLfr{\Until}{2}{\monodic}$
formula satisfiability on $k$-bounded traces are $\NExpTime$-complete
problems.
\end{restatable}
%


\subsection{Bounded Trace and Domain Properties}
\label{sec:boundprop}

In this section, we prove that
$\Lmc \in \{ \QTLfr{\Until}{mo}{\monodic}, \QTLfr{\Until}{2}{\monodic}
\}$ on finite traces enjoys two kinds of bounded model properties, 
one which bounds the domain of elements and one which bounds
the number of time points in a trace.

First, we show that an $\Lmc$ formula has the \emph{bounded trace
  property}: if it is satisfiable on finite traces, then there is a
$k$-bounded trace satisfying it,
where $k$ is
at most 
double exponential in $|\p|$.
%
\begin{theorem}\label{thm:boundtrace}
  Satisfiability of an
  $\Lmc \in \{ \QTLfr{\Until}{mo}{\monodic},
  \QTLfr{\Until}{2}{\monodic} \}$ formula $\varphi$ on finite traces
  implies
  satisfiability of $\varphi$ on $k$-bounded traces,
  with $k \leq 2^{2^{|\p|}}$.
\end{theorem}
\begin{proof}
  In order to prove the statement, we require some preliminary lemmas.
  First, we recall the following result~\cite[Lemma 11.22]{GabEtAl},
  applied to the case of finite traces.
\begin{lemma}
\label{lem:finqm}
An $\Lmc$ sentence $\p$ is satisfiable on finite traces iff there is a
quasimodel for $\p$.
\end{lemma}
Moreover, we adapt~\cite[Lemma 11.28]{GabEtAl} to the case of finite
traces, formalising it as follows.
\begin{lemma}
\label{lem:shortqm}
For every quasimodel $(S, \Rmf)$ for $\p$, there is a quasimodel
$(S', \Rmf')$ for $\p$ such that $|S'| \leq 2^{|\types{\varphi}|}$.
\end{lemma}
\begin{proof}
We first introduce the following notation.
Given a sequence $\Sigma = (\sigma_0, \sigma_1, \sigma_2, \ldots)$ and
$n \in \mathbb{N}$, we denote by $\Sigma_{n}$ and $\Sigma^{n}$ the
prefix ending at $n$ and the suffix starting at $n$ of $\Sigma$,
respectively.  Moreover, given sequences $\Sigma, \Sigma'$,
we denote by $\Sigma \cdot \Sigma'$ the concatenation of $\Sigma$ with
$\Sigma'$.  Then, we require the following claim, obtained by
rephrasing~\cite[Lemma 11.27]{GabEtAl} to our terminology.
\begin{claim}
\label{cla:shortqm}
Given a quasimodel $(S, \Rmf)$ for $\p$ such that $S(n) = S(m)$, for
some $n < m$, we have that
$(S_{n} \cdot S^{m + 1}, \Rmf_{n} \cdot \Rmf^{m + 1})$, with
$\Rmf_{n} \cdot \Rmf^{m + 1} = \{ r_{n} \cdot r'^{m + 1} \mid r, r'
\in \Rmf, r(n) = r'(m) \}$, is a quasimodel for $\p$.
\end{claim}
%
%
Now, let $(S, \Rmf)$ be a quasimodel for $\p$, with $S = (S(0), \ldots, S(l) )$.
We
notice that the number of different quasistates is
$\leq 2^{|\types{\varphi}|}$.
If, for every $n, m \in [0, l]$, we have $S(n) \neq S(m)$ (meaning
that all quasistates in $S$ are distinct), then
$|S| \leq 2^{|\types{\varphi}|}$, and thus $(S, \Rmf)$ is as required.
Otherwise, suppose that there are $n, m \in [0, l]$ such that
$S(n) \neq S(m)$.  Without loss of generality, we can assume that
$n < m$.
Since $S$ is finite, by (repeatedly) applying Claim~\ref{cla:shortqm}
to $(S, \Rmf)$, we obtain a quasimodel $(S', \Rmf')$ such that $S'$
contains only distinct quasistates, hence
$| S' | \leq 2^{|\types{\varphi}|}$.
\end{proof}
The theorem now follows from
Lemmas~\ref{lem:finqm},~\ref{lem:shortqm}, and~\ref{lem:qm}, since
$|\types{\varphi}| \leq 2^{|\p|}$.
\end{proof}

%

We now establish that an $\Lmc$ formula interpreted on
$k$-bounded traces
has the \emph{bounded domain property}, i.e., if it is satisfiable on
$k$-bounded traces,
then
it is satisfied on a trace having
finite
domain
with bounded cardinality.

 \begin{restatable}{theorem}{TheorTALCkBounded}\label{thm:bounddomain}
Satisfiability of an $\Lmc \in  \{ \QTLfr{\Until}{mo}{\monodic}, \QTLfr{\Until}{2}{\monodic} \}$
formula
$\varphi$
on
$k$-bounded traces
implies satisfiability
of $\p$
on traces having finite domain
of cardinality
at most
$|\types{\p}|^{k} \cdot 2^{|\types{\p}|}$.
\end{restatable}
\begin{proof}
We require the following preliminary definitions and result. A quasimodel $(S, \Rmf)$ for $\p$, with $S = (S(0), \ldots, S(l))$, is said to be \emph{finitary} if $S(i)$ is finitely realisable, for every $i \in [0, l]$, and $\Rmf$ is finite.
The next lemma is an adaptation of~\cite[Lemma 11.41]{GabEtAl} to the case of $k$-bounded traces.
\begin{lemma}
\label{lem:finkboundqm}
An $\Lmc$ sentence $\p$ is satisfiable on $k$-bounded traces having
finite domain iff there is a finitary quasimodel for $\p$ with a
sequence of quasistates of length at most $k$.
\end{lemma}

Now, suppose that $\p$ is satisfied on a $k$-bounded trace.  By
Lemma~\ref{lem:qm}, there is a quasimodel $(S, \Rmf)$ for $\p$, with
$S = (S(0), \ldots, S(l))$ and $l < k$.  It is known that, for
$\Lmc \in \{ \QTLfr{\Until}{mo}{\monodic}, \QTLfr{\Until}{2}{\monodic}
\}$ formulas, a state candidate is realisable iff it is finitely
realisable,
since monadic and $2$-variable first-order formulas enjoy the
exponential (hence, finite)
model property~\cite[Proposition 6.2.7, Corollary 8.1.5]{BoeEtAl97}.
Thus, we have that every $S(i)$, for $i \in [0, l]$, is finitely realisable.
Moreover, because $S$ is finite, we have that $\Rmf$, which is a set of functions from $\{ 0, \ldots, l \}$ to $\bigcup_{0 \leq i \leq l} S(i)$, is finite as well.
Therefore, $(S, \Rmf)$ is finitary and, by Lemma~\ref{lem:finkboundqm}, $\p$ is satisfiable on 
$k$-bounded traces
having finite domain.
Finally,
having recalled
that monadic and $2$-variable first-order formulas enjoy the exponential model property, we can assume without loss of generality that a first-order interpretation realising a quasistate in $S$ has domain of cardinality at most
$2^{|\types{\p}|}$.
Thus, one can adjust the construction in~\cite[Lemma 11.41]{GabEtAl} to obtain, from a finitary quasimodel $(S, \Rmf)$ for $\p$, with $|S| \leq k$, a $k$-bounded trace that satisfies $\p$ with domain
$\Delta = \{ (r, i) \mid r \in \Rmf, 0 \leq i < m_{\p} \}$,
for some $m_{\p} \leq 2^{|\types{\p}|}$.
Since $|\Rmf| \leq |\types{\p}|^{k}$, we have that $|\Delta| \leq |\types{\p}|^{k} \cdot 2^{|\types{\p}|}$.
\end{proof}

\noindent
 Since $\QTLfr{\Until}{2}{\monodic}$ satisfiability on finite traces implies satisfiability on $k$-bounded traces, for some $k>0$, 
  the $\QTLfr{\Until}{2}{\monodic}$  formula
%
\begin{equation*}
\label{eq:inf}
\vartheta_{i} = 
P(a) \land
\Box^{+} \forall x \Big(P(x){\rightarrow}  
\Next 
\Box^{+} \big(\neg P(x) \wedge \exists y ( R(x,y)\wedge P(y) ) \big) \Big),
\end{equation*}
which only admits models with an infinite domain~\cite{LutEtAl},
is unsatisfiable over finite traces.
%


\subsection{Temporal Description Logics}
\label{sec:TALC}

We conclude this section investigating the complexity of the
satisfiability problem in temporal DLs.
We consider the temporal language \TALC~\cite{GabEtAl} as a temporal
extension of the DL $\ALC$~\cite{handbookDL}.
Let $\NC, \NR \subseteq \textsf{N}_{\textsf{P}}$ be, respectively,
countably infinite and disjoint
sets of unary and binary predicates called \emph{concept} and
\emph{role names}.
%
%
A \emph{\TALC concept} is an expression of the form:
\[
  C, D ::= A \mid \lnot C \mid C \sqcap D \mid \exists R.C \mid C \U
  D,
\]
where $A \in \NC$ and $R \in \NR$.  
%
%
A \emph{\TALC axiom} is either a \emph{concept inclusion (CI)} of the
form $C\sqsubseteq D$, or an \emph{assertion}, $\alpha$, of the form
$A(a)$ or $R(a,b)$, where $C,D$ are \TALC concepts, $A\in\NC$,
$R\in\NR$, and $a,b\in\NI$.
%
\emph{\TALC formulas} have the form:
\[
  \p, \psi ::= \alpha \mid C\sqsubseteq D \mid \neg \p \mid \p \land
  \psi \mid \p \U \psi.
\]
%

The semantics of $\TALC$
is given again (with a small abuse of notation) over \emph{finite traces}
$\Fmf= (\Delta, (\Fmc_n)_{n \in [0, l]})$,
where
$l\in \mathbb{N}$, $\Delta$ is a non-empty domain, and,
for every $n\in [0,l]$,
${\Fmc_n}$
is an $\ALC$ interpretation with domain $\Delta$,
mapping
each concept name $A\in \NC$ to a subset $A^{\Fmc_n}$ of $\Delta$, each
role name $R\in\NR$ to a binary relation $R^{\Fmc_n}$ on $\Delta$, and
each individual name $a\in\NI$ to a domain element $a^{\Fmc_n}$ in
such a way that $a^{\Fmc_i} = a^{\Fmc_j}$, for all $i,j\in[0,l]$ (thus,
we just use the notation $a^\Fmc$).
  The interpretation is extended to
concepts as usual:
\begin{align*}
  (\lnot C)^{\Fmc_n} &= \Delta\setminus C^{\Fmc_n},\\
  (C\sqcap D)^{\Fmc_n} &= C^{\Fmc_n} \cap D^{\Fmc_n},\\
  (\exists R.C) ^{\Fmc_n} &= 
  							\{d\in\Delta\mid  \text{there is } e\in C ^{\Fmc_n}
                            \text{ such that } (d,e)\in R^{\Fmc_n}\},\\
  (C \U D)^{\Fmc_n} &= \{d\in\Delta\mid
    					\text{there is } m\in (n,l] \text{
                      such that } d\in D^{\Fmc_n} \text{ and } 
                      d\in C ^{\Fmc_i}, \text{ for all } i \in (n,m)\}.
\end{align*}
Given a \TALC formula $\p$, the \emph{satisfaction} of $\p$ in $\Fmf$ at time
point $n\in [0,l]$, written $\Fmf, n\models \p$, is
inductively defined as:
\[
\begin{array}{lcl}
  \Fmf, n\models C \sqsubseteq D & \text{ iff } & C^{\Fmc_n} \subseteq D^{\Fmc_n},\\
  \Fmf, n\models A(a)                    & \text{ iff } & a^{\Fmc}\in A^{\Fmc_n},  \\
  \Fmf, n\models R(a,b)                 & \text{ iff } & (a^{\Fmc},b^{\Fmc})\in R^{\Fmc_n},\\
  \Fmf, n\models \lnot \p              & \text{ iff } & \Fmf, n\not\models \p,\\
  \Fmf, n\models \p \land \psi      & \text{ iff } & \Fmf, n\models \p
                                                       \text{ and } \Fmf, n\models \psi,\\
  \Fmf, n\models \p \U \psi          & \text{ iff } &
    													\text{there is }  m \in (n,l]
                                                       \text{ such that } \Fmf, m\models\psi \text{ and } 
                                                       \Fmf,i\models\p, \text{ for all } i\in (n,m).
\end{array}
\]
We say that a \TALC formula $\p$ is \emph{satisfiable on finite traces}
if there exists a finite trace
$\Fmf$ such that
$\Fmf,0\models\p$.
If $\p$ is satisfiable on finite traces with at most $k$ instants, with $k$ given in binary, we say that $\p$ is \emph{satisfiable on $k$-bounded traces}.

Since a $\TALC$ formula can be mapped into an equisatisfiable
$\QTLfr{\Until}{2}{\monodic}$ formula~\cite{GabEtAl} we can transfer
the upper bounds of Propositions~\ref{thm:upper} and~\ref{th:fixe} to
$\TALCf$ on finite and $k$-bounded traces, respectively.  The lower
bounds can be obtained
from Propositions~\ref{thm:ltlfxs5} and~\ref{thm:nexptimelowerbound},
since $\QTLfr{\Until}{1,mo}{\not c}$ can be seen as a fragment of
$\TALC$ without role names~\cite{GabEtAl}.  Thus, the following holds.
\begin{restatable}{theorem}{TheorTALCkgComplete}\label{th:talc-fixed}
  $\TALCf$ satisfiability is \ExpSpace-complete on finite traces, and
  \NExpTime-complete on $k$-bounded traces.
\end{restatable} 
Moreover, from Theorems~\ref{thm:boundtrace}
and~\ref{thm:bounddomain}, we obtain immediately that $\TALCf$ on
finite traces has both the bounded trace and domain properties.

\noindent

We also consider the satisfiability problem on \emph{$k$-bounded}
traces of $\TALCfk$ \emph{restricted to global CIs}~\cite{LutEtAl,ArtEtAl02},
defined as the fragment of \TALCfk in which formulas can only be of
the form $\B^{+}(\Tmc) \wedge \psi$,
where $\Tmc$ is a conjunction of CIs and $\psi$ does not contain CIs.
The $\ExpTime$ upper bound we provide has a rather challenging proof
that uses a form of \emph{type
  elimination}~\cite{GabEtAl,LutEtAl,DBLP:conf/ecai/Gutierrez-Basulto16},
but in a setting where the number of time points is bounded by a
natural number $k > 0$. 
\footnote{The main challenge in solving this problem when the number of time
points is arbitrarily large, but finite, is
mostly
due to the presence
of \emph{last} sub-formulas (i.e., formulas of the form
$\Box \bot$) that can hold just in the last instant of the model.}
The
complexity is tight since satisfiability in $\ALC$ is already
$\ExpTime$-hard~\cite{handbookDL}.
%
%
%
%
%

To show the following theorem, we rely again on
quasimodels~\cite{GabEtAl}, which have been used to prove the
satisfiability of various temporal DLs.
Our definitions here are similar to those in
Section~\ref{sec:compkbound},
now adapted to temporal $\ALC$.
\begin{restatable}{theorem}{TheorTALCkgComplete}\label{th:fixed}
  \TALCfk satisfiability on $k$-bounded traces restricted to global
  CIs is \ExpTime-complete.
\end{restatable} 
\begin{proof}
  It is enough to show that satisfiability in $\TALCfk$ restricted to global CIs on $k$-bounded traces is in $\ExpTime$.
  Let $\varphi$ be a
  \TALCfk formula restricted to global CIs.  Assume without loss of
  generality that $\varphi$ does not contain abbreviations (i.e., it
  only contains the logical connectives $\neg$, $\sqcap$, $\wedge$,
  the existential quantifier $\exists$, and the temporal operator
  $\Until$), that 
  \Tmc has the equivalent form
  $\top\sqsubseteq C_\Tmc$ (where $C_\Tmc$ is of polynomial size with
  respect to the size of \Tmc), and that $\varphi$ contains at least
  one individual name (we can always add the assertion $A_a(a)$, for a
  fresh constant $a$ and a fresh concept name $A_a$).  Let
  $\individuals{\varphi}$ be the set of individuals occurring
  in~$\varphi$.  Following the notation provided
  by~\citeauthor{DBLP:conf/frocos/BaaderBKOT17}~\citeyear{DBLP:conf/frocos/BaaderBKOT17},
  denote by $\clf(\varphi)$ the closure under single negation of the
  set of all formulas occurring in~$\varphi$.  Similarly, we denote by
  $\clc(\varphi)$ the closure under single negation of the set of all
  concepts union the concepts $A_{a}, \exists R.A_{a}$, for any
  $a\in\individuals{\varphi}$ and $R$ a role occurring in $\varphi$,
  where $A_a$ is fresh.
%
  A \emph{concept type for $\varphi$} is any subset~$t$ of
  $\clc(\varphi)\cup\individuals{\varphi}$ such that:
\begin{enumerate}[label=$\mathbf{T\arabic*}$,leftmargin=*]
\item\label{ct:neg} $\neg C\in t$ iff $C\not\in t$, for all
  $\neg C\in\clc(\varphi)$;
\item\label{ct:con} $C\sqcap D\in t$ iff $C,D\in t$, for all
  $C\sqcap D\in\clc(\varphi)$; and
\item\label{ct:ind} $t$ contains at most one individual name
  in~$\individuals{\varphi}$.
\end{enumerate}

\noindent
Similarly, we define \emph{formula types} $t\subseteq\clf(\varphi)$
for $\varphi$ with the conditions:
\begin{enumerate}[label=$\mathbf{T\arabic*'}$,leftmargin=*]
\item\label{ft:neg} $\neg\chi\in t$ iff $\chi\not\in t$, for all
  $\neg\chi\in\clf(\varphi)$; and
\item\label{ft:fcon} $\chi\wedge\psi\in t$ iff $\chi,\psi\in t$, for
  all $\phi\wedge\psi\in\clf(\varphi)$.
\end{enumerate}
We omit `for $\varphi$' when there is no risk of confusion.  A concept
type describes one domain element at a single time point, while a
formula type expresses assertions or constraints on all domain elements.
%
If $a\in t\cap\individuals{\varphi}$, then $t$ describes a named
element and is denoted as $t^a$.  We denote with ${\sf tp}(\varphi)$
the set of all concept and formula types.

The next notion captures how sets of types need to be constrained so
that the DL dimension is respected.
We say that a pair of concept types $(t,t')$ is \emph{$R$-compatible}
if $\{\neg F\mid \neg \exists R.F\in t\}\subseteq t'$.  A
\emph{quasistate} for~$\varphi$ is a set
$S\subseteq {\sf tp}(\varphi)$
such that:
\begin{enumerate}[label=\textbf{Q\arabic*},leftmargin=*]
\item\label{q:fseg} $S$ contains exactly one formula type~$t_S$;
\item\label{q:ind} $S$ contains exactly one named type $t^a$ for each
  $a\in\individuals{\varphi}$;
\item\label{q:gci} 
  $C_\Tmc\in t$, for all concept types $t \in S$; 
\item\label{q:cass} for all $A(a)\in\clf(\varphi)$, we have
  $A(a)\in t_S$ iff $A\in t^a$;
\item \label{q:quasi} $t \in S$ and $\exists R.D \in t$ implies that there
  is $t' \in S$ such that $D\in t'$ and $(t,t')$ is $R$-compatible;
\item \label{q:last} for all $R(a,b)\in\clf(\varphi)$, we have
  $R(a,b)\in t_S$ iff $(t^a,t^b)$ is $R$-compatible.
\end{enumerate}
We 
notice that the critical `iff' condition
in~\ref{q:last} can be realised using the extra concepts
$A_a, \exists R.A_a$, introduced for all $a\in
\individuals{\varphi}$. In particular, for any quasistate $S$, with
$t_S\in S$, and any $R(a, b)\in \clf(\varphi)$, if
$\neg R(a, b)\in t_S$ then the pair of named types $(t^a,t^b)$ can be
made not $R$-compatible by including $\neg \exists R.A_b$ in $t^a$ and
$A_b$ in $t_b$. This trick will be used in the proof of
Lemma~\ref{cla:globalquasimod}.

A \emph{(concept/formula) run segment} for~$\varphi$ is a finite
sequence $\sigma= (\sigma(0),\dots,\sigma(n))$ composed exclusively of
concept or formula types, respectively, such that:
\begin{enumerate}[label=\textbf{R\arabic*},leftmargin=*,series=run]
\item\label{enu:run2} for all $a\in\individuals{\varphi}$ and all
  $i\in(0,n]$, we have $a\in\sigma(0)$ iff $a\in\sigma(i)$;
\item\label{enu:run3} for all $\alpha\Until\beta\in\cls(\varphi)$ and
  all $i\in [0,n]$, we have ${\alpha\Until\beta} \in\sigma(i)$
  iff 
  there is $j \in (i,n]$ such that $\beta\in\sigma(j)$ and
  $\alpha\in\sigma(m)$ for all $m\in (i,j)$,
\end{enumerate}
where \cls is either \clc or \clf (as appropriate), and~\ref{enu:run2}
does not apply to formula run segments.  
also notice
that condition~\ref{enu:run3} disallows any run to contain until
concepts/formulas in the last instant $n$\footnote{A concept/formula
  of the form $\neg(\alpha\Until\beta)$ is not to be considered an
  until concept/formula. In particular, $\Box\bot$ is allowed in the
  last time point.}.
Intuitively,
  a concept run segment describes the
temporal dimension of a single domain element, whereas a
  formula run segment describes constraints on the whole DL
  interpretation.
 
Finally, a \emph{quasimodel}
for~$\varphi = \B^{+}(\top \sqsubseteq C_\T) \wedge \psi$ is a pair
$(S,\Rmf)$, with $S$ a finite sequence of quasistates
$(S(0), \ldots, S(n))$ and $\Rmf$ a non-empty set of run segments such
that:
\begin{enumerate}[label=\textbf{\upshape{M\arabic*}},leftmargin=*]
\item\label{m:formula} $\psi\in t_{S_0}$
where
  $t_{S_0}$ is the formula type in $S(0)$;
\item\label{m:runs} for every $\sigma\in\Rmf$ and every $i\in [0,n]$,
  $\sigma(i)\in S(i)$; and, conversely, for every $t\in S(i) $, there
  is $\sigma\in\Rmf$ with $\sigma(i)=t$.
\end{enumerate}
By~\ref{m:runs} and the definition of a quasistate for $\varphi$, \Rmf
always contains exactly one formula run segment and one named run
segment for each $a\in\individuals{\varphi}$.

Every quasimodel for~$\varphi$ describes an interpretation
satisfying~$\varphi$ and, conversely, every such interpretation can be
abstracted into a quasimodel for~$\varphi$. We formalise this notion
for finite traces with the following lemma.

 
\begin{lemma}
  \label{cla:globalquasimod}
  There is a finite trace satisfying $\varphi$ with at most $k$ time
  points iff there is a quasimodel for $\varphi$ with a sequence of
  quasistates of length at most $k$.
\end{lemma}
\begin{proof}
  $(\Rightarrow)$ Assume there is a
  finite trace $\Fmf= (\Delta, (\Fmc_n)_{n \in [0, l]})$, with
  $l < k$, that satisfies $\p$, i.e., $\Fmf, 0 \mdl \p$.  Without loss
  of generality, assume that, for all $a\in\individuals{\varphi}$, we
  have $A^{\Fmc_n}_a = \{a^{\Fmc}\}$ for all $n\in [0, l]$, where
  $A_a$ are those fresh concept names we used to extend
  $\clc(\varphi)$.
  We define $(S,\Rmf)$ in the following way. First,
  for all $n \in [0, l]$, $d\in\Delta$ and $a \in\individuals{\varphi}$, we set:
  \begin{align*}
      \s_{d}(n) &= \{ C \in \clc(\varphi) \mid d \in C^{\Fmc_{n}} \}, \\
      \sigma_{a}(n) &= \{ C \in \clc(\varphi) \mid a^{\Fmc} \in C^{\Fmc_{n}} \} \cup \{ a \},    
  \end{align*}
  and let
    $\s_{\ast} = (\s_{\ast}(0), \ldots, \s_{\ast}(n))$, with $\ast \in \{ d, a \}$.
  Moreover, we set
  $\s_{\Fmf}(n) = \{ \psi \in \clf(\varphi) \mid \Fmf, n \mdl \psi
  \}$, and $\s_{\Fmf} = (\s_{\Fmf}(0), \ldots, \s_{\Fmf}(n))$.
  Finally, define
  $S(n) = \{ \s_{d}(n) \mid d \in \Delta \} \cup \{ \sigma_{a}(n) \mid a \in \individuals{\varphi} \} \cup \{ \s_{\Fmf}(n) \}$,
  and
   $\Rmf = \{ \s_{d} \mid d \in \Delta \} \cup \{ \s_{a} \mid a \in \individuals{\varphi} \} \cup \{ \s_{\Fmf} \}$.
  We
  now show that $(S,\Rmf)$ is a quasimodel for $\p$. The only critical
  point is the ($\Leftarrow$) direction of
  condition~\ref{q:last}. Here we use the concepts
  $A_a, \exists R.A_a$ introduced for each
  $a\in\individuals{\varphi}$. We show the contrapositive, i.e., if
  $\neg R(a,b)\in \s_{\Fmf}(n)$, for some $n\in[0,l]$, then,
  $(\s_a(n),\s_b(n))$ is not $R$-compatible. If
  $\neg R(a,b)\in \s_{\Fmf}(n)$, then, $\Fmf,n\models \neg
  R(a,b)$. But then, $\neg \exists R.A_b\in\s_a(n)$, and, since
  $A_b\in\s_b(n)$, $(\s_a(n),\s_b(n))$ is not $R$-compatible.
  

  $(\Leftarrow)$ Suppose there is a quasimodel $(S,\Rmf)$ for $\p$,
  with $S = (S(0), \ldots, S(l))$ and $l < k$.  Define a finite trace
  $\Fmf= (\Delta, (\Fmc_n)_{n \in [0, l]})$ as follows:
\begin{align*}
  \Delta & = \{ d_{\s} \mid \s \in \Rmf, \s \text{ concept run segment} \}; \\
  A^{\Fmc_{n}} & = \{ d_{\s} \mid A \in \s(n), \s \in \Rmf \}; \\
  R^{\Fmc_{n}} & = \{ (d_{\s}, d_{\s'}) \mid
                 (\s(n),\s'(n))~\text{are $R$-compatible} \}; \\
  a^{\Fmc} & = d_{\s}, \text{ for the unique } \s \in \Rmf \text{ with } a \in \s(0).
\end{align*}
By~\ref{q:ind},~\ref{enu:run2}, and~\ref{m:runs}, $a^{\Fmc}$ is
well-defined.  In order to show that $\Fmf$ satisfies $\p$, we first
show the following claim.
%
\begin{claim}
\label{cla:concqm}
\text{For all $\s \in \Rmf$, $C \in \clc(\varphi)$ and $n \in [0, l]$,
$C \in \s(n) \text{ iff } d_{\s} \in C^{\Fmc_{n}}$.}
\end{claim}
\begin{proof}[Proof of Claim~\ref{cla:concqm}]
%
The proof is
by induction on the structure of $C$. For
$C = A$, we have that $A \in \s(n)$ iff (by definition of $\Fmf$)
$d_{\s} \in A^{\Fmc_{n}}$. The cases $C = \lnot C_{1}$ and
$C = C_{1} \sqcap C_{2}$ are straightforward, by using \ref{ct:neg}
and \ref{ct:con}, respectively. It remains to show the following
cases.
\begin{itemize}
\item $C = \exists R.C_{1}$. $(\Rightarrow)$ If
  $\exists R.C_{1} \in \s(n)$, by~\ref{m:runs} and~\ref{q:quasi}, we
  have that there is $t' \in S(n)$ such that $C_1\in t'$ and
  $(\s(n), t')$ is $R$-compatible.  Thus, again by \ref{m:runs}, there
  is $\s' \in \Rmf$ such that $C_1\in \s'(n)$ and $(\s(n), \s'(n))$ is
  $R$-compatible.
%
  By i.h.  and definition of $\Fmf$, we have
  $d_{\s'} \in C_{1}^{\Fmc_{n}}$ and
  $(d_{\s}, d_{\s'}) \in R^{\Fmc_{n}}$, i.e.,
  $d_{\s} \in (\exists R.C_{1})^{\Fmc_{n}}$.
  $(\Leftarrow)$ Suppose that
  $d_{\s} \in (\exists R.C_{1})^{\Fmc_{n}}$.  Then $d_{\s}$ has an
  $R$-successor $d_{\s'}\in C^{\Fmc_{n}}_1$ and, by i.h.,
  $C_1\in \s'(n)$.  By definition of $R^{\Fmc_{n}}$,
  $(\s(n),\s'(n))~\text{are}~R\mbox{-}compatible$. By absurd, assume
  that $\neg\exists R.C_{1} \in \s(n)$, then by the $R$-compatibility,
  $\neg C_1\in \s'(n)$, which is a contradiction.
 %
\item $C = C_{1} \U C_{2}$. $C_{1} \U C_{2} \in \s(n)$ iff (by
  \ref{enu:run3}) there is $m \in (n, l]$ such that $C_{2} \in \s(m)$
  and for all $i \in (n, m), C_{1} \in \s(i)$. By i.h., there is
  $m \in (n, l]$ such that $d_{\s} \in C_{2}^{\Fmc_{m}}$ and
  $d_{\s} \in C_{1}^{\Fmc_{i}}$ for all $i \in (n, m)$, i.e.,
  $d_{\s} \in (C_{1} \U C_{2})^{\Fmc_{n}}$.
  \qedhere
\end{itemize}
\end{proof}
We now show that $\Fmf,0\models \p$, with
$\p = \Box^+(\top \sqsubseteq C_\T) \land \psi$. By~\ref{q:gci},
\ref{m:runs} and Claim~\eqref{cla:concqm},
$\Fmf,0\models \Box^+(\top \sqsubseteq C_\T)$. It remains to show that
$\Fmf,0\models\psi$, where $\psi$ is a (possibly temporal) Boolean
combination of assertions, denoted in the following as \emph{assertion
  formula}. Let $\s_{\Fmf}\in \Rmf$ be a formula run segment, which
by~\ref{q:fseg} and~\ref{m:runs} exists and is unique.  It is enough
to show
the following claim.
%
\begin{claim}
\label{cla:formqm}
For all assertion formulas $\psi \in \clf(\varphi)$, and
  for all $n \in [0, l]$, $\psi \in \s_{\Fmf}(n)$ iff $\Fmf, n \mdl \psi$.
\end{claim}
\begin{proof}[Proof of Claim~\ref{cla:formqm}]
The proof is by induction on $\psi$.
\begin{itemize}
\item $\psi = A(a)$. By~\ref{q:cass} and~\ref{m:runs}, $A(a)\in
  \s_{\Fmf}(n)$ iff $A\in \s_a(n)$ iff, by~\eqref{cla:concqm},
  $a^{\Fmc}\in A^{\Fmc_n}$, i.e., $\Fmf, n \mdl A(a)$.
\item $\psi = R(a,b)$. By~\ref{q:last} and~\ref{m:runs}, $R(a,b)\in
  \s_{\Fmf}(n)$ iff $(\s_a(n),\s_b(n))$ is $R$-compatible. By $\Fmf$
  construction, $\Fmf,n\models R(a,b)$.
  \item The cases $\neg \chi,~\chi\land \zeta,~\chi \U \zeta$ are
    similar to Claim~\ref{cla:concqm} by using~\ref{ft:neg},
    \ref{ft:fcon}, and~\ref{enu:run3}.
    \qedhere
\end{itemize}
\end{proof}
Therefore, by Claim~\ref{cla:formqm} and~\ref{m:formula}, we can conclude
that $\Fmf, 0 \mdl \psi$. This finishes the proof of Lemma~\ref{cla:globalquasimod}.
\qedhere
%
%
\end{proof}
%
%
  %

%
%
%

Before 
presenting our algorithm we need the following
definition.  We say that a pair $(t,t')$ of (concept/formula) types is
$\U$-\emph{compatible} if: 
\begin{itemize} 
  \item $\alpha \U \beta \in t$ iff either  $\beta\in t'$  
   or $\{\alpha,\alpha \U \beta\} \subseteq t'$, for all $\alpha\U\beta\in
   \cls(\varphi)$,
 \end{itemize}
where \cls is either \clc or \clf (as appropriate).

Our type elimination algorithm iterates over the values in $[1,k-1]$
to determine in exponential time in $|k|$, with $k$ given in binary,
the length of the sequence of quasistates of a quasimodel for
$\varphi$, if one exists.  We 
assume that $\varphi$ has
the form $\Box^+(\top \sqsubseteq C_\T)\land \psi$. For each
$l\in [1,k-1]$, the $l$-th iteration starts with
sets: $$S_0,\ldots,S_{l-1},S_{l}$$
and each $S_i$ is initially set to ${\sf
tp}(\varphi)$.  We start by exhaustively eliminating concept   types $t$ from
some $S_i$, with $i \in [0, l]$, if 
  $t$ violates one of the following
conditions:

\begin{enumerate}[label=\bf E\arabic*,leftmargin=*]
  \item \label{te1} for all $\exists R.D \in t$, there is $t'\in S_i$
    such that $D\in t'$ and $(t,t')$ is $R$-compatible;
  \item \label{te2} if $i>0$, then there is $t' \in S_{i-1}$ such that 
  $(t',t)$ is $\U$-compatible;
  \item \label{te3} if $i<l$, then there is $t' \in S_{i+1}$ such that 
  $(t,t')$ is $\U$-compatible;
  \item \label{te4} if $i=l$, then there is no $C\U D \in t$;
  \item \label{te5} $C_\T \in t$. 
\end{enumerate}
For each $a\in\individuals{\varphi}$, if $t$ is a named type $t^a$
then, in~\ref{te2} and~\ref{te3}, we further require that the
mentioned types in a $\Until$-compatible pair contain $a$.  This phase
of the algorithm stops when no further concept types can be
eliminated.  Next, for each formula type $t$, we say that a function
$f_t$, mapping each $a\in\individuals{\varphi}$ to a named type
containing $a$, is \emph{consistent with $t$} if:
(i) for all $A(a)\in \clf(\varphi)$, $A(a)\in t$ iff $A\in f_t(a)$;
and
(ii) for all $R(a,b) \in \clf(\varphi)$, $R(a,b) \in t$ iff
$(f_t(a),f_t(b))$ is $R$-compatible.  We are going to use these
functions to construct our quasimodel as follows.  We first add to
each $S_i$ all $f_t$ consistent with each formula type $t\in S_i$ such
that the image of $f_t$ is contained in $S_i$.  We then exhaustively
eliminate such functions $f_t$ 
from some $S_i$, with $i \in [0, l]$, if $f_t$ violates one of the
following conditions:
\begin{enumerate}[label=$\mathbf{E\arabic*'}$,leftmargin=*]
 %
%
\item \label{te22} if $i<l$, then there is $f_{t'} \in S_{i+1}$ such that
  $(t,t')$ is $\U$-compatible and, for all
  $a\in\individuals{\varphi}$, $(f_t(a),f_{t'}(a))$ is
  $\U$-compatible;
\item \label{te33} if $i=l$, then there is no $\alpha\U \beta \in
  t$. 
 %
    %
\end{enumerate}
  
%
It remains to ensure that each $S_i$ contains exactly one formula type
$t_{i}$ and one named type $t^a$ for each $a\in \individuals{\varphi}$
(and no functions $f_t$).  For this choose any formula type function
$f_{t_{0}}$ in $S_0$ such that $\psi \in t_{0}$ (if one exists) and
remove formula types $t'_{0}\neq t_{0}$ from $S_0$.  Then, for each
$i\in [1,l]$, select a formula type function $f_{t_{i}}\in S_i$ such
that $({t_{{i-1}}},{t_{i}})$ 
is $\Until$-compatible and for all $a\in\individuals{\varphi}$,
$(f_{t_{{i-1}}}(a),f_{t_{i}}(a))$ is $\U$-compatible, removing formula
types $t'_{i}\neq t_{i}$ from $S_i$, where $f_{t_{i}}$ is the selected
function.  The existence of such $f_{t_{{i}}}$
is ensured by~\ref{te22}. 
For each selected function $f_{t_{i}}$ and each
$a\in\individuals{\varphi}$, with $i\in [1,l]$,
we remove from $S_i$ all named types $t^a$ 
such that $t^a\neq f_{t_{i}}(a)$.
We now have that each $S_i$ contains exactly one formula type $t_{i}$
and one named type $t^a$ for each $a\in \individuals{\varphi}$.  Finally,
we proceed removing all functions $f_t$.  We have thus constructed a
sequence of quasistates.  Until concepts/formulas $\alpha\U\beta$ are
satisfied thanks to the \U-compatibility conditions and the fact that
there are no expressions of the form $\alpha\U\beta$ in
concept/formula types in the last quasistate.
 
This last step does not affect conditions~\ref{te1}-\ref{te5} (in
particular~\ref{te1}) for the remaining concept types since for each
named type there is an unnamed (concept) type which is the result of
removing the individual name from it, and if the named type was not
removed during type elimination then the corresponding unnamed type
was also not removed.
If the algorithm succeeds on these steps with a surviving concept type
$t \in S_0$ and a formula type $t_{S_0}$ in $S_0$ such that
$\psi \in t_{S_0}$ then it returns `satisfiable'.  Otherwise, it
increments $l$ or returns `unsatisfiable' if $l=k-1$ (i.e., there are
no further iterations).
\begin{lemma}
\label{cla:typelim}
The type elimination algorithm returns `satisfiable' iff there is a
quasimodel for $\varphi$.
\end{lemma}
%
\begin{proof}
  For $(\Rightarrow)$, let $S^*=S_0^*,\ldots,S_l^*$ be the result of
  the type elimination procedure.  Define $(S^*,\Rmf)$ with \Rmf as
  the set of sequences $\sigma$ of (concept/formula) types such that,
  for all $i \in [0,l]$:
\begin{enumerate}
\item $\sigma(i) \in S^*_i$, and for every $t\in S^*_i$, there is
  $\sigma\in\Rmf$ with $\sigma(i)=t$;
\item for all $a\in\individuals{\varphi}$, we have $a\in\sigma(0)$ iff
  $a\in\sigma(i)$;
\item for all $\alpha\Until\beta\in\cls(\varphi)$, we have
  $\alpha\Until\beta\in\sigma(i)$ iff there is $j \in (i,l]$ such that
  $\beta\in\sigma(j)$ and $\alpha\in\sigma(n)$ for all $n\in (i,j)$.
\end{enumerate}
where \cls is either \clc or \clf (as appropriate).
We now argue that $(S^*,\Rmf)$ is a quasimodel for $\varphi$.  We
first argue that $S^*$ is a sequence of quasistates for $\varphi$.
\ref{te1} ensures Condition~\ref{q:quasi}, while Condition~\ref{q:gci}
is 
guaranteed by condition~\ref{te5}.
For Conditions~\ref{q:cass} and~\ref{q:last}, we have the fact that
named types are taken from functions consistent with the formula
types.  The last step of our algorithm consists in eliminating formula
and named types so that we satisfy Conditions~\ref{q:fseg}
and~\ref{q:ind}.  Thus, $S^*$ is a sequence of quasistates for
$\varphi$.
Concerning the construction of $\Rmf$, Point~(2) can be enforced
thanks to our selection procedure for named types, which  
enforces $\U$-compatibility, while Point~(3) is a consequence of
\begin{itemize}
\item Conditions~\ref{te2}, \ref{te3} and \ref{te4}, for concept
  types; and
\item Conditions~\ref{te22} and~\ref{te33}, for formula types,
  together with the selection procedure.
\end{itemize}
Points~(2)-(3) coincide with conditions~\ref{enu:run2}
and~\ref{enu:run3}, so \Rmf is a set of run segments for $\varphi$.
Finally, Point~(1) ensures that condition~\ref{m:runs} holds, and
thus, when the algorithm returns `satisfiable', also
condition~\ref{m:formula} holds.  Thus, 
$(S^*,\Rmf)$ is a quasimodel for $\varphi$.

For 
the other direction
$(\Leftarrow)$, assume there is a quasimodel $(S,\Rmf)$ for
$\varphi$.
Assume $S$ is of the form $S_0 \ldots S_{l-1} S_l$, for some
$l\in [1,k-1]$.  Let $S_0^*,\ldots,S_l^*$ be the result of the type
elimination at the $l$-th iteration. Since $(S,\Rmf)$ is a
quasimodel, each concept type satisfies~\ref{te1}.  Moreover,
conditions~\ref{te2}-\ref{te5} are consequences of the existence of
run segments through each type (by~\ref{m:runs}).
Then, for all unnamed (concept) types $t$, if $t\in S_i$ then
$t\in S_i^*$, $i \in [0,l]$.  If $t$ is a formula type or a named type
then $t\in S_i$ does not necessarily imply that $t\in S_i^*$,
$i \in [0,l]$.  However, the existence of such types implies that the
algorithm should find a sequence of functions $f_{t_i}$, for
$i\in[0,l]$, satisfying~\ref{te22} and~\ref{te33} which is then used
to select formula and named types satisfying the quasimodel
conditions.  In particular, due to~\ref{m:formula}, the selection
procedure will select a function $f_{t_0}$ associated with a formula
type $t_0\in S_0^*$ containing $\psi$.  So there is a surviving
formula type in $S_0^*$ containing $\varphi$ and the algorithm returns
`satisfiable'.

%
 
This finishes the proof of Lemma~\ref{cla:typelim}.
\end{proof}

\medskip

We now argue that our type elimination algorithm runs in exponential
time.  Since there are polynomially many individuals (with respect to
the size of $\varphi$) occurring in $\varphi$, the number of functions
$f_t$ consistent with a formula type is exponential.  As the number of
(concept/formula) types is exponential the total number of functions
and types to consider is exponential.
In every step some concept type or function is eliminated
(by~\ref{te1}-\ref{te5} or by~\ref{te22}-\ref{te33},
respectively). Conditions~\ref{te1}-\ref{te5}
and~\ref{te22}-\ref{te33} can clearly be checked in exponential
time. Also, the selection procedure of functions for each $S_i$, which
determine the formula and named types in the result of the algorithm,
can also be checked in exponential time, since we can pick any
function in $S_{i+1}$ satisfying the $\Until$-compatibility relation,
which is a local condition.
As this can also be implemented in exponential time, this concludes
the proof of Theorem~\ref{th:fixed}.
\end{proof}

We leave the complexity of the satisfiability problem on \emph{finite
traces} for $\TALCf$ restricted to global CIs
as an open problem.  It is known that the complexity of the
satisfiability problem in this fragment over \emph{infinite} traces is
\ExpTime-complete~\cite{LutEtAl,DBLP:conf/frocos/BaaderBKOT17}.
However, the end of time formula $\psi_{f}$ is not expressible in this fragment. 
Thus, we cannot use the same strategy of defining 
a translation for the semantics based on infinite traces, as we 
did in Section~\ref{sec:redinf}.
%
Moreover, the upper bound in~\cite{LutEtAl} 
is based on type elimination. 
The main difficulty in devising a type elimination procedure in the case of arbitrary finite traces 
is that the number of time points 
is not fixed and the argument in~\cite{LutEtAl},
showing that there is a quasimodel iff there 
is a quasimodel $(S,\Rmf)$ such that 
 $S(i+1)\subseteq S(i)$, for all $i\geq 0$, is not applicable 
 to finite traces.
 A type with a concept equivalent to
 $\last$
 can only be in the last quasistate of the quasimodel.
 Therefore, it is not clear 
 whether one can show that if there is a quasimodel, then there is a quasimodel 
 with an exponential sequence of quasistates, as done in Theorem~\ref{th:fixed}.


\section{Applications}
\label{sec:appl}


Understanding the connections between finite and infinite traces is of interest to several applications.
In the following, we focus on planning and verification.
First, we lift to the first-order temporal logic setting the $\LTL$
notion of insensitivity to infiniteness~\cite{DegEtAl}, introduced in
the planning domain.  Then, we discuss how, in $\LTL$, the concepts of
safety, as well as impartiality and anticipation~\cite{BauEtAl}, can
be related to the semantic properties of Section~\ref{sec:fininf} for
bridging finite and infinite traces.

\subsection{Planning}
%
In automated planning, the sequence of states generated by actions is
usually finite~\cite{CerMay,BauHas,DegVar1,DegEtAl}.  To reuse
temporal logics based on infinite traces for specifying plan
constraints, one approach, developed
by~\citeauthor{DegEtAl}~\citeyear{DegEtAl} for $\LTLf$ on finite
traces, is based on the notion of \emph{insensitivity to
  infiniteness}.
%
%
%
This property is meant to capture those formulas that can be
equivalently interpreted on infinite traces, provided that, from a
certain instant, these traces satisfy an end event forever and falsify
all other atomic propositions.  The motivation for this comes from the
fact that
propositional letters represents atomic tasks/actions that cannot be
performed anymore after the end of a process.

In order to lift this notion of insensitivity to our first-order
temporal setting, and to provide a characterisation analogous to the
propositional one, we introduce the following definitions.
Let $\Fmf = (\Delta^{\Fmf}, (\Fmc_n)_{n \in [0, l]})$ be a finite
trace, and let
$\Emf = (\Delta^{\Emf}, (\Emc_{n})_{n \in [0, \infty)})$ be the
infinite trace such that $\Delta^{\Emf} = \Delta^{\Fmf}$ (we write
just $\Delta$), $a^{\Emc} = a^{\Fmc}$ for all $a \in
\NI$, 
and for all $P \in \textsf{N}_{\textsf{P}} \setminus \{ E \}$,
$P^{\Emc_{n}} = \eset$, while $E^{\Emc_{n}} = \Delta$, for any
$n \in [0, \infty)$.
The end extension (cf.~Section~\ref{sec:redinf}) of $\Fmf$ with $\Emf$,
$\Fmf\cdot_{E}\Emf$, will be called the \emph{insensitive extension
  of} $\Fmf$.  A $\QTLfr{\Until}{}{}$ formula $\p$ is \emph{insensitive to
  infiniteness} (or simply \emph{insensitive}) if, for every finite
trace $\Fmf$ and all assignments $\assign$, $\Fmf \mdl^{\assign} \p$
iff $\FEmfC \mdl^{\assign} \p$.  Clearly, all insensitive
$\QTLfr{\Until}{}{}$ formulas are also $\finit{\Rightarrow}{\exists}$.
%
Moreover, let $\Sigma$ be a finite subset of $\textsf{N}_{\textsf{P}}$ such that $E \in \Sigma$.
Assume without loss of generality that the $\QTLfr{\Until}{}{}$ formulas we mention in this
subsection have predicates in $\Sigma$.
%
Given an infinite trace
$\Imf$, the \emph{$\Sigma$-reduct of} $\Imf$
is the infinite trace $\reduct$ coinciding with $\Imf$ on $\Sigma$ and
such that  $X^{\Imc_{{n}\mid_{\Sigma}}} = \eset$, for $X \not \in \Sigma$ and 
$n \in [0, \infty)$.
Finally, recalling the definition of $\psif$, we
define $\vartheta_{f} = \psif \land \chi_{f}$, with
\begin{align*}
  \chi_{f} & = \B \forall x \forall \bar{y} (E(x) \to \hspace{-0.2cm} \bigwedge_{P \in \Sigma \setminus \{ E \}} \hspace{-0.2cm}  \lnot P(x,\bar{y}) ).
\end{align*}
%
Before
we proceed with a formal characterisation of
insensitive formulas, we require the following preliminary lemmas.

\begin{restatable}{lemma}{LemmaEndChar}\label{lemma:endchar}
  For every infinite trace $\Imf$, $\Imf \mdl \vartheta_{f}$ iff
  $\reduct = \FEmfC$, for some finite trace $\Fmf$.
\end{restatable}
\begin{proof}
($\Leftarrow$) 
If $\reduct = \FEmfC$, 
by Lemma~\ref{lemma:finchar},
$\FEmf \mdl \psif$.
Moreover, where $l$ is the last time point of $\Fmf$, we have by
definition: for all $n \in [0, l]$, $E^{\FEmc_{n}} = \emptyset$; for
all $n \in [l + 1, \infty)$, $E^{\FEmc_{n}} = \Delta$, and
$P^{\FEmc_{n}} = \eset$, for every $P \in \NPr \setminus \{E\}$.
Thus, for all $n \in (0, \infty)$, for all objects $d$ and all tuples
of objects $\bar{d}$ in $\Delta$, we have: if $d \in E^{\FEmc_{n}}$,
then $(d, \bar{d}) \notin P$, for all
$P \in \Sigma \setminus \{ E \}$.  Therefore, $\FEmf \mdl \chi_{f}$,
and hence $\FEmf \mdl \vartheta_{f}$.

($\Rightarrow$) Suppose $\Imf \mdl \vartheta_{f}$. Since in particular
it satisfies $\psi_{f}$, by Lemma~\ref{lemma:finchar}, we have that
$\Imf = \FEImf'$, for some finite trace
$\Fmf = (\Delta, (\Fmc_{n})_{n \in [0, l]})$ and some infinite trace
$\Imf'$. Thus:
\[
  E^{\FImc'_{n}} =
  \begin{cases}
    \emptyset, & \text{for all $n \in [0, l]$} \\
    \Delta, & \text{for all } n \in [l + 1, \infty) \\
  \end{cases}
\]
Since $\Imf \mdl \chi_{f}$, for all $n \in [l + 1, \infty)$, for all
objects $d$ and all tuples of objects $\bar{d}$ in $\Delta$, we have
that: if
$d \in E^{\Fmc\cdot_{E}\Imc'_{n}} = \Delta$,
then
$(d, \bar{d}) \notin P^{\Fmc\cdot_{E}\Imc'_{n}}$, for
all $P \in \Sigma \setminus \{ E \}$.  This is equivalent to
$P^{\I_{n}} = \eset$, for all $P \in \Sigma \setminus \{ E \}$ and
$n \in [l + 1, \infty)$.  Therefore, we have that $\reduct = \FEmf$.
\end{proof}

\begin{restatable}{lemma}{LemmaEndExtChar}\label{lemma:endextchar}
  Let $\p$ be a $\QTLfr{\Until}{}{}$ formula, $\Fmf$ a finite trace, and $\assign$ an assignment. We have that $\Fmf \mdl^{\assign} \p \text{ iff \,} \FEmfC \mdl^{\assign} \p\tr$.
\end{restatable}
\begin{proof}
By definition of $\FEmf$ and as a consequence of Lemma~\ref{lemma:finextchar}.
\end{proof}

We can now state the following characterisation result for insensitive formulas, which extends~\cite[Theorem 4]{DegEtAl}
to the first-order language
$\QTLfr{\Until}{}{}$.
\begin{restatable}{theorem}{TheorInsDegEtAlChar}\label{theor:inschar}
  A $\QTLfr{\Until}{}{}$ formula $\p$ is insensitive to infiniteness iff
  $\vartheta_{f} \models_{i} \p \leftrightarrow \p\tr$.
\end{restatable}
\begin{proof}
  $(\Rightarrow)$ Assume that $\p$ is insensitive. We want to prove
  that, for every infinite trace $\Imf$ and all assignments $\assign$,
  if $\Imf \mdl^{\assign} \vartheta_{f}$, then
  $\Imf \mdl^{\assign} \p \leftrightarrow \p\tr$. Suppose
  $\Imf \mdl^{\assign} \vartheta_{f}$.  By Lemma~\ref{lemma:endchar},
  $\reduct = \FEmf$, for a finite trace $\Fmf$.  Moreover, thanks to
  Lemma~\ref{lemma:endextchar}, $\Fmf \mdl^{\assign} \p$ iff
  $\FEmf \mdl^{\assign} \p\tr$.  Since $\p$ is by hypothesis
  insensitive, for every finite trace $\Fmf$ and all assignments
  $\assign$, $\Fmf \mdl^{\assign} \p$ iff $\FEmf \mdl^{\assign} \p$.
  Thus, $\FEmf \mdl^{\assign} \p$ if and only if
  $\FEmf \mdl^{\assign} \p\tr$.  That is,
  $\FEmf \mdl^{\assign} \p \leftrightarrow \p\tr$, and therefore
  $\Imf \mdl^{\assign} \p \leftrightarrow \p\tr$ (since all the
  predicates occurring in $\p,\p\tr$ are in $\Sigma$).

  $(\Leftarrow)$ Assume that
  $\vartheta_{f} \mdl \p \leftrightarrow \p\tr$. By
  Lemma~\ref{lemma:endchar}, for every infinite trace $\Imf$ and every
  assignment $\assign$, $\Imf \mdl^{\assign} \vartheta_{f}$ means that
  $\reduct = \FEmf$, for a finite trace $\Fmf$. Given our assumption,
  this implies $\FEmf \mdl^{\assign} \p \leftrightarrow \p\tr$, that
  is $\FEmf \mdl^{\assign} \p$ iff
  $\FEmf \mdl^{\assign} \p\tr$, for all assignments ${\assign}$. By
  Lemma~\ref{lemma:endextchar}, $\FEmf \mdl^{\assign} \p\tr$ if and
  only if $\Fmf \mdl^{\assign} \p$. In conclusion, we obtain that, for
  all assignments ${\assign}$, $\FEmf \mdl^{\assign} \p$ iff $\Fmf \mdl^{\assign} \p$, meaning that $\p$ is insensitive.
\end{proof}
%

We
now analyse syntactic features of insensitive
formulas.  Firstly, non-temporal $\QTLfr{\Until}{}{}$ formulas are
insensitive.  Moreover, this property is preserved under non-temporal
operators.  We generalise~\cite[Theorem 5]{DegEtAl}
in our setting
as follows.
\begin{restatable}{theorem}{TheorInsDegEtAlPres}\label{theor:boolins}
  Let $\p, \psi$ be insensitive
  formulas.  Then
  $\lnot \p$, $\exists x \p$, and $\p \land \psi$ are insensitive.
\end{restatable}%
\begin{proof}
  Let $\Fmf$ be a finite trace and $\assign$ be an assignment.  For
  $\lnot \p$, we have that $\Fmf \mdl^{\assign} \lnot \p$ iff
  $\Fmf \not \mdl^{\assign} \p$.  Since $\p$ is insensitive by
  hypothesis, this means that $\FEmf \not \mdl^{\assign} \p$.
  Therefore, $\lnot \p$ is insensitive as well.  For $\exists x \p$,
  we have that $\Fmf \mdl^{\assign} \exists x \p$ iff
  $\Fmf \mdl^{\assign[x \mapsto d]} \p$, for some $d \in \Delta$.
  Given that $\p$ is insensitive, this is equivalent to
  $\FEmf \mdl^{\assign[x \mapsto d]} \p$, for some $d \in \Delta$.
  That is, $\FEmf \mdl^{\assign} \exists x \p$, and so $\exists x \p$
  is insensitive.  For $\p \land \psi$, we have that
  $\Fmf \mdl^{\assign} \p \land \psi$ is equivalent to
  $\Fmf \mdl^{\assign} \p$ and $\Fmf \mdl^{\assign} \psi$.  Since both
  $\p$ and $\psi$ are assumed to be insensitive, the previous step is
  equivalent to: $\FEmf \mdl^{\assign} \p$ and
  $\FEmf \mdl^{\assign} \psi$, i.e.,
  $\FEmf \mdl^{\assign} \p \land \psi$.
\end{proof}
%

Concerning temporal operators, in~\cite{DegEtAl} it is shown how
several standard temporal patterns derived from the declarative
process modelling language~\textsc{declare}~\cite{AalPes} are
insensitive.  On the other hand, negation affects the insensitivity of
temporal formulas.  For instance, given a non temporal $\QTLfr{\Until}{}{}$
formula
$\psi$, we have that
$\D^{+} \psi$ is insensitive while $\D^{+} \lnot \psi$ is not.
Dually, $\B^{+} \lnot \psi$ is insensitive, while $\B^{+} \psi$ is
not.  Therefore, if a $\QTLfr{\Until}{}{}$ formula $\varphi$ is insensitive, it
cannot be concluded that formulas of the form $\D^{+} \varphi$ or
$\B^{+} \varphi$ are insensitive.

Finally,
we have that insensitivity is sufficient to ensure that if formulas
are equivalent on infinite traces, then they are equivalent on finite
traces.
\begin{restatable}{theorem}{Theoreminsonlyif}\label{thm:equivinsif}
For all insensitive
formulas $\p, \psi$, $\equivinf{\p}{\psi}$ implies $\equivfin{\p}{\psi}$.  	
\end{restatable}
\begin{proof}
  Given a finite trace $\Fmf$ and an assignment $\assign$, if
  $\Fmf \models^{\assign} \p$ then, as $\p$ is insensitive,
  $\FEmf \models^{\assign} \p$.  By assumption, $\equivinf{\p}{\psi}$,
  so $\FEmf \models^{\assign} \psi$.  As $\psi$ is insensitive,
  $\FEmf \models^{\assign} \psi$ implies
  $\Fmf \models^{\assign} \psi$.  The converse direction is obtained
  by swapping $\p$ and $\psi$.
\end{proof}

Since $\bot$ is insensitive, we obtain the following immediate
corollary of the previous result.
%
\begin{corollary}
  All insensitive formulas satisfiable on finite traces are
  satisfiable on infinite traces.
\end{corollary}
However, the converse directions of the above results do not not hold,
as witnessed, e.g., by formula~$\vartheta_{i}$
(cf.~Section~\ref{sec:boundprop}), which is trivially insensitive, but
satisfiable only on infinite traces.
We can obtain the converse directions by using our
Theorem~\ref{thm:equivconverse}.  For instance,
$\D^{+} P(x) \lor \D^{+} R(x)$ and $\D^{+} (P(x) \lor R(x))$ are
insensitive and $\infinit{}{\exists}$ formulas for which equivalence
on finite and infinite traces coincide.

\subsection{Verification}

In this section we show how our comparison between finite and infinite
traces can be related to the literature on temporal logics for
verification.  In particular, we establish connections between the
finite and infinite trace properties, introduced in
Section~\ref{sec:fininf}, and: $(i)$ the definition of safety in
$\LTL$ on infinite traces~\cite{Sis,BaiKat}; $(ii)$ maxims related to
monitoring procedures in runtime verification~\cite{BauEtAl}.

\subsubsection{Safety}
Recall that a safety property intuitively guarantees that ``bad
things'' never happen during the execution of a program.
In verification, $\LTL$ is often used as a specification language for
such properties, and the notion of safety is defined accordingly on
infinite traces~\cite{Sis,BaiKat}.  In the rest of this section, we
will thus restrict ourselves to $\LTL$.  A typical example of an
$\LTL$ formula used to specify a safety property is represented by
$\Box \lnot P$, where $P$ is an atom standing for an action or task.
Dual to safety properties are \emph{co-safety properties}, expressing
that ``good things'' will eventually happen in the execution of a
program.  The $\LTL$ formula $\Diamond P$ is a standard example of a
formula specifying a co-safety property.

To fix notions that will be used in the rest of this section, we start
by recalling the definitions of \emph{safety} and \emph{co-safety
  fragments} of $\LTL$.
%
The \emph{$\LTL$ safety formulas}~\cite{Sis} are defined as the $\LTL$
formulas obtained from $\bot$, $\top$, and \emph{literals} (i.e.,
\emph{propositional letters} $P$, or negated propositional letters
$\lnot P$), by applying conjunction $\land$, disjunction $\lor$,
strong next $\Next$, and reflexive release $\ReleaseP$ operators.
The \emph{$\LTL$ co-safety formulas}~\cite{KupVar99} are
dually defined
as those $\LTL$ formulas obtained from $\bot$, $\top$, and literals,
by applying conjunction $\land$, disjunction $\lor$, strong next
$\Next$, and reflexive until $\UntilP$ operators.
It is known~\cite{Sis,KupVar99,CamEtAl18a} that every safety formula
$\p$ \emph{expresses a safety property}, i.e.,
for every infinite trace $\Imf$ such that $\Imf \not \models \p$,
there exists $\Fmf \in \Pre(\Imf)$ so that, for all
$\Imf' \in \Ext(\Fmf)$, it holds that $\Imf' \not \models \p$.  We
call such a finite trace $\Fmf$ a \emph{bad prefix for $\p$}, and we
define $\BadPre(\p)$ as the set of bad prefixes for $\p$.  On the
other hand, every co-safety formula $\p$ \emph{expresses a co-safety
  property}, i.e., for every infinite trace $\Imf$ satisfying $\p$,
there is a \emph{good prefix for $\p$}, that is, a finite prefix
$\Fmf$ of $\Imf$ such that every infinite extension $\Imf'$ of $\Fmf$
satisfies $\p$.
%
Given the equivalence, $\Next \lnot \p \equiv_{i} \lnot \Next \p$,
holding on infinite traces, we have that the $\LTL$ formulas that are,
respectively, in the $\Release$- and $\Until$-fragments defined
in~Section~\ref{sec:fininf} express, respectively, safety and
co-safety properties.

In order to establish connections between safety properties and finite
traces semantics, we now require, following~\cite{BaiKat}, further
definitions and notation.
Let $\mathsf{N}^{0}_{\mathsf{P}}$ be the subset of $\NPr$ containing
$0$-ary predicates, i.e., propositional letters.
Given a suborder
$\Tmf$ of $(\mathbb{N}, <)$ of the form $[0, \infty)$ or $[0, l]$,
with $l \in \mathbb{N}$,
a \emph{trace} is now viewed simply as a sequence
$\Mmf = (\Mmc_{n})_{n \in \Tmf}$ with
$\Mmc_{n} \in 2^{\mathsf{N}^{0}_{\mathsf{P}}}$.
The notion of an $\LTL$ formula $\p$ being satisfied in a trace
$\Mmf$, $\Mmf \models \p$, is given similarly as above
(cf. Section~\ref{sec:semantics}), and we let $\Traceinf(\p)$,
respectively, $\Tracefin(\p)$, be the set of infinite, respectively,
finite, traces satisfying $\p$.
The sets of prefixes of a trace $\Mmf$ and of extensions of a finite trace $\Fmf$
are defined as in Section~\ref{sec:semantics}.
Moreover, given a set of infinite traces, $\Cmc$, and of finite
traces, $\Dmc$, we define
\[
\Pre(\Cmc) = \bigcup_{\Imf \in \Cmc} \Pre(\Imf),
\qquad
\Ext(\Dmc) = \bigcup_{\Fmf \in \Dmc} \Ext(\Fmf).
\]
For an $\LTL$ formula $\p$, we may write $\Pre(\p)$ and $\Ext(\p)$ in
place of, respectively, $\Pre(\Traceinf(\p))$ and
$\Ext(\Tracefin(\p))$.
%
%
Moreover, given a set of infinite traces $\Cmc$, we define the
\emph{closure of $\Cmc$} as
$\Clo(\Cmc) = \{ \Imf \mid \Pre(\Imf) \subseteq \Pre(\Cmc) \}$.
Consider, for instance,
the set of infinite traces
$\Cmc = \bigcup_{i \in [0, \infty)} \{ \Imf^{i} \}$, where
$\Imf^{i} = ( \Imc^{i}_{n})_{n \in [0, \infty)}$ is such that
$\Imc^{i}_{n} = \{ P \}$, if $i = n$, and $\Imc^{i}_{n} = \emptyset$,
otherwise.  We have that the infinite trace
$\Imf = (\Imc_{n})_{n \in [0, \infty)}$ defined so that
$\Imc_{n} = \emptyset$, for every $i \in [0, \infty)$, is such that
$\Imf \not \in \Cmc$, whereas $\Imf \in \Clo(\Cmc)$.
The following preliminary lemma adapts to our setting a result
presented in~\cite[Lemma 3.25]{BaiKat}.
\begin{lemma}
\label{lemma:badprefvar}
Let $\p, \chi$ be
$\LTL$
formulas, with
$\p \in \LTL(\finit{\Rightarrow}{\exists}) \cap \LTL(\infinit{\Rightarrow}{\forall})$
and $\chi$ expressing a safety property. It holds that:
\[
	\p \modelsinf \chi \quad \text{iff} \quad \Tracefin(\p) \cap \BadPre(\chi) = \emptyset.
\]
\end{lemma}
\begin{proof}
$(\Rightarrow)$
By contraposition, assume there exists $\Fmf \in \Tracefin(\p) \cap \BadPre(\chi)$.
We have that $\Fmf \models \p$, and
since
$\p$ is $\finit{\Rightarrow}{\exists}$, there exists $\Imf \in \Ext(\Fmf)$ such that $\Imf \models \p$.
Given that $\Fmf \in \BadPre(\chi)$ and that $\Fmf \in \Pre(\Imf)$, we have in particular that $\Imf \not \models \chi$.
Therefore, $\p \not \modelsinf \chi$.

$(\Leftarrow)$
By contraposition, assume that $\p \not \modelsinf \chi$, i.e., there exists an infinite trace $\Imf$ such that $\Imf \models \p$ and $\Imf \not \models \chi$.
Since $\chi$ expresses a safety property, there exists $\hat{\Fmf} \in \Pre(\Imf)$ such that,
for all $\Imf' \in \Ext(\hat{\Fmf})$,
$\Imf' \not \models \chi$.
Thus,
$\hat{\Fmf} \in \BadPre(\chi)$.
Moreover, since $\p$ is $\infinit{\Rightarrow}{\forall}$, we have that every $\Fmf \in \Pre(\Imf)$ is such that $\Fmf \models \p$. Hence, $\hat{\Fmf} \in \Tracefin(\p)$, and so $\Tracefin(\p) \cap \BadPre(\chi) \neq \emptyset$.
\end{proof}

Using the previous result, we can show the following characterisation
of two $\LTL$ formulas being equivalent on finite traces in terms of
safety properties, under the assumption that they satisfy the
properties $\finit{\Rightarrow}{\exists}$ and
$\infinit{\Rightarrow}{\forall}$, as it is the case for $\LTL$
$\ReleaseP$-formulas (cf. Section~\ref{sec:fininf}).  This
characterisation mirrors the one given in~\cite[Corollary
3.29]{BaiKat}, but it is presented here in terms of finite traces and
$\LTL$ formulas, as opposed to the one in~\cite{BaiKat} based on sets
of traces and transition systems.
\begin{proposition}
\label{prop:finequivsafe}
For every
$\p, \psi \in \LTL(\finit{\Rightarrow}{\exists}) \cap \LTL(\infinit{\Rightarrow}{\forall})$,
the following are equivalent:
	\begin{itemize}
		\item[$(i)$] $\equivfin{\p}{\psi}$;
		\item[$(ii)$] for every $\LTL$ formula $\chi$ expressing a safety property,
							$\p \modelsinf \chi$ iff $\psi \modelsinf \chi$.
	\end{itemize}
\end{proposition}
\begin{proof}
We are going to show that the following statements are equivalent:
	\begin{itemize}
		\item[$(i')$] $\Tracefin(\p) \subseteq \Tracefin(\psi)$;
		\item[$(ii')$] for every $\LTL$ formula $\chi$ expressing a safety property,
							$\psi \modelsinf \chi$ implies $\p \modelsinf \chi$.
	\end{itemize}
Then, by swapping $\p$ and $\psi$, the required result will follow.

$(i') \Rightarrow (ii')$
Assume that $\Tracefin(\p) \subseteq \Tracefin(\psi)$ and that, for every $\LTL$ formula $\chi$ expressing a safety property, $\psi \modelsinf \chi$. By Lemma~\ref{lemma:badprefvar}, this is equivalent to $\Tracefin(\psi) \cap \BadPre(\chi) = \emptyset$. Since $\Tracefin(\p) \subseteq \Tracefin(\psi)$, we have $\Tracefin(\p) \cap \BadPre(\chi) = \emptyset$, which means, again by Lemma~\ref{lemma:badprefvar}, that $\p \modelsinf \chi$.

$(ii') \Rightarrow (i')$ Assume that, for every $\LTL$ formula $\chi$
expressing a safety property, $\psi \modelsinf \chi$ implies
$\p \modelsinf \chi$ and consider the set $\Clo(\Traceinf(\psi))$.  It
is known~\cite{BaiKat} that the closure $\Clo(\Cmc)$ of a set of
infinite traces $\Cmc$ corresponds to the \emph{topological closure}
of $\Cmc$ in the topological space
$( ( 2^{\mathsf{N}^{0}_{\mathsf{P}}})^\omega, \tau)$,
  where $\tau$ is the topology
induced by the metric
$\delta \colon
(2^{\mathsf{N}^{0}_{\mathsf{P}}})^{\omega}
\to
\mathbb{R}_{0}$ (where $\mathbb{R}_{0}$ is the set of non-negative
real numbers), defined as follows: $\delta(\Imf, \Imf) = 0$;
$\delta(\Imf, \Imf') = 2^{-n}$, with
$n = \min\{ m \in \mathbb{N} \mid \Imc_{m} \neq \Imc'_{m} \}$, for
$\Imf \neq \Imf'$.  Moreover, as shown in~\cite[Corollary
14]{MarEtAl14}, the following holds.

\begin{claim}\label{claim:clo}
For every $\LTL$ formula $\p$, there is an $\LTL$ formula $\p'$ such that $\Clo(\Traceinf(\p)) = \Traceinf(\p')$.
\end{claim}
%
%
We now require the following claim.

\begin{claim}
\label{cla:closafe}
For every $\LTL$ formula $\chi$, $\chi$ expresses a safety property
iff $\Clo(\Traceinf(\chi)) = \Traceinf(\chi)$.
\end{claim}
\begin{proof}[Proof of Claim~\ref{cla:closafe}]
  Following~\cite{BaiKat}, we call \emph{safety property} a set of
  infinite traces $\Smc$ such that, if $\Imf \not \in \Smc$, then
  there is $\Fmf \in \Pre(\Imf)$ such that, for every
  $\Imf' \in \Ext(\Fmf)$, it holds that $\Imf' \not \in \Smc$.
  From the definitions, we have that $\chi$ expresses a safety
  property iff $\Traceinf(\chi)$ is a safety property.  It is known
  that safety properties coincide with closed sets in the topological
  space $((2^{\mathsf{N}^{0}_{\mathsf{P}}})^{\omega}, \tau)$
  defined above~\cite{BaiKat}.
  The claim then follows from the fact that a set is closed in a
  topological space iff it is equal to its topological closure, and
  the fact that the topological closure of a set $\Smc$ in
  $((2^{\mathsf{N}^{0}_{\mathsf{P}}})^{\omega}, \tau)$
  coincides with $\Clo(\Smc)$.
\end{proof}
We can now finish the proof of the $(ii') \Rightarrow (i')$ direction.
From Claim~\ref{claim:clo} we have that there exists an $\LTL$ formula
$\psi'$ such that $\Clo(\Traceinf(\psi)) = \Traceinf(\psi')$.  Since
$\Clo(\cdot)$ is a closure operator, we have that
\[
\Clo(\Traceinf(\psi')) = \Clo(\Clo(\Traceinf(\psi))) = \Clo(\Traceinf(\psi)) = \Traceinf(\psi').
\]
Thus, by Claim~\ref{cla:closafe}, we have that $\psi'$ expresses a
safety property.  Moreover, the fact that $\Clo(\cdot)$ is a closure
operator implies also that
$\Traceinf(\psi) \subseteq \Clo(\Traceinf(\psi)) = \Traceinf(\psi')$,
thus $\psi \modelsinf \psi'$.  From $(ii')$, we obtain that
$\p \modelsinf \psi'$. Thus,
$\Traceinf(\p) \subseteq \Traceinf(\psi')$ and, by
Claim~\ref{claim:clo},
$\Traceinf(\p) \subseteq \Clo(\Traceinf(\psi))$.  Since $\Pre(\cdot)$
is monotonic, i.e., for every set of infinite traces $\Cmc, \Cmc'$
such that $\Cmc \subseteq \Cmc'$, we have
$\Pre(\Cmc) \subseteq \Pre(\Cmc')$, it follows that
$\Pre(\Traceinf(\p)) \subseteq \Pre(\Clo(\Traceinf(\psi)))$.
Moreover, since $\p$ is $\finit{\Rightarrow}{\exists}$, it holds that
$\Tracefin(\p) \subseteq \Pre(\Traceinf(\p))$, and since $\psi$ is
$\infinit{\Rightarrow}{\forall}$, we have
$\Pre(\Traceinf(\psi)) \subseteq \Tracefin(\psi)$.  Finally, it can be
seen that $\Pre(\Clo(\Traceinf(\psi))) = \Pre(\Traceinf(\psi))$.
Hence, we obtain:
\[
  \Tracefin(\p) \subseteq \Pre(\Traceinf(\p)) \subseteq
  \Pre(\Clo(\Traceinf(\psi))) = \Pre(\Traceinf(\psi)) \subseteq
  \Tracefin(\psi).  \qedhere
\]
\end{proof}

\subsubsection{Runtime verification maxims}

We recall that in runtime verification the task is to evaluate a
property with respect to the current history (which is finite at each
given instant) of a dynamic system, and to check whether this property
is satisfied in all its possible future evolutions
~\cite{BauEtAl,BaaLip,DegEtAl2}.
%
%
Here we discuss the relationship between our semantic conditions and
the \emph{maxims} for runtime verification in (variants of) $\LTL$
introduced by~\citeauthor{BauEtAl}~\citeyear{BauEtAl}, which relate
finite trace semantics to the infinite case.  Although the authors
consider also semantics for $\LTL$ that allow for more than two
truth-values, in this section we will restrict our attention to $\LTL$
interpreted on finite traces only.
\citeauthor{BauEtAl} suggest that any $\LTL$ semantics to be used in
runtime verification should satisfy, for every $\LTL$ formula $\p$,
the 
maxims
of \emph{impartiality} and \emph{anticipation},
defined as follows.
\begin{description}
\item[Impartiality] For every finite trace $\Fmf$,
  \begin{center}
    $\Fmf \models \p \Rightarrow \forall \Imf \in \Ext(\Fmf) . \Imf
    \models \p$ \quad and \quad
    $\Fmf \not \models \p \Rightarrow \forall \Imf \in \Ext(\Fmf)
    . \Imf \not \models \p$.
  \end{center}
\item[Anticipation] For every finite trace $\Fmf$,
  \begin{center}
    $\Fmf \models \p \Leftarrow \forall \Imf \in \Ext(\Fmf) . \Imf
    \models \p$ \quad and \quad
    $\Fmf \not \models \p \Leftarrow \forall \Imf \in \Ext(\Fmf)
    . \Imf \not \models \p$.
  \end{center}
\end{description}
%

It can be easily seen that $\LTL$ on finite traces does not satisfy,
for every $\LTL$ formula $\p$, impartiality and anticipation.  An
example of a formula that does not satisfy impartiality is $\B^{+} P$,
whereas $\D \top$ violates anticipation.
However, the properties of impartiality and anticipation can be used
to define the corresponding sets of $\LTL$ formulas that satisfy them.
Indeed, we have that impartiality is captured by
$\LTL(\finit{\Rightarrow}{\forall}) \cap
\LTL(\finit{\Leftarrow}{\exists})$, while anticipation corresponds to
$\LTL(\finit{\Leftarrow}{\forall}) \cap
\LTL(\finit{\Rightarrow}{\exists})$.  Therefore, any set of $\LTL$
formulas satisfying both impartiality and anticipation
is included in the intersection
$\LTL(\finit{}{\forall}) \cap \LTL(\finit{}{\exists})$.
Concerning the possibility to
syntactically characterise these formulas, we have that, due to
Lemmas~\ref{thm:diamond} and~\ref{thm:box}, impartiality and
anticipation are not guaranteed to be preserved for
$\UntilP$-
or
$\ReleaseP$-formulas.



\section{Conclusion}
\label{sec:conc}

We investigated first-order temporal logic
on finite traces, by comparing its semantics with the usual one based
on infinite traces, and by studying the complexity of formula
satisfiability in some of its decidable fragments.

In an effort to systematically clarify the correlations between finite
vs. infinite reasoning we introduced various semantic conditions that
allow to formally specify when it is possible to blur the distinction
between finite and infinite traces.
%
Grammars for $\QTLfr{\Until}{}{}$ formulas satisfying some of these
conditions have been provided as well.  In particular, we have shown
that for $\UntilP$- and $\ReleaseP$-formulas, equivalence over finite
and infinite traces coincide.  Moreover, we have shown that, for
the class of $\UntilP\ReleaseP$-formulas,
satisfiability is preserved from finite to infinite traces.

Concerning the complexity of the satisfiability problem in decidable
fragments on finite traces, we have shown that the constant-free
one-variable monadic fragment $\QTLfr{\Until}{1,mo}{\not c}$, the
one-variable fragment $\QTLfr{\Until}{1}{}$, the monadic monodic
fragment $\QTLfr{\Until}{mo}{\monodic}$, and the the two-variable
monodic fragment $\QTLfr{\Until}{2}{\monodic}$, while being
$\ExpSpace$-complete over arbitrary finite traces, lower down to
$\NExpTime$-complete when
interpreted on traces with at most $k$ time points.  Similar results
have been shown here for $\TALCf$, a temporal extension of the
description logic $\ALC$, interpreted on finite or $k$-bounded
traces. Moreover, we proved that $\TALCfk$ restricted to global CIs is
$\ExpTime$-complete on traces with at most $k$ time points.

Finally, we have lifted results related to the notion of insensitivity
to infiniteness~\cite{DegEtAl}, introduced in the planning context, to
our first-order setting.  Moreover, we have analysed the connections
between notions from the verification literature (in particular,
safety~\cite{Sis}, as well as the runtime verification maxims of
impartiality and anticipation~\cite{BauEtAl}), and our framework of
semantic conditions relating reasoning over finite and infinite
traces.

As future work,
we are interested in strengthening the results obtained in
Section~\ref{sec:dist}, so to obtain semantic and syntactic conditions
that are both necessary and sufficient (as opposed to sufficient only)
to characterise equivalences on finite and infinite traces.  We
conjecture also that for $\UntilP\forall$- and
$\ReleaseP\exists$-formulas, i.e., $\QTLfr{\Until}{}{}$ formulas in
negation normal form involving only one kind of reflexive temporal
operator (either $\UntilP$ or $\ReleaseP$, respectively), the
equivalences on finite traces coincide with the equivalences on
infinite traces.

Moreover, we plan to to study the axiomatisability of
fragments of
first-order temporal logic on finite traces, and to apply the semantic
conditions introduced in this work to the analysis of monitoring
functions for runtime verification~\cite{BauEtAl,BaaLip,DegEtAl2}.  It
would also be interesting to determine the precise complexity of the
satisfiability problem in $\TALCf$ on finite traces with just global
CIs, as well as in those DLs from the temporal $\DLite$ family for
which this problem remains open~\cite{ArtEtAl19b}.


\begin{acks}
Ozaki is supported by the Research Council of Norway, project number 316022.
\end{acks}


\bibliographystyle{ACM-Reference-Format}
\bibliography{tocl-bibliography}

\newpage

\end{document}